\pdfoutput=1
\documentclass[a4paper,11pt]{article}
\usepackage{amsthm,amsfonts,amssymb,amsmath,wasysym,stmaryrd,textcomp}
\usepackage[section]{placeins}
\usepackage[pdftex]{graphicx}
  \DeclareGraphicsExtensions{.pdf}
\usepackage{color}
\usepackage{booktabs}
\newtheorem{theorem}{Theorem}[section]
\newtheorem{lemma}[theorem]{Lemma}
\newtheorem{proposition}[theorem]{Proposition}

\newtheorem{definition}[theorem]{Definition}

\newtheorem{remark}[theorem]{Remark}
\newtheorem{notation}{Notation}
\begin{document}
\title{A graph-theoretical approach for the computation of connected iso-surfaces based on
volumetric data}
\author{Abdulaziz Ali and Dieter Bothe\\
\mbox{}\\
Mathematical Modeling and Analysis\\
Center of Smart Interfaces\\
TU Darmstadt\\
Alarich-Weiss-Str.\ 10\\
64287 Darmstadt\\
Germany}
\maketitle
\newpage
\tableofcontents
\newpage
\begin{abstract}
The existing combinatorial methods for iso-surface computation are efficient for pure visualization purposes,
but it is known that the resulting iso-surfaces can have holes, and topological problems like missing or wrong connectivity
can appear. To avoid such problems, we introduce a graph-theoretical method for the computation of iso-surfaces on cuboid
meshes in $\mathbb{R}^3$. The method for the generation of iso-surfaces employs labeled cuboid graphs
$G(V,E,\mathcal{F})$ such that $V$ is the set of vertices of a cuboid $C\subset\mathbb{R}^3$, $E$ is the set of
edges of $C$ and $\mathcal{F}\,:\,V\rightarrow [0,1]$. The nodes of $G$ are weighted by the values of $\mathcal{F}$
which represents the volumetric information, e.g. from a Volume of Fluid method. Using a given iso-level
$c\in (0,1)$, we first obtain all iso-points, i.e. points where the value $c$ is attained by the edge-interpolated
$\mathcal{F}$-field. The iso-surface is then built from iso-elements which are composed of triangles and are
such that their polygonal boundary has only iso-points as vertices. All vertices lie on the faces of a single mesh cell.

We give a proof that the generated iso-surface is connected up to the boundary of the domain and it can be
decomposed into different oriented components. Two different components may have discrete points or line segments
in common. The graph-theoretical method for the computation of iso-surfaces developed in this paper enables to
recover local information of the iso-surface that can be used e.g. to compute discrete mean curvature and to
solve surface PDEs. Concerning the computational effort, the resulting
algorithm is as efficient as existing combinatorial methods.
\end{abstract}
\newpage
\section{Introduction}
An iso-surface is a level set of a continuous function whose domain is, in the considered case, $\mathbb{R}^3$.
Iso-surfaces are for example used to visualize scalar volume data processed in medicine, computational
fluid dynamics (CFD), geophysics, and chemistry. In medical imaging, by applying X-ray computed tomography
(CT) one obtains volume data which can be used to detect bones, tumors and cancer. In two-fluid systems,
the Volume of Fluid (VOF) method provides VOF-data which implicitly describes the interface between the fluids.

The Marching Cubes method~\cite{Lorensen87marchingcubes:} is a well known method for volume
visualization. It is a combinatorial and not a graph-theoretical method, being based on the tabulation
of 256 different configurations. This set of possible configurations is not complete and, hence,
resulting iso-surfaces can have holes. More generally, it is known that commonly employed iso-surface
algorithms can lead to surfaces with wrong connectivities and holes; see~\cite{Etiene:2012:TVI:2197070.2197097}
and further references therein. While this may not be an issue in a pure visualization context, it is a severe
problem if surface transport equations are to be solved, like surfactant transport on a fluid
surface~\cite{Alke_and_Bothe}. Other works like~\cite{Nielson:1991:ADR:949607.949621}
and~\cite{Chernyaev95marchingcubes} resolve the
ambiguity in Marching Cubes. In~\cite{conf/dgci/Lachaud96}, a topological approach for the
computation of iso-surfaces is given, some geometrical properties of the iso-surface are derived and an
algorithm for the computation of iso-surfaces is given. But there are still configurations which are not
covered by~\cite{Nielson:1991:ADR:949607.949621} as well as by~\cite{conf/dgci/Lachaud96}. For
example configurations in which an iso-surface only touches one or more vertices of a cuboid are not investigated
in~\cite{Chernyaev95marchingcubes},~\cite{conf/dgci/Lachaud96} and~\cite{Nielson:1991:ADR:949607.949621}.

The present work, we introduce a novel graph-theoretical method for the generation of iso-surfaces. We will show
that the generated iso-surface is connected up to the boundary of the domain and it can be decomposed into
different oriented components. Two different components may have discrete points or line segments as an
intersection. If two different components have common points or line segments, then each single component
can be identified and computed as if they are disjoint. This decomposition
of iso-surfaces into oriented components can be used e.g. to compute discrete mean curvature and to solve
surface PDEs for instance by means of a Finite Area method.

Our graph-theoretical approach employs labeled cuboid graphs. A labeled cuboid graph is denoted by
$G(V,E,\mathcal{F})$ such that $V$ is the set of vertices of a cuboid  $C\subset\mathbb{R}^3$, $E$
is the set of edges of $C$ and $\mathcal{F}\,:\,V\rightarrow [0,1]$. The nodes of $G$ are weighted
by the values of $\mathcal{F}$ and the weights lie in $[0,1]$. Using a given iso-level $c\in (0,1)$,
we interpolate on each edge of the graph to get all points, where the value $c$ is attained.
We call such a point on an edge of $G$, where the interpolated values equals $c$, an iso-point if one
of the edge end points has a label greater than $c$ and the other one less than or equal to $c$. From
this we get a piece of iso-surface whose boundary is a polygon such that its vertices
are iso-points and each of the edges lies on a face of $C$. We call such an iso-surface piece an iso-element, its
boundary an iso-path and each of the edges of the iso-path an iso-line. To compute the iso-element we use
a center point $P\in\mathbb{R}^3$ which is the arithmetic mean of the corresponding iso-points.
Then the iso-element is defined as the union of all triangles spanned by an iso-line edge and the vertex $P$.
Each of the sketches $(a)$ to $(d)$ of Figure~\ref{image_1} shows an iso-element of a labeled
cuboid graph. Sketch $(a)$ has a triangular iso-element. Sketches $(b)$ to $(d)$ of
Figure~\ref{image_1} have iso-elements such that the iso-paths may not lie on a common plane.

The iso-elements described above result from iso-paths which run on the faces of a single graph $G(V,E,\mathcal{F})$.
Besides this, iso-elements can also occur from two neighboring graphs in which case they lie on their common
face. This special case is of fundamental importance for connectivity of the iso-surface.

\begin{notation}
Let $G(V,E,\mathcal{F})$ be a labeled cuboid graph and $c\in (0,1)$ be an iso-level. Then we use the symbols
$\circ,\circleddash,\square,\bullet$ in sketches which illustrate $G(V,E,\mathcal{F})$ as well as subgraphs
of it. The symbols are used to characterize the labels of the nodes of $G$ and for interpolated values at points
that lie on edges of $G$. The symbol $\circleddash$ correspond to labels less than $c$. The symbol $\bullet$
correspond to labels greater than $c$ and the symbol $\square$ to labels equal to $c$. In addition, the symbol
$\circ$ correspond to labels less than or equal to $c$.
\label{int:not-1}
\end{notation}
\begin{figure}[!ht]
\includegraphics[width=0.8\linewidth]{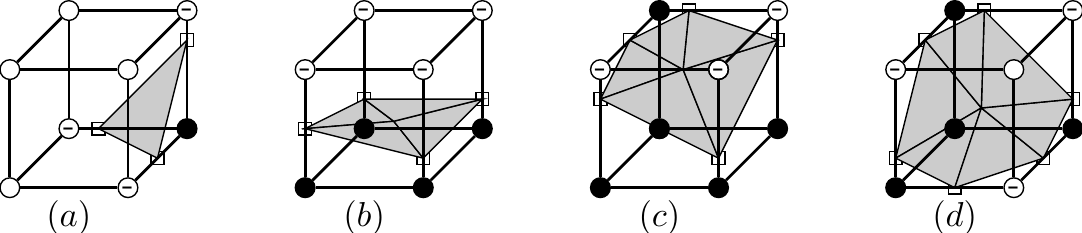}
\caption{The sketches $(a)$ to $(d)$ show iso-elements of labeled cuboid graphs. Each iso-element
is determined by its polygonal boundary, an iso-path.}
\label{image_1}
\end{figure}
\FloatBarrier

The iso-surface construction algorithms known so far do not yield in a simple way the local
information of the iso-surface such as neighborhood relations at common points or
edges of two or more iso-elements. Such topological information is required
to compute discrete mean curvature and to solve PDEs on an iso-surface. Hence, the development of a
new iso-surface computation method which provides a decomposition of the iso-surface into connected
components is required. We achieve this by introducing a purely graph-theoretical method for the
computation and decomposition of iso-surfaces. The resulting algorithm is very efficient.

The sketches in Figure~\ref{image_1} show labeled cuboid graphs having one iso-path.
But a labeled cuboid graph can have more than one iso-path and hence more than
one iso-element. The sketches in Figure~\ref{image_2} demonstrate that the number of
iso-paths depends on the labels of the graph. These are not meaningless, pathological cases, but they
occur naturally in dynamic processes with topological changes such as the crown splash of an impacting
droplet considered as an application in Section~5. Therefore, one of the main tasks is to introduce a
classification of labeled
cuboid graphs according to their subgraphs such that the number of iso-paths in a labeled cuboid graph
can be identified. The identification of the number of iso-paths and further analysis of the subgraphs of
a labeled cuboid graph is the main part of the present work.
\begin{figure}[!ht]
\includegraphics[width=0.8\linewidth]{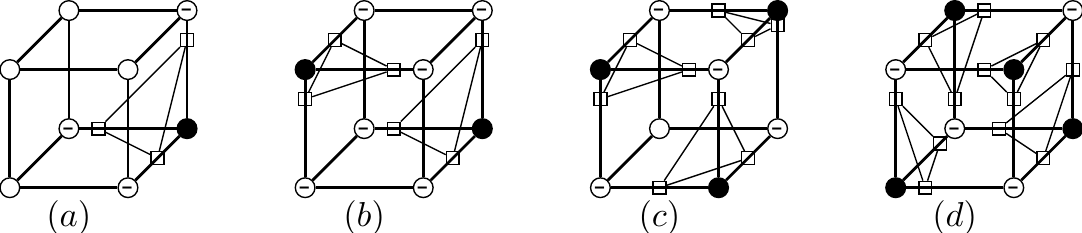}
\caption{Sketches $(a)$ to $(d)$ show labeled cuboid graphs with one to four distinct iso-paths.}
\label{image_2}
\end{figure}

To understand the topology of an iso-surface, we need to investigate pairs of labeled cuboid graphs
$G_1(V_1,E_1,\mathcal{F}_1)$ and $G_2(V_2,E_2,\mathcal{F}_2)$ for a common iso-level $c\in (0,1)$,
where the underlying cuboids $C_1$ and $C_2$, have a common face. Such graphs will be called
face-neighbored. If both labeled cuboid graphs have iso-paths with a common iso-line lying on the
common face of the cuboids, then they have iso-elements which have a common edge
(see Figure~\ref{image_3}). In case both end points of the common edge
are vertices of the common face then there is a possibility that more than two iso-paths can meet at
the common edge. Additionally, if the common iso-line of two iso-paths is a diagonal of a cuboid face,
then in each cuboid we can have a maximum of two iso-paths passing through the common edge; hence, it
can be a common edge for four different iso-paths. An edge of a cuboid can be a common edge for four
distinct cuboids and in each cuboid we can have a maximum of two iso-paths passing
through the common edge. Hence, there can be up to eight distinct iso-paths passing through the common edge.
All these cases have to be treated for a rigorous iso-surface computation. The sketches in Figure~\ref{image_3}
demonstrate two iso-paths of a pair of labeled cuboid graphs with a common edge.
\begin{figure}[!ht]
\includegraphics[width=0.6\linewidth]{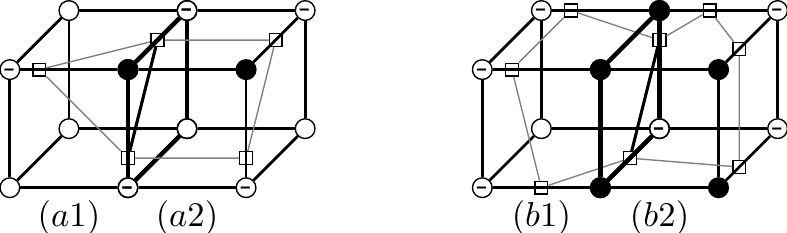}
\caption{Each pair of sketches $(a1), (a2)$ and $(b1), (b2)$ shows the iso-paths of
face-neighbored labeled cuboid graphs. The common face edges in both sketches is bold-framed.
Each pair of iso-paths in both sketches above has a common iso-line which is drawn in bold.}
\label{image_3}
\end{figure}
\FloatBarrier

An iso-point $P\in\mathbb{R}^3$ of an iso-path can be a common point of $4,5,6,7$ or $8$ iso-paths.
Let $P_i$ for $i=1,\ldots,n$ lie on the $i$-th iso-path $\omega_i$ such that $\overline{PP_i}$ is
an iso-line of $\omega_i$. Let $P_{c_i}$ be the center of $\omega_i$. Then we construct a
piece of iso-surface region $\Gamma_p$ using all iso-points $P_i$ and all iso-path centers
$P_{c_i}$ and $P$. Such a region which will be computed in Section 8 is required for computing
discrete mean curvature at the iso-point $P$ (see e.g. \cite{Meyer02Vismath}).

The iso-surface computation algorithm needs a partition $\mathcal{T}$ of a closed polygonal domain
$\overline{\Omega}\subset\mathbb{R}^3$ into cuboids such that two cuboids have either a common face or
a common vertex or a common edge or they are disjoint. Here, polygonal domain means domain with
polygonal boundary. Then, to obtain labeled cuboid graphs, the vertices of each of the cuboids in
the partition $\mathcal{T}$ of $\overline{\Omega}$ are labeled with real numbers in $[0,1]$. The labeled
cuboid graphs will be divided into two classes, depending on whether they have a single iso-path
or at least two iso-paths (cf.,Figure~\ref{image_2}).
Labeled cuboid graphs of the first class, having only a single iso-path, are called irreducible;
members of the other class are called reducible. This classification of labeled cuboid graphs is
fundamental for our algorithm.

The algorithm for iso-surface construction uses operations on labeled cuboid graphs such
that a reducible labeled cuboid graph $G(V,E,\mathcal{F})$ will be transformed to an irreducible one
$G'(V,E,\mathcal{F}')$. Each step in the transformation of $G$ from reducibility to irreducibility
gives an iso-element of $G$. An irreducible labeled cuboid graph finally has a single iso-element.
In addition, we can have an iso-path on single faces of a reducible labeled cuboid graph, on
which at least two iso-lines lie.

We get the iso-surfaces by collecting all iso-elements of each of the labeled cuboid graphs in
$\mathcal{T}$. The algorithm provides rich local information on the iso-surface like the number of
iso-elements with a common edge, or iso-elements with at least one edge of it being an edge of the
underlying cuboid net, or at least one edge of the iso-element being a face diagonal
of a cuboid. These local information is used for a simple identification of the connected components
of the iso-surface.

Another main result of the present work is that the number $n\in\mathbb{N}$ of iso-elements with a
common iso-line is an element of $\{2,4,6,8\}$. Hence, the number of connected components of the
iso-surface with a common iso-line may be an element of $\{1,2,3,4\}$. In particular, we will show
that triple lines at which three surfaces meet do not occur.

The same principle ideas works to discretization using tetrahedra \cite{second-article-ali-bothe}.
It is very likely that the same principle ideas can be adapted to other discretizations (more
general polyhedral cells) and also to space dimensions different from 3.

This algorithm for computation of iso-surfaces has computational complexity of order $O(N)$, where $N$ is
the number of cuboids in $\mathcal{T}$. Iso-surface computation algorithms based on
combinatorial approaches have as well a complexity of $O(N)$, but the algorithms do not give the full
topological information of the iso-surfaces. More severe, the iso-surface will in general be not complete.\\

\vspace{-0.3cm}
\noindent The paper is organized as follows:\\
\vspace{-0.3cm}

\noindent In Section 2 we introduce the main definitions and notations. We define
the concept of a labeled cuboid graph and its subgraphs. In addition, graph operations on
labeled cuboid graphs and on subgraphs are defined. Furthermore, we define iso-paths which
are the boundary of iso-elements that will be computed using labeled cuboid graphs. In Section 3,
equivalence classes of labeled cuboid graphs and its subgraphs are introduced. Rules for computation
of iso-paths of labeled cuboid graphs are given in Section 4. The classification of the different
types of labeled cuboid graphs for the computation of iso-paths and iso-path computation rules are
given in Section 5. Additionally, the algorithm for a complete iso-path computation of all labeled
cuboid graphs $\mathcal{T}$ is given. Furthermore, we include figures illustrating computed iso-surfaces
for snap shots of a simulated collision of two liquid droplets and
for a crown splash of an impacting droplet. The connectedness of the constructed iso-surface is proven
in Section 6. Finally, in Section 7 we give definitions of neighborhoods of iso-paths on which
the decomposition of iso-surfaces in connected components is based. Having this decomposition, we
show in Section 8 how to efficiently compute iso-surface normals with common orientation on a single component
and how to compute a surface region around an iso-point within the component from which discrete mean curvature
can be computed.

In this paper we introduce appropriate symbols and definitions which do not always follow the traditional way
but efficiently describe the method. Proofs are often given with the help of appropriate figures.
\newpage
\section{Labeled Cuboid Graphs}
In this paper a {\it labeled graph} $G(V,E,\mathcal{F})$ denotes a triple $(V,E,\mathcal{F})$ of
$V\subset\mathbb{R}^3$, $E\subset\mathbb{R}^3$ and $\mathcal{F}\,:\,V\longrightarrow [0,1]$ such
that $V$ is a set of vertices (nodes), $E$ is a set of edges with end points in $V$ and
$\mathcal{F}$ assigns real numbers from $[0,1]$ as weights (labels) to the nodes.

We call $C\subset\mathbb{R}^3$ a {\it cuboid} if $C=[a_1,b_1]\times[a_2,b_2]\times[a_3,b_3]$.
Let $G(V,E,\mathcal{F})$ be a labeled graph. We call $G(V,E,\mathcal{F})$ a {\it labeled cuboid graph},
if $V=\{P_1,\ldots,P_8\}$ and $E=\{e_1,\ldots,e_{12}\}$ are the set of vertices and edges of $C$,
respectively. In this case we call $C$ the {\it cuboid} of $G$. We say that a graph
$H(V_h,E_h,\mathcal{F}_h)$ is a {\it subgraph} of $G(V,E,\mathcal{F})$ if $V_h\subset V$,
$E_h\subset E$, and $\mathcal{F}_h$ is the restriction of $\mathcal{F}$ to $V_h$.

Let $G(V,E,\mathcal{F})$ be a graph and $P_1$, $P_2\in V$. We call $P_1$ {\it incident} to $P_2$
if $P_1$ and $P_2$ are the end points of an edge $e\in E$. In case $H(V_h,E_h,\mathcal{F}_h)$ is
a subgraph of $G(V,E,\mathcal{F})$ we write "$P_1$ {\it is incident to} $P_2$ {\it in} $H$" to
express the fact that $P_1$ and $P_2$ are the end points of an edge $e\in E_h$.

In the present paper, we only consider labeled graphs $G(V,E,\mathcal{F})$ having the following
connectivity property: For any node $P_1\in V$ there exists another node $P_2\in V$ such
that $P_1$ is incident to $P_2$. \\

\noindent{\bf Note: }If suitable, we use the abbreviation $G$ for a labeled cuboid graph
$G(V,E,\mathcal{F})$. Additionally, we sometimes abbreviate a subgraph $H(V_h,E_h,\mathcal{F}_h)$
of $G$ by $H$ and also use the notation $G_i(V_i,E_i,\mathcal{F}_i)$, abbreviated as $G_i$, for
labeled cuboid graphs $G(V_i,E_i,\mathcal{F}_i)$, where $i\in\mathbb{N}$. Analogous abbreviations
will be used for subgraphs of $G_i$. Moreover, we use the shorthand notation $G'$ for the labeled
cuboid graph $G(V,E,\mathcal{F}')$; analogously for subgraphs of it.

\subsection{Notations}
Below we give some notations which will be used in this work. Recall from Notation~\ref{int:not-1}
the symbols $\circ,\circleddash,\square,\bullet$ which will be used to indicate the weight of the
nodes of a labeled cuboid graph, its subgraphs and iso-points.\\

\noindent{\bf Partition of a domain into cuboids: }Let $\Omega\subset\mathbb{R}^3$ be a polygonal domain.
Here, polygonal domain means a domain with polygonal boundary. We denote by $\mathcal{T}$ the partition
of $\overline{\Omega}$ into cuboids such that two cuboids have either a common face or a common
vertex or a common edge or they are disjoint. We call such a partition $\mathcal{T}$ of $\overline{\Omega}$
{\it partition of} $\overline{\Omega}$ {\it into cuboids}.\\

\noindent{\bf Cuboid grid: }Let $\Omega\subset\mathbb{R}^3$ be a polygonal domain and let $\mathcal{T}$
be the partition of $\overline{\Omega}$ into cuboids. Then we call the vertices of all cuboids in
$\mathcal{T}$ a cuboid grid.\\

\noindent{\bf Parallel faces: } We say that two faces $F_1$ and $F_2$ of a cuboid $C$ are {\it parallel} if
both $F_1$ and $F_2$ have no common nodes; this is abbreviated by $F_1\parallel F_2$. In case $F_1$ and $F_2$
have common nodes we say that both faces are not parallel, symbolized as $F_1 \nparallel F_2$.\\

\noindent{\bf Face subgraph (a face, for short): }We say that a graph $H(V_h,E_h,\mathcal{F}_h)$ is a
{\it face subgraph} of a labeled cuboid graph $G(V,E,\mathcal{F})$, if $H$ contains a single face of $G$ and $V_h$
is the set of nodes of the face in $H$.\\

\noindent{\bf Edge subgraph (an edge, for short) or edge: }We say that a graph $H(V_h,E_h,\mathcal{F}_h)$ is an
{\it edge subgraph} of a labeled cuboid graph $G(V,E,\mathcal{F})$, if $H$ contains a single edge of $G$ and $V_h$
is the set of nodes of the edge in $H$.\\

\noindent{\bf Face and Edge neighbors: }We say that two labeled cuboid graphs $G_1(V_1,E_1,\mathcal{F}_1)$ and
$G_2(V_2,E_2,\mathcal{F}_2)$ are {\it face-neighbors} if, given $C_1$ and $C_2$ as the cuboids of $G_1$ and
$G_2$, respectively, the intersection $C_1\cap C_2$ is a face of both cuboids. We say that two cuboid graphs
$G_1$ and $G_2$ are {\it edge-neighbors} if the intersection $C_1\cap C_2$ is an edge of both cuboids.

\subsection{Interpolations and Iso-paths}
Suppose we have a labeled cuboid graph $G(V,E,\mathcal{F})$ and an iso-level $c\in (0,1)$. Then we
interpolate linearly between the end points of $e\in E$ and their weights to get a possible point on $e$
with the value $c$. Then, by connecting two distinct interpolated points of $G$ that lie on the same face
of $G$, we get a line. If, by joining each pair of these points that lie on the same face of $G$, we obtain
a closed path that does not cross itself, we call this path a simple closed path. Any such closed path
is the boundary of a continuous 2-dimensional manifold in $\mathbb{R}^3$, which will be defined later.

In this section we give notations and definitions which will be used throughout the text.
\begin{definition}(Iso-point). Let $G(V,E,\mathcal{F})$ be a labeled cuboid graph, $\xi=\cup_{e\in E}e$
the union of all edges of $G$ and $c\in (0,1)$ be an iso-level. We define a function
$f\,:\,\xi\longrightarrow [0,1]$ by piecewise definition:  on $e\in E$, let $f$ be defined by
\begin{equation}
f(x)=f_0+(f_1-f_0)\frac{(x-x_0)}{(x_1-x_0)}\qquad \mbox{ for } x\in e,
\label{eq:labeled-graph-11}
\end{equation}
where $x_0$ and $x_1$ are the nodes of $e$ and $f_0=\mathcal{F}(x_0)$, $f_1=\mathcal{F}(x_1)$.
For $x\in e$ we call the value $f(x)$ the f-value of $x$. If $f(x)=c$ such that $f_0\leq c<f_1$ or
$f_1\leq c<f_0$ then we call $x\in e$ an {\it iso-point} of $G$. For an iso-point $x\in e$ we have
\begin{equation}
x(c) = x_0+(x_1-x_0)\frac{c-f_0}{f_1-f_0}.
\label{eq:labeled-graph-12}
\end{equation}
\label{def:labeled-graph-5}
\end{definition}

\noindent{\bf Iso-nodes, disperse nodes and continuous nodes: }Let $G(V,E,\mathcal{F})$ be a labeled cuboid
graph and $c\in (0,1)$ an iso-level. Then we call all nodes of $G$ with label greater than $c$ {\it disperse nodes},
otherwise {\it continuous nodes}. All nodes of $G$ which are iso-points are also called {\it iso-nodes}. We denote
by $D(G)$ the total number of disperse nodes of $G$. \\

\noindent{\bf disperse/continuous graph, face and edge: }Let $G(V,E,\mathcal{F})$ be a labeled cuboid graph and
$c\in (0,1)$ an iso-level. We call $G$ a {\it disperse graph} if all nodes of $G$ are disperse nodes and in case
all nodes of $G$ are continuous we call $G$ a {\it continuous graph}. A disperse graph $G$ is symbolized as
$G_{disp}$ and a continuous graph $G$ is symbolized as $G_{cont}$. A face subgraph $H(V_h,E_h,\mathcal{F}_h)$ of
$G$ is a {\it disperse face} if all nodes of $H$ are disperse; we call it a {\it continuous face} if all nodes of
$H$ are continuous. We call an edge subgraph $H(V_h,E_h,\mathcal{F}_h)$ of $G$ a {\it disperse edge} if all nodes of
$H$ are disperse; we call it a {\it continuous edge} if all nodes of $H$ are continuous.\\

\noindent{\bf L-face and Trivial L-face: }Let $G(V,E,\mathcal{F})$ be a labeled cuboid graph and
$c\in (0,1)$ an iso-level. We call a face subgraph $H(V_h,E_h,\mathcal{F}_h)$ of $G$
an {\it L-face} if there exist two disperse and two continuous nodes of $H$ such that
the disperse nodes are not incident in $H$. An L-face $H(V_h,E_h,\mathcal{F}_h)$
is called a {\it trivial L-face}, if both continuous nodes are iso-nodes. All other L-faces are called
{\it non-trivial L-faces}. We denote by $L(G)$ the set of all L-faces of $G$. Figure~\ref{image_4} shows
all possible L-faces of $G$.
\begin{figure}[!ht]
\includegraphics[width=0.7\linewidth]{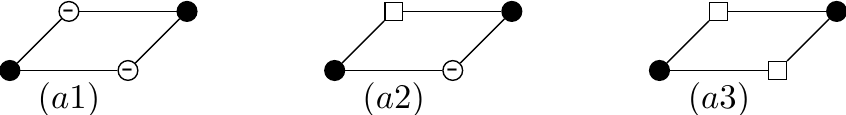}
\caption{Sketches $(a1)$ and $(a2)$ represent non-trivial L-faces and sketch $(a3)$ represents a trivial L-face.}
\label{image_4}
\end{figure}
\FloatBarrier

\noindent{\bf Singular/regular face: }Let $G(V,E,\mathcal{F})$ be a labeled cuboid graph and
$c\in (0,1)$ an iso-level. A face subgraph $H(V_h,E_h,\mathcal{F}_h)$ of
$G$ is called a {\it singular face} if three nodes of $H$ are disperse and the other node is
an iso-node. We say that a face subgraph $H(V_h,E_h,\mathcal{F}_h)$ of $G$ is a
{\it regular face} if  $H$ is not an L-face, not a disperse face, not a continuous face and not a
singular face. Figure~\ref{image_5} shows all possible regular faces and a singular face in
a labeled cuboid graph $G$.
\begin{figure}[!ht]
\includegraphics[width=0.9\linewidth]{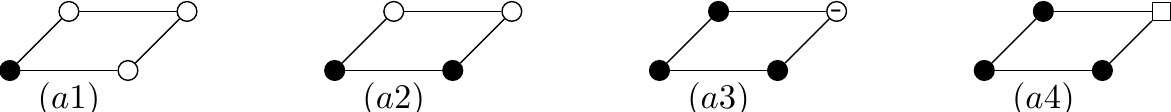}
\caption{Sketches $(a1)$, $(a2)$ and $(a3)$ represent regular faces and sketch $(a4)$ represents a singular face.}
\label{image_5}
\end{figure}
\FloatBarrier

\subsubsection{Iso-path and Iso-Elements}
The definitions of an iso-element and an iso-path are intertwined to each other. An iso-element is
a 2-dimensional manifold in $\mathbb{R}^3$ composed of flat triangular patches and computed
by using the iso-points of a labeled cuboid graph for a given iso-level $c\in (0,1)$. An iso-path is
the boundary of an iso-element.\\

\noindent{\bf Cyclically ordered points: }Let $P_1,\ldots,P_n$ be distinct points in $\mathbb{R}^3$, where
$n\geq 3$. Suppose that $J=\cup_{i=1}^ne_i$ with $e_i=\overline{P_iP_{i+1}}$ for $i=1,\ldots,n-1$ and
$e_n=\overline{P_nP_1}$ is a simply closed path. Then we call the points
$P_1,\ldots,P_n$ {\it cyclically ordered} and denote the simply closed path $J$ by $[P_1,\ldots,P_n]$.

We denote by $P_c:=\frac{1}{n}\sum_{i=1}^nP_i$ the center of $\{P_1,\ldots,P_n\}$ and define a surface with
boundary $[P_1,\ldots,P_n]$ by
\begin{equation}
[P_1,\ldots,P_n|P_c]:=\cup_{i=1}^{n-1}Tri(P_c,P_i,P_{i+1})\cup Tri(P_c,P_1,P_n),
\label{eq:labeled-graph-13}
\end{equation}
where $Tri(A,B,C)$ denotes the filled triangle spanned by the points $A,B,C$ of
$\mathbb{R}^3$. Note that the patches
which form such a surface do not lie in a common plane, in general, but for $n=3$ we have
\[
[P_1,P_2,P_3|P_c]=Tri(P_1,P_2,P_3).
\]
The surface $[P_1,\ldots,P_n|P_c]$ is an oriented and connected manifold in $\mathbb{R}^3$.\\

\noindent{\bf Iso-line: }Let $G(V,E,\mathcal{F})$ be a labeled cuboid graph and $c\in (0,1)$ be a
given iso-level. We say that a line segment $l$ is an {\it iso-line} of $G$ if its two end points
$P_1$ and $P_2$ are both iso-points, lying on the same face of the cuboid of $G$. We then call $P_1$
incident to $P_2$. This defines incidence between iso-points.

\begin{definition}(Iso-path and iso-element). Let $G(V,E,\mathcal{F})$ be a labeled
cuboid graph and let $c\in (0,1)$ be an iso-level. Let $G$ have $n\geq 3$ iso-points
$Q_1,\ldots,Q_n$. Let there be a subset $\{P_1,\ldots,P_m\}\subset\{Q_1,\ldots,Q_n\}$ with $3\leq m\leq n$
such that the set of line segments $\{\overline{P_1P_2},\ldots,\overline{P_{m-1}P_m}, \overline{P_mP_1}\}$
is a subset of the set of iso-lines of $G$. Furthermore, let the iso-points $P_1,\ldots,P_m$
define a simply closed path $[P_1,\ldots,P_m]$ with $P_c$ being its center. Then we call $[P_1,\ldots,P_m]$
an inner iso-path of $G$ if $[P_1,\ldots,P_m]$ does not lie on a single non-trivial L-face of $G$.

We call $[P_1,\ldots,P_m]$ an outer iso-path of $G$ if it lies on a single non-trivial L-face
$H(V_h,E_h,\mathcal{F}_h)$ of $G$ and satisfies
\[
[P_1,\ldots,P_m]=\bigcup_{k\in K}J_k\cap M_h,
\]
where $\{J_k\}_{k\in K}$ are the inner iso-paths of $G$ and its face-neighbor $G'(V',E',\mathcal{F}')$,
where $V\cap V'=V_h$ and $M_h$ is the common face of the cuboids of $G$ and $G'$.

We call $[P_1,\ldots,P_m]$ an iso-path if it is an inner or an outer iso-path, and we then call
$[P_1,\ldots,P_m|P_c]$ an iso-element of $G$.
\label{def:labeled-graph-6}
\end{definition}

\noindent{\bf Corresponding iso-element, iso-path and iso-line: }Let $G(V,E,\mathcal{F})$ be a labeled
cuboid graph and let $c\in (0,1)$ be an iso-level. Let us denote by $\omega$ one of the iso-paths
of $G$ and by $Z$ the iso-element which is bounded by all of the edges of $\omega$
as given by Definition~\ref{def:labeled-graph-6}. Then we say $\omega$ {\it corresponds to} $Z$ or $Z$
{\it corresponds to} $\omega$. Furthermore, if $l$ is an iso-line of $G$ which lies in $\omega$ then we say
that $l$ {\it corresponds to} $\omega$.\\

\noindent{\bf Iso-line neighbor: }Let $G_1(V_1,E_1,\mathcal{F}_1)$ and $G_2(V_2,E_2,\mathcal{F}_2)$
be labeled cuboid graphs and $c\in (0,1)$ a given iso-level. Assume $G_1=G_2$ or $G_1$ and $G_2$ are
face or edge-neighbors according to the given context. Let $\omega_1$ and $\omega_2$ be distinct
iso-paths of $G_1$ and $G_2$, respectively. Furthermore, let $l_1$ and $l_2$ be iso-lines
corresponding to $\omega_1$ and $\omega_2$, respectively such that $l_1$ and $l_2$ have the same end points.
Then we call the iso-lines $l_1$ and $l_2$ {\it neighbors}.\\

\noindent{\bf Connectedness of iso-path: }Let $G(V,E,\mathcal{F})$ be a labeled cuboid graph and
let $c\in (0,1)$ be an iso-level. Let $\omega$ be an iso-path of $G$. By definition of an iso-path,
$\omega$ contains $3\leq n\leq 6$ iso-lines $l_1,\ldots,l_n$ of $G$. We say that the iso-path $\omega$
of $G$ is {\it connected} if every $l_i$ lies on an iso-path $\omega_i$ of $G_i$, where $G_1,\ldots,G_n$
are labeled cuboid graphs with the same iso-level $c$ and $\omega\neq \omega_i$ for all $i=1,\ldots,n$.
This means, here {\it connectedness} refers to {\it connectedness to all sides of the iso-path}.\\

\noindent{\bf Iso-surfaces and connectedness of iso-surfaces: }Let $\Omega\subset\mathbb{R}^3$ be a
polygonal domain and let $\mathcal{T}$ be the partition of $\overline{\Omega}$ into cuboids.
Furthermore, let the vertices of the cuboids in $\mathcal{T}$ be labeled by real weights from $[0,1]$.
Let $c\in (0,1)$ be an iso-level. Then we call the surfaces obtained by joining all iso-elements
of the labeled cuboids in $\mathcal{T}$ {\it iso-surfaces}. We say that an iso-surface is {\it connected}
if each iso-path which does not have an edge that lies on the boundary of $\Omega$ is connected. That
means, connected iso-surfaces have {\it no holes} except, possibly, at the boundary $\partial\Omega$.
For further theoretical investigations we here assume that the iso-surfaces do not touch the boundary
of $\Omega$.

\subsection{Mapping between labeled graphs}
Let $C$ be a cuboid with set of vertices $V$ and set of edges $E$. Then we denote by $\mathbb{G}(V,E)$
the set of all labeled cuboid graphs $G(V,E,\mathcal{F})$ with $\mathcal{F}\,:\,V\longrightarrow [0,1]$.
Let $q\,:\,[0,1]\longrightarrow [0,1]$ be a given function. Then we define a graph operation
$\mathcal{Q}\,:\,\mathbb{G}\longrightarrow\mathbb{G}$ by
\begin{equation}
G(V,E,\mathcal{F}')=G(V,E,q\circ\mathcal{F}),
\label{eq:labeled-graph-3}
\end{equation}
where $G(V,E,\mathcal{F})\in \mathbb{G}(V,E)$. We denote by $I$ the identity
mapping on $\mathbb{G}(V,E)$.

\begin{definition}(Subgraph mapping). Let $V_h\subset V$ and $E_h\subset E$, where $V$ and $E$ are the set
of vertices and the set of edges of a cuboid $C$, respectively. Then we denote by $\mathbb{H}(V_h,E_h)$ the
set of all subgraphs with set of vertices $V_h$ and set of edges $E_h$ of labeled graphs
$G(V,E,\mathcal{F})\in\mathbb{G}(V,E)$. Let $q_h\,:\,[0,1]\longrightarrow [0,1]$ be a given function.
Then we define the graph operation $\mathcal{Q}_h\,:\,
\mathbb{H}(V_h,E_h)\longrightarrow\mathbb{H}(V_h,E_h)$ by
\begin{equation}
\mathcal{Q}_h(H(V_h,E_h,\mathcal{F}_h))=H(V_h,E_h,\mathcal{F}'_h),
\label{eq:labeled-graph-6}
\end{equation}
where $\mathcal{F}'_h=q_h\circ\mathcal{F}_h$.
\label{def:labeled-graph-2}
\end{definition}

\begin{definition}(Subgraph replacement). Let $H(V_h,E_h,\mathcal{F}_h)$ be a subgraph of the labeled cuboid
graph $G(V,E,\mathcal{F})$. Given $q_h\,\,:[0,1]\longrightarrow [0,1]$, we define a function $\mathcal{F}'$
on $V$ by
\begin{equation}
\mathcal{F}':=\left\{
   \begin{array}{ll}
     \mathcal{F} & \mbox{on }\;\; V\setminus V_h \\
     q_h\circ\mathcal{F}_h & \mbox{on }\;\; V_h
   \end{array}\right..
\label{eq:labeled-graph-7}
\end{equation}
We then set
\begin{equation}
G|_{H\rightarrow H'}:=G(V,E,\mathcal{F}'),
\label{eq:labeled-graph-8}
\end{equation}
and say that $G|_{H\rightarrow H'}$ is obtained from the graph $G(V,E,\mathcal{F})$ by replacing
the subgraph $H(V_h,E_h,\mathcal{F}_h)$ of $G(V,E,\mathcal{F})$ by $H(V_h,E_h,q_h\circ\mathcal{F}_h)$.
\label{def:labeled-graph-3}
\end{definition}

\begin{definition}(Subgraph replacement to a continuous graph). Let $H(V_h,E_h,\mathcal{F}_h)$ be a subgraph
of the labeled cuboid graph $G(V,E,\mathcal{F})$. Additionally, let $G'(V,E,\mathcal{F}')$ be a continuous graph
with $\mathcal{F}'\,:\,V\,\longrightarrow \{0\}$ and let $H'(V_h,E_h,\mathcal{F}_h')$ be the subgraph of $G'$
corresponding to $H$. The subgraphs $H$ and $H'$ have the same nodes and edges. Then we define a function
$\hat{\mathcal{F}}'$ on $V$ by
\begin{equation}
\hat{\mathcal{F}}':=\left\{
   \begin{array}{ll}
     \mathcal{F}' & \mbox{on }\;\; V\setminus V_h \\
     \mathcal{F}_h & \mbox{on }\;\; V_h
   \end{array}\right..
\label{eq:labeled-graph-7-1}
\end{equation}
We set
\begin{equation}
G'|_{H'\rightarrow H}:=G'(V,E,\hat{\mathcal{F}}'),
\label{eq:labeled-graph-8-1}
\end{equation}
and say that $G'|_{H'\rightarrow H}$ is obtained from the graph $G'(V,E,\mathcal{F}')$ by replacing
the subgraph $H'(V_h,E_h,\mathcal{F}_h')$ of $G'(V,E,\mathcal{F}')$ by $H(V_h,E_h,\mathcal{F}_h)$.
\label{def:labeled-graph-3-1}
\end{definition}

\begin{definition}(Subgraph replacement to a disperse graph). Let $H(V_h,E_h,\mathcal{F}_h)$ be a subgraph
of the labeled cuboid graph $G(V,E,\mathcal{F})$. Additionally, let $G'(V,E,\mathcal{F}')$ be a disperse graph
with $\mathcal{F}'\,:\,V\,\longrightarrow \{1\}$ and let $H'(V_h,E_h,\mathcal{F}_h')$ be the subgraph of $G'$
corresponding to $H$. The subgraphs $H$ and $H'$ have the same nodes and edges. Then we define a function
$\hat{\mathcal{F}}'$
on $V$ by
\begin{equation}
\hat{\mathcal{F}}':=\left\{
   \begin{array}{ll}
     \mathcal{F}' & \mbox{on }\;\; V\setminus V_h \\
     \mathcal{F}_h & \mbox{on }\;\; V_h
   \end{array}\right..
\label{eq:labeled-graph-7-2}
\end{equation}
We set
\begin{equation}
G'|_{H'\rightarrow H}:=G'(V,E,\hat{\mathcal{F}}'),
\label{eq:labeled-graph-8-2}
\end{equation}
and say that $G'|_{H'\rightarrow H}$ is obtained from the graph $G'(V,E,\mathcal{F}')$ by replacing
the subgraph $H'(V_h,E_h,\mathcal{F}_h')$ of $G'(V,E,\mathcal{F}')$ by $H(V_h,E_h,\mathcal{F}_h)$.
\label{def:labeled-graph-3-2}
\end{definition}

Now we define a so-called general labeled graph $G'(V',E',\mathcal{F}')$ which we get by substituting
part of a labeled cuboid graph $G(V,E,\mathcal{F})$ for a given iso-level $c\in (0,1)$ by a graph, where the
nodes of the graph are labeled with values in $[0,1]$. The graph that will be substituted consists of
one or more inner iso-paths of $G$. Therefore, such a general labeled cuboid graph contains at least one
inner iso-path of $G$.
\begin{definition}(General labeled graph). Let $G(V,E,\mathcal{F})$ be a labeled cuboid graph.
Let $V_e\subset\mathbb{R}^3$ be a set of $3\leq m\leq 4$ points such that each $P\in V_e$ lies on
an edge $e\in E$. Let $E_e$ be a given set of edges with end points in $V_e$. Assume that each
$P\in V_e$ is an end point of two edges in the set $E_e$.
Let $\mathcal{F}_e\,:\,V_e\longrightarrow [0,1]$ be a labeling on $V_e$.
Then we call $H_e(V_e,E_e,\mathcal{F}_e):=(V_e,E_e,\mathcal{F}_e)$ a labeled graph.
Assume that $\mathcal{F}_e=\mathcal{F}$ on $V\cap V_e$. Then we define a labeled graph
$\tilde{G}(\tilde{V},\tilde{E},\tilde{F}):=(\tilde{V},\tilde{E},\tilde{F})$ by
\begin{equation}
G(\tilde{V},\tilde{E},\tilde{F}):=(V\cup V_e,E\cup E_e,\tilde{F}),
\end{equation}
where
\begin{equation}
\tilde{F}:=\left\{
   \begin{array}{ll}
     \mathcal{F} & \mbox{on }\;\; V \\
     \mathcal{F}_e & \mbox{on }\;\; V_e
   \end{array}\right..
\label{eq:labeled-graph-9a}
\end{equation}
Then $\tilde{G}(\tilde{V},\tilde{E},\tilde{F})$ is called a general labeled graph, or labeled graph for short.
\label{def:labeled-graph-4a}
\end{definition}

\begin{definition}(General subgraph replacement). Let $H(V_h,E_h,\mathcal{F}_h)$ be a subgraph of the labeled
cuboid graph $G(V,E,\mathcal{F})$. Let $\tilde{H}(\tilde{V},\tilde{E},\tilde{\mathcal{F}})$ be a labeled graph
such that $V_h\subset \tilde{V}$, $E_h\subset \tilde{E}$, $\tilde{V}\subset E_h$ and  $\tilde{\mathcal{F}}=\mathcal{F}$
on $V\cap \tilde{V}$. We define a function $\mathcal{F}'$ on $V\cup V_h$ by
\begin{equation}
\mathcal{F}':=\left\{
   \begin{array}{ll}
     \mathcal{F} & \mbox{on }\;\; V\setminus V_h \\
     \mathcal{F}_h & \mbox{on }\;\; V_h\\
     \tilde{\mathcal{F}} & \mbox{on }\;\; \tilde{V}\setminus V_h
   \end{array}\right..
\label{eq:labeled-graph-9}
\end{equation}
Then we set
\begin{equation}
G|_{H\rightarrow \tilde{H}}:=G(\hat{V},\hat{E},\mathcal{F}'),
\label{eq:labeled-graph-10}
\end{equation}
where $\hat{V}=V\cup (\tilde{V}\backslash V_h)$ and $\hat{E}=E\cup (\tilde{E}\backslash E_h)$.
We say that $G|_{H\rightarrow \tilde{H}}$ is obtained from the graph $G(V,E,\mathcal{F})$ by replacing
the subgraph $H(V_h,E_h,\mathcal{F}_h)$ of $G(V,E,\mathcal{F})$ by the graph
$\tilde{H}(\tilde{V},\tilde{E},\tilde{\mathcal{F}})$. Note that the labeled graph $G|_{H\rightarrow \tilde{H}}$
may no longer be a labeled cuboid graph.
\label{def:labeled-graph-4}
\end{definition}

\noindent{\bf Note: }In this paper, any graph operation applied on a labeled graph with a given iso-level $c\in (0,1)$
does not change the iso-level $c$, i.e. the transformed labeled graph has the same iso-level
$c$.

\section{Equivalence Classes of Labeled Cuboid Graphs}
From here on, whenever we consider (one or several) labeled cuboid graphs, it is understood that also
an iso-level $c\in (0,1)$ has been chosen. We then speak of "a labeled cuboid graph $G(V,E,\mathcal{F})$ with
iso-level $c\in (0,1)$".

For the computation of iso-paths for a labeled cuboid graph $G(V,E,\mathcal{F})$ with iso-level $c\in (0,1)$
we compare the node-labels with $c$. The exact values of the nodes are not important for the
investigation of iso-paths of the graph $G$. Therefore, in the following two definitions we introduce
the important concept of equivalence classes of labeled cuboid graphs. The first definition considers
for each node of $G$ whether the node is disperse or not. In the second definition of an equivalence
class of labeled cuboid graphs, besides the disperse nodes of $G$, the differences between iso-nodes
and nodes with node value less than $c$ are accounted for.

\begin{definition}
Suppose $G_1(V_1,E_1,\mathcal{F}_1)$ and $G_2(V_2,E_2,\mathcal{F}_2)$ are labeled cuboid graphs with
iso-level $c\in (0,1)$. We call the graphs $G_1$ and $G_2$ $\circ$-equivalent if the following conditions
are satisfied:
\begin{enumerate}
\item[(i)] $D(G_1)=D(G_2)$,
\item[(ii)] both $G_1$ and $G_2$ have the same number of L-faces,
\item[(iii)] to each $Q\in V_1$ with $\mathcal{F}_1(Q)>c$ and $P_1,P_2,P_3\in V_1$ such that $P_1, P_2,P_3$ are
incident to $Q$, there exists $Q'\in V_2$ with $\mathcal{F}_2(Q')>c$ and $P_1',P_2',P_3'\in V_2$ such that
$P_1', P_2',P_3'$ are incident to $Q'$ and, to each $i=1,2,3$, one of the following holds:
  \begin{enumerate}
   \item[(a)] if $\mathcal{F}_1(P_i)>c$ then $\mathcal{F}_2(P_i')>c$,
   \item[(b)] if $\mathcal{F}_1(P_i)\leq c$ then $\mathcal{F}_2(P_i')\leq c$.
  \end{enumerate}
\end{enumerate}
The mapping from $Q$ to $Q'$ is required to be a bijection. We denote the $\circ$-equivalence between $G_1$
and $G_2$ by $G_1\Longleftrightarrow_{\circ}G_2$. Additionally, we denote by $[G_1(V_1,E_1,\mathcal{F}_1)]_{\circ}$
the $\circ$-equivalence class, defined by
\begin{equation}
[G_1(V_1,E_1,\mathcal{F}_1)]_{\circ}:=\{G(V,E,\mathcal{F})\,:\,G\Longleftrightarrow_{\circ}G_1\}.
\label{eq:equivalence-1}
\end{equation}
\label{def:equivalence-1}
\end{definition}

\begin{definition}
Suppose $G_1(V_1,E_1,\mathcal{F}_1)$ and $G_2(V_2,E_2,\mathcal{F}_2)$ are labeled cuboid graphs
with iso-level $c\in (0,1)$. We call the graphs $G_1$ and $G_2$ $\square$-equivalent if the following
conditions are satisfied:
\begin{enumerate}
\item[(i)] $D(G_1)=D(G_2)$,
\item[(ii)] both $G_1$ and $G_2$ have the same number of $L$-faces,
\item[(iii)] to each $Q\in V_1$ with $\mathcal{F}_1(Q)>c$ and $P_1,P_2,P_3\in V_1$ such that $P_1, P_2,P_3$ are
incident to $Q$, there exists $Q'\in V_2$ with $\mathcal{F}_2(Q')>c$ and $P_1',P_2',P_3'\in V_2$ such that
$P_1', P_2',P_3'$ are incident to $Q'$ and, to each $i=1,2,3$, one of the following holds:
  \begin{enumerate}
   \item[(a)] if $\mathcal{F}_1(P_i)>c$ then $\mathcal{F}_2(P_i')>c$,
   \item[(b)] if $\mathcal{F}_1(P_i)<c$ then $\mathcal{F}_2(P_i')<c$,
   \item[(c)] if $\mathcal{F}_1(P_i)=c$ then $\mathcal{F}_2(P_i')=c$.
  \end{enumerate}
\end{enumerate}
The mapping from $Q$ to $Q'$ is required to be a bijection. We denote the $\square$-equivalence between $G_1$
and $G_2$ by $G_1\Longleftrightarrow_{\square}G_2$. Additionally, we denote by $[G_1(V_1,E_1,\mathcal{F}_1)]_{\square}$
the $\square$-equivalence class, defined by
\begin{equation}
[G_1(V_1,E_1,\mathcal{F}_1)]_{\square}:=\{G(V,E,\mathcal{F})\,:\,G\Longleftrightarrow_{\square}G_1\}.
\label{eq:equivalence-2}
\end{equation}
\label{def:equivalence-2}
\end{definition}
\noindent{\bf Illustration of $\circ$-equivalence class: }Let $G(V,E,\mathcal{F})$ be a labeled cuboid
graph with iso-level $c\in (0,1)$. Let $V=\{P_1,\ldots,P_8\}$ be the nodes of $G$ and let  $C$ be the
cuboid of $G$. Consider a sketch that shows the edges and vertices of $C$, and such that the vertices
of $C$ are marked as follows:
\begin{enumerate}
\item[(i)] $P_i$ is marked by $\bullet$ if $\mathcal{F}(P_i)>c$,
\item[(ii)] $P_i$ is marked by $\circ$ if $\mathcal{F}(P_i)\leq c$.
\end{enumerate}
Then we say that the sketch represents $[G(V,E,\mathcal{F})]_{\circ}$.\\

\noindent{\bf Illustration of $\square$-equivalence class: }Let $G(V,E,\mathcal{F})$ be a labeled cuboid
graph with iso-level $c\in (0,1)$. Let $V=\{Q_1,\ldots,Q_m,P_{m+1},\ldots,P_8\}$ with $m=D(G)$
be the nodes of $G$, where $\mathcal{F}(Q_i)>c$ for $i=1,\ldots,m$
and $\mathcal{F}(P_j)\leq c$ for $j=m+1,\ldots,8$. Let $C$ be the cuboid of $G$. Consider a sketch
that shows the edges and vertices of $C$, and such that the vertices of $C$ are marked as follows:
\begin{enumerate}
\item[(i)] $Q_i$ is marked by $\bullet$ for all $i\in\{1,\ldots,m\}$,
\item[(ii)] $P_i$ is marked by $\circ$ for all $i\in\{m+1,\ldots,8\}$ with $\mathcal{F}(P_i)\leq c$
and $P_i$ is not incident to any one of the nodes in $\{Q_1,\ldots,Q_m\}$,
\item[(iii)] $P_i$ is marked by $\circleddash$ for all $i\in\{m+1,\ldots,8\}$ with $\mathcal{F}(P_i)< c$
and $P_i$ is incident to one of the nodes in $\{Q_1,\ldots,Q_m\}$,
\item[(iv)]  $P_i$ is marked by $\square$ for all $i\in\{m+1,\ldots,8\}$ with $\mathcal{F}(P_i)= c$
and $P_i$ is incident to any one of the nodes in $\{Q_1,\ldots,Q_m\}$.
\end{enumerate}
Then we say that the sketch represents $[G(V,E,\mathcal{F})]_{\square}$.\\

Sketches will be "numbered" by $(a),(b),(c),\ldots$ or, $(a1),(a2),\ldots$ throughout this paper. Examples
for $\circ$- and $\square$-equivalence classes are given in Figure~\ref{image_6}.
\begin{figure}[!ht]
\includegraphics[width=0.8\linewidth]{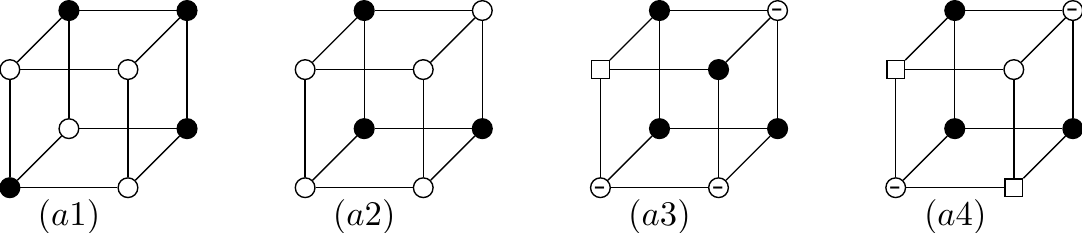}
\caption{Sketches $(a1),\ldots,(a4)$ represent equivalence classes $[(a1)]_{\circ}$,
$[(a2)]_{\circ}$, $[(a3)]_{\square}$ and $[(a4)]_{\square}$, respectively.}
\label{image_6}
\end{figure}
\FloatBarrier

\noindent{\bf Equivalence class of a subgraph: }Suppose $H(V_h,E_h,\mathcal{F}_h)$ is a subgraph of
a labeled cuboid graph $G(V,E,\mathcal{F})$ with iso-level $c\in (0,1)$. Then we denote by
$[H(V_h,E_h,\mathcal{F}_h)]_{\circ}$ the $\circ$-equivalence class of $H$ such that each element
in $[H(V_h,E_h,\mathcal{F}_h)]_{\circ}$ satisfies Definition~\ref{def:equivalence-1} restricted
to $H(V_h,E_h,\mathcal{F}_h)$. Analogously, we denote by $[H(V_h,E_h,\mathcal{F}_h)]_{\square}$ the
$\square$-equivalence class of $H$ such that each element in $[H(V_h,E_h,\mathcal{F}_h)]_{\square}$
satisfies Definition~\ref{def:equivalence-2} restricted to $H(V_h,E_h,\mathcal{F}_h)$.\\

\noindent{\bf Illustration of $\circ$- and $\square$-equivalence subclasses: }Analogously to
$\circ$- and $\square$-equivalence classes, we can represent both subclasses by sketches.
Given a sketch $(a)$, we denote by $[(a)]_{\circ}$ and $[(a)]_{\square}$ the equivalence
classes of the graph or subgraph represented by $(a)$.\\

\noindent Figure~\ref{image_7} illustrates some examples of $\circ$- and $\square$-equivalence
subclasses.
\begin{figure}[!ht]
\includegraphics[width=0.9\linewidth]{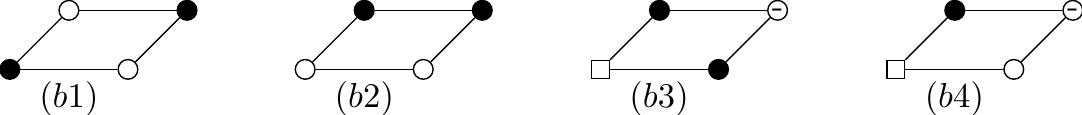}
\caption{Sketches $(b1),\ldots,(b4)$ are representations of equivalence subclasses
$[(b1)]_{\circ}$, $[(b2)]_{\circ}$, $[(b3)]_{\square}$ and $[(b4)]_{\square}$, respectively.}
\label{image_7}
\end{figure}
\FloatBarrier

Finally, we define an additional class of labeled cuboid graphs which we call $\circ\!-\!\star$-class and
$\square\!-\!\star$-class, where for certain nodes no restriction on the label is done.
\begin{definition}
Let $G(V,E,\mathcal{F})$ be a labeled cuboid graph with iso-level $c\in (0,1)$. Let
$\tilde{V}\subset V$. Then we define the $\circ\!-\!\star$-class of $G$ corresponding to $\tilde{V}$ by
\begin{equation}
[G(V,E,\mathcal{F};\tilde{V})]_{\circ}^{\star}:=\bigcup_{\tilde{\mathcal{F}}\in\mathbb{F}}
[G(V,E,\tilde{\mathcal{F}})]_{\circ},
\label{eq:equivalence-3}
\end{equation}
where
\begin{equation}
\mathbb{F}=\{\hat{\mathcal{F}}\,:\,V\longrightarrow [0,1]\,|\,\hat{\mathcal{F}}=\mathcal{F}\mbox{ on }
V\setminus\tilde{V}\mbox\}.
\label{eq:equivalence-4}
\end{equation}
\label{def:equivalence-3}
\end{definition}

\begin{definition}
Let $G(V,E,\mathcal{F})$ be a labeled cuboid graph with iso-level $c\in (0,1)$. Let $\tilde{V}\subset V$.
Then we define the $\square\!-\!\star$-class of $G$ corresponding to $\tilde{V}$ by
\begin{equation}
[G(V,E,\mathcal{F};\tilde{V})]_{\square}^{\star}:=\bigcup_{\tilde{\mathcal{F}}\in\mathbb{F}}[G(V,E,\tilde{\mathcal{F}})]_{\square},
\label{eq:equivalence-5}
\end{equation}
where
\begin{equation}
\mathbb{F}=\{\hat{\mathcal{F}}\,:\,V\longrightarrow [0,1]\,|\,\hat{\mathcal{F}}=\mathcal{F}\mbox{ on }
V\setminus\tilde{V}\mbox\}.
\label{eq:equivalence-6}
\end{equation}
\label{def:equivalence-4}
\end{definition}

\noindent{\bf Illustration of $\circ\!-\!\star$-class: }Suppose a sketch $(a)$ represents
$[G(V,E,\mathcal{F})]_{\circ}$, i.e. $[(a)]_{\circ}=[G(V,E,\mathcal{F})]_{\circ}$. Consider a sketch
$(\tilde{a})$ which is a copy of $(a)$ but those nodes in $\tilde{a}$ corresponding to $\tilde{V}\subset V$
are marked by the symbol $\diamondsuit\!\!\!\!\star$. We then say that $(\tilde{a})$ represents $[G(V,E,
\mathcal{F};\tilde{V})]_{\circ}^{\star}$ and write $[(\tilde{a})]_{\circ}^{\star}$ for this. \\

\noindent{\bf Illustration of $\square\!-\!\star$-equivalence class: }Suppose a sketch $(a)$ represents
$[G(V,E,\mathcal{F})]_{\square}$, i.e. $[(a)]_{\square}=[G(V,E,\mathcal{F})]_{\square}$. Consider a sketch
$(\tilde{a})$ which is a copy of $(a)$ but those nodes in $\tilde{a}$ corresponding to $\tilde{V}\subset V$
are marked by the symbol $\diamondsuit\!\!\!\!\star$. We then say that $(\tilde{a})$ represents
$[G(V,E,\mathcal{F};\tilde{V})]_{\square}^{\star}$ and write $[(\tilde{a})]_{\square}^{\star}$ for this. \\

\noindent{\bf $\circ\!-\!\star$-subclass and $\square\!-\!\star$-subclass: }Suppose $H(V_h,E_h,\mathcal{F}_h)$
is a subgraph of a labeled cuboid graph $G(V,E,\mathcal{F})$ with iso-level $c\in (0,1)$. Let $\tilde{V}\subset V$
and assume $\tilde{V}_h=V_h\cap \tilde{V}\neq\emptyset$. With an analogous definition as \eqref{eq:equivalence-3}
for $H(V_h,E_h,\mathcal{F}_h)$, we get the  $\circ\!-\!\star$-subclass of $H(V_h,E_h,\mathcal{F}_h)$ corresponding
to $\tilde{V}_h$ which is denoted by $[H(V_h,E_h,\mathcal{F}_h;\tilde{V}_h)]_{\circ}^{\star}$. With an analogous
definition as \eqref{eq:equivalence-5} for $H(V_h,E_h,\mathcal{F}_h)$, we get the  $\square\!-\!\star$-subclass of
$H(V_h,E_h,\mathcal{F}_h)$ corresponding to $\tilde{V}_h$ which is denoted by
$[H(V_h,E_h,\mathcal{F}_h;\tilde{V}_h)]_{\square}^{\star}$.\\

\noindent{\bf Illustration of $\circ\!-\!\star$-subclass and $\square\!-\!\star$-subclass: }Suppose
$[(a)]_{\circ}=[H(V_h,E_h,\mathcal{F}_h)]_{\circ}$, i.e. the sketch $(a)$ represents $[H(V_h,E_h,\mathcal{F}_h)]_{\circ}$.
Consider a sketch $(\tilde{a})$ which is a copy of $(a)$ but the nodes in $\tilde{a}$ corresponding to
$\tilde{V}_h\subset V_h$ are marked by the symbol $\diamondsuit\!\!\!\!\star$. We then say that $(\tilde{a})$ represents
$[H(V_h,E_h,\mathcal{F}_h;\tilde{V}_h)]_{\circ}^{\star}$ and write $[(\tilde{a})]_{\circ}^{\star}$ for this. If we start
with $[(a)]_{\square}=[H(V_h,E_h,\mathcal{F}_h)]_{\square}$ instead, we obtain a sketch $\tilde{a}$ that represents
the subclass $[H(V_h,E_h,\mathcal{F}_h;\tilde{V}_h)]_{\square}^{\star}$; we write $[(\tilde{a})]_{\square}^{\star}$ for this.\\

\noindent Figure~\ref{image_8_9} shows examples of $\circ\!-\!\star$- and
$\square\!-\!\star$-classes and subclasses with iso-level $c\in (0,1)$.
\begin{figure}[!ht]
\includegraphics[width=0.8\linewidth]{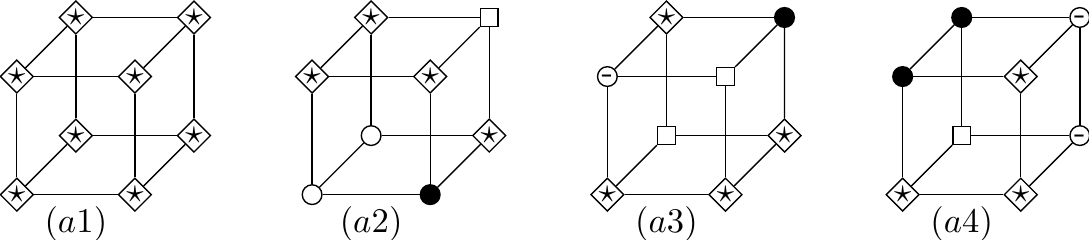}\\ \\
\includegraphics[width=0.8\linewidth]{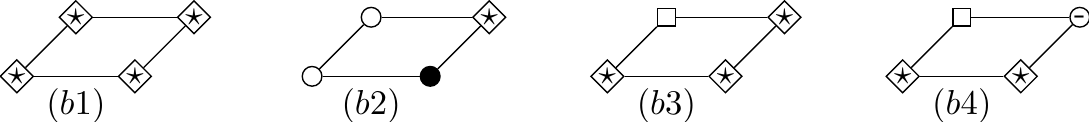}
\caption{Sketches $(a1)$ and $(a2)$ represent $\circ\!-\!\star$-classes and sketches $(a3)$ and $(a4)$ represent
$\square\!-\!\star$-classes. Sketches $(b1)$ and $(b2)$ represent $\circ\!-\!\star$-subclasses and
sketches $(b3)$ and $(b4)$ represent $\square\!-\!\star$-subclasses.}
\label{image_8_9}
\end{figure}
\FloatBarrier

\begin{notation}
Sometimes we use sketches which can have either all four symbols $\circ,\circleddash,\square,\bullet$ or
all three symbols $\circ,\circleddash,\bullet$ to represent a class of labeled cuboid graphs or subgraphs of it.
These special sketches that we use in the next sections are given in Figure~\ref{image_10}. We
call the equivalence classes represented by the sketches $(a1)$, $(a2)$ and $(a3)$ the $\square$-equivalence classes
and the equivalence class represented by the sketch $(a4)$ the $\square$-equivalence subclass.
\label{not:special}
\end{notation}
\begin{figure}[!ht]
\includegraphics[width=0.8\linewidth]{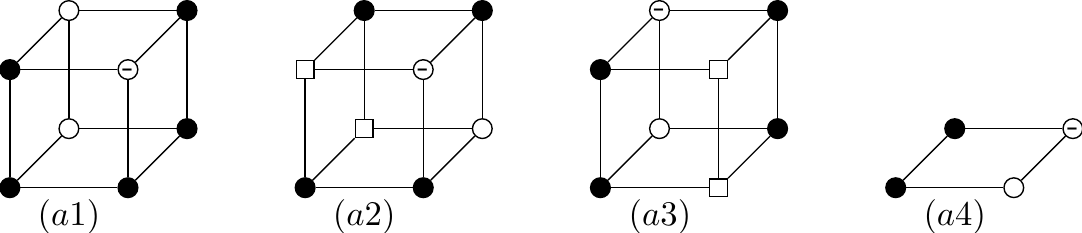}
\caption{Sketches $(a1)$, $(a2)$ and $(a3)$ represent $\square$-equivalence classes and sketch
$(a4)$ represents $\square$-equivalence subclass.}
\label{image_10}
\end{figure}

\section{Rules of Iso-path Computations}
In this section we define mappings that will be applied on $\mathbb{G}(V,E)$. They
will be used for the computation of iso-paths in labeled cuboid graphs with iso-level $c\in (0,1)$.
These mappings replace subgraphs of a graph by other graphs just as described by
Definitions~\ref{def:labeled-graph-3} and~\ref{def:labeled-graph-4}.

The distinction between different types of face subgraphs of a labeled cuboid graph $G(V,E,\mathcal{F})$
with iso-level $c\in (0,1)$ introduced in Section 2 is important
for the computation of iso-paths of $G$. We illustrate by Figure~\ref{image_11} the different types
of faces of $G$. Sketch $(a1)$ represents the $\circ\!-\!\star$-subclass $[(a1)]_{\circ}^{\star}$ of
regular faces and sketch $(a2)$ represents the $\square$-equivalence subclass $[(a2)]_{\square}$
of regular faces. Sketch $(a3)$ represents the $\square$-equivalence subclass of singular faces.
Sketches $(a4)$, $(a5)$ and $(a6)$ represent the possible $\circ$- and $\square$-equivalence
subclasses of L-faces. Here $[(a5)]_{\square}$ is a $\square$-equivalence subclass of trivial L-faces
and $[(a6)]_{\square}$ is a $\square$-equivalence subclass of non-trivial L-faces.
\begin{figure}[!ht]
\includegraphics[width=0.7\linewidth]{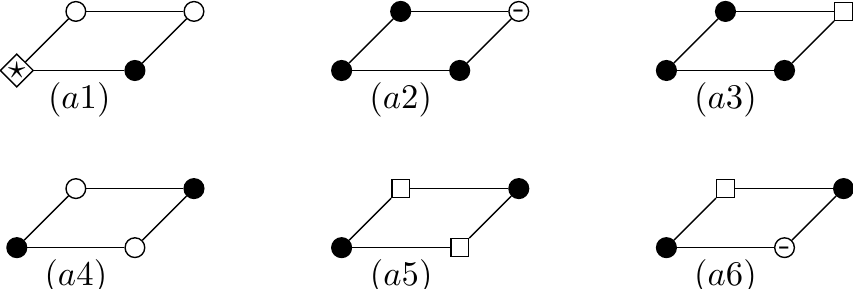}
\caption{Sketches $(a1),\ldots,(a6)$ illustrate the different types of face subgraphs.}
\label{image_11}
\end{figure}
\FloatBarrier

\subsection{Removing Singular and Isolated Iso-paths}
Let $G(V,E,\mathcal{F})$ be a labeled cuboid graph with iso-level $c\in (0,1)$. In this section
we introduce graph operations on $G$ which are called $T$-rules and $F$-rules.
The $T$-rules are denoted by $T_1$ and $T_2$ and remove {\it singular iso-paths} in $G$, where
a singular iso-path is a degenerate closed path with only one or two nodes. Singular
iso-paths correspond to iso-elements of surface measure zero. The $T_1$-rule removes singular
iso-paths with only one node and the $T_2$-rule removes singular iso-paths with two nodes.

The $F$-rules are graph operations denoted by $F_1$ and $F_2$. The $F_1$-rule removes
iso-paths in $G$ such that the iso-element computed from the iso-paths separates only
disperse nodes of $G$. The $F_2$-rule removes an iso-path of $G$ if $G$ has the
labeled cuboid graph $G'(V',E',\mathcal{F}')$ with the same iso-level $c$
as a face neighbor and both $G$ and $G'$ contain four iso-nodes and four disperse
nodes such that the iso-nodes lie on the common face. In this case, the
iso-element computed from the iso-path separates only the disperse nodes of both
face neighboring graphs. An iso-path of $G$ which is removed using the $F_1$- or $F_2$-rule is
called an {\it isolated iso-path}.\\

\noindent{\bf Isolated iso-element: }An iso-element computed from an isolated iso-path is called
an {\it isolated iso-element}.\\

\noindent{\bf Regular labeled cuboid graph: }Let $G(V,E,\mathcal{F})$ be a labeled cuboid graph
with iso-level $c\in (0,1)$. Assume $G$ is neither a disperse nor a continuous graph. Furthermore,
assume $G$ has neither singular iso-paths nor an isolated iso-path. Then we call $G$ a {\it regular
labeled cuboid graph}. Sometimes we use the abbreviation $G$ is {\it regular} instead of $G$ is a
regular labeled cuboid graph.

\subsubsection{T- and F-graphs}
In this section we introduce the T-subgraphs and F-graphs of a labeled cuboid graph $G(V,E,\mathcal{F})$
with iso-level $c\in (0,1)$. The existence of such graphs in $G$ is a possible indication of singular
iso-path or isolated iso-path existence in $G$. The detection of isolated iso-paths is necessary
for the computation of iso-paths of $G$, since isolated iso-paths may not be connected. Removing
isolated iso-paths guarantees the iso-path connectedness as will be shown in Section 6.

\begin{definition}($T_1$-subgraph). Let $G(V,E,\mathcal{F})$ be a labeled cuboid graph
with iso-level $c\in (0,1)$. Let $H(V_h,E_h,\mathcal{F}_h)$ be a subgraph of $G$ such that the following
conditions hold:
\begin{enumerate}
\item the number of points in $V_h$ is four ($|V_h|=4$),
\item there exist a point $P'\in V_h$ such that $P'$ is incident to all points $P\in V_h\setminus\{P'\}$ in $H$ and
\[
\mathcal{F}_h(P')=c\qquad\mbox{and } \qquad \mathcal{F}_h(P)>c \quad\;\forall P\in V_h\setminus\{P'\}.
\]
\end{enumerate}
Then $H$ is called a $T_1$-subgraph.
\label{def:iso-path-2}
\end{definition}

\begin{definition}($T_2$-subgraph). Let $G(V,E,\mathcal{F})$ be a labeled cuboid graph
with iso-level $c\in (0,1)$. Let $H(V_h,E_h,\mathcal{F}_h)$ be a subgraph of $G$ such that the following
conditions hold:
\begin{enumerate}
\item the number of points in $V_h$ is seven ($|V_h|=7$),
\item there exist two points $P_1,P_2\in V_h$ such that $P_1$ is incident to $P_2$ in $H$ and
\[
\mathcal{F}_h(P_1)=\mathcal{F}_h(P_2)=c\qquad\mbox{and }\qquad
\mathcal{F}_h(P)>c \quad\;\forall P\in V_h\setminus\{P_1,P_2\}.
\]
\end{enumerate}
Then $H$ is called a $T_2$-subgraph.
\label{def:iso-path-3}
\end{definition}
\noindent $T_1$- and $T_2$-subgraphs are subsumed as $T$-subgraphs.

\begin{definition}($F_1$-graph). Let $G(V,E,\mathcal{F})$ be a labeled cuboid graph
with iso-level $c\in (0,1)$. Let there be $\tilde{V}\subset V$ with $|\tilde{V}|=4$ and
\begin{enumerate}
\item each $P\in\tilde{V}$ is incident only to two points in $V\setminus \tilde{V}$,
\item $\mathcal{F}(P)=c \quad\;\forall P\in \tilde{V}$,
\item $\mathcal{F}(P)>c \quad\;\forall P\in V\setminus \tilde{V}$.
\end{enumerate}
Then $G$ is called an $F_1$-graph.
\label{def:iso-path-4}
\end{definition}

\begin{definition}($F_2$-graph). Suppose $G_1(V_1,E_1,\mathcal{F}_1)$ and $G_2(V_2,E_2,\mathcal{F}_2)$
are face-neighbored labeled cuboid graphs with iso-level $c\in (0,1)$. Let the following properties hold:
\begin{enumerate}
\item $\mathcal{F}_1(P)=\mathcal{F}_2(P)=c \quad\;\forall P\in V_1\cap V_2$,
\item $\mathcal{F}_1(P)>c \quad\;\forall P\in V_1\setminus (V_1\cap V_2)$,
\item $\mathcal{F}_2(P)>c \quad\;\forall P\in V_2\setminus (V_1\cap V_2)$.
\end{enumerate}
Then $G_1$ is called an $F_2$-graph.
\label{def:iso-path-5}
\end{definition}
\noindent $F_1$- and $F_2$-graphs are subsumed as $F$-graphs.

Figure~\ref{image_12} illustrates the $T$-subgraphs and $F$-graphs. The sketches $(a)$ and $(b)$
shown in Figure~\ref{image_12} represent the equivalence subclasses $[(a)]_{\square}$ and
$[(b)]_{\square}$. The $T_1$- and $T_2$-subgraphs lie in $[(a)]_{\square}$ and $[(b)]_{\square}$,
respectively. The sketches $(c),(d_1),(d_2)$ in Figure~\ref{image_12} represent $[(c)]_{\square}$,
$[(d_1)]_{\square}$ and $[(d_2)]_{\square}$, respectively. The $F_1$-graph lies in $[(c)]_{\square}$
and the $F_2$-graph lies in $[(d_1)]_{\square}$ as well as in $[(d_2)]_{\square}$.
\begin{figure}[!ht]
\includegraphics[width=0.9\linewidth]{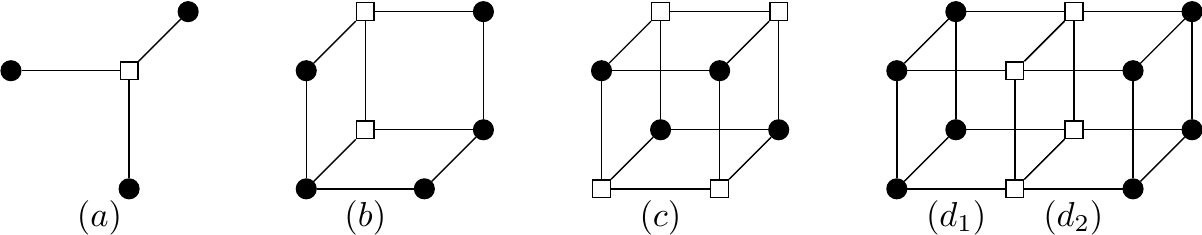}
\caption{Sketches $(a)$ and $(b)$ represent $T_1$- and $T_2$-subgraphs, respectively and $F_1$-
and $F_2$-graphs are represented by sketch $(c)$ and the sketches $(d1),(d2)$, respectively.}
\label{image_12}
\end{figure}

\subsubsection{T- and F-rules}
A labeled cuboid graph $G(V,E,\mathcal{F})\in\mathbb{G}(V,E)$ with iso-level $c\in (0,1)$ can have
one to four singular iso-paths or an isolated iso-path, but not both types of iso-paths. This section
is devoted to the indication and deletion of singular iso-paths or of an isolated iso-path.

Let $G(V,E,\mathcal{F})\in\mathbb{G}(V,E)$ and $c\in (0,1)$ be an iso-level. Let
$q_0, q_1\,:\,[0,1]\longrightarrow [0,1]$ be defined by
\begin{equation}
q_0(x):=\left\{
   \begin{array}{ll}
     0 & \mbox{if }\;\; x>c \\
     x & \mbox{else}
   \end{array}\right.
\label{eq:iso-path-1}
\end{equation}
and
\begin{equation}
q_1(x):=\left\{
   \begin{array}{ll}
     1 & \mbox{if }\;\; x\leq c \\
     x & \mbox{else}
   \end{array}\right..
\label{eq:iso-path-2}
\end{equation}
Let $W\subset V$. We define mappings $Q_0, Q_1\,:\,\mathbb{G}(V,E)\longrightarrow\mathbb{G}(V,E)$ by
$Q_0(G(V,E,\mathcal{F})) = G(V,E,\mathcal{R}_0)$ and $Q_1(G(V,E,\mathcal{F})) = G(V,E,\mathcal{R}_1)$, where
\begin{equation}
\mathcal{R}_0:=\left\{
   \begin{array}{ll}
     \mathcal{F} & \mbox{on }\;\; V\setminus W \\
     q_0\circ\mathcal{F} & \mbox{on }\;\; W
   \end{array}\right.
\label{eq:iso-path-3}
\end{equation}
and
\begin{equation}
\mathcal{R}_1:=\left\{
   \begin{array}{ll}
     \mathcal{F} & \mbox{on }\;\; V\setminus W \\
     q_1\circ\mathcal{F} & \mbox{on }\;\; W
   \end{array}\right..
\label{eq:iso-path-4}
\end{equation}

\begin{definition}($T_1$-rule). Let $G(V,E,\mathcal{F})$ be a labeled cuboid graph with iso-level
$c\in (0,1)$. Let $H(V_h,E_h,\mathcal{F}_h)$ be a $T_1$-subgraph of $G$. We call the mapping
$Q_1$, defined by setting $W:=V_h$ in \eqref{eq:iso-path-4}, a $T_1$-rule. If required for clarity,
we speak of the $T_1$-rule with respect to $H$.
\label{def:iso-path-6}
\end{definition}

\begin{definition}($T_2$-rule). Let $G(V,E,\mathcal{F})$ be a labeled cuboid graph with iso-level
$c\in (0,1)$. Let $H(V_h,E_h,\mathcal{F}_h)$ be a $T_2$-subgraph of $G$. We call the mapping
$Q_1$, defined by setting $W:=V_h$ in \eqref{eq:iso-path-4}, a $T_2$-rule. If required for clarity,
we speak of the $T_2$-rule with respect to $H$.
\label{def:iso-path-7}
\end{definition}
\noindent We subsume the $T_1$- and $T_2$-rules as $T$-rules.

\begin{definition}($T_1^*$-rule). Let $G(V,E,\mathcal{F})$ be a labeled cuboid graph with iso-level
$c\in (0,1)$. Let $G$ have $1\leq n\leq 4$ distinct $T_1$-subgraphs, denoted as $H_1(V_{h_1},E_{h_1},
\mathcal{F}_{h_1}),\ldots,H_n(V_{h_n},E_{h_n},\mathcal{F}_{h_n})$.  We consider four cases, where
in the case $i\in\{1,\ldots,n\}$, the $T_1$-rule changes the node values of $H_i$ according to:
\begin{enumerate}
\item Case $i=1$: $G_1(V,E,\mathcal{F}_1):=T_1(G(V,E,\mathcal{F}))$ ($T_1$-rule w.r. to $H_1$)
\item Case $i=2$: $G_2(V,E,\mathcal{F}_2):=T_1(G_1(V,E,\mathcal{F}_1))$ ($T_1$-rule w.r. to $H_2$)
\item Case $i=3$: $G_3(V,E,\mathcal{F}_3):=T_1(G_2(V,E,\mathcal{F}_2))$ ($T_1$-rule w.r. to $H_3$)
\item Case $i=4$: $G_4(V,E,\mathcal{F}_4):=T_1(G_3(V,E,\mathcal{F}_3))$ ($T_1$-rule w.r. to $H_4$).
\end{enumerate}
Then we write
\begin{equation}
T^*_1(G(V,E,\mathcal{F})):=G_n(V,E,\mathcal{F}_n).
\label{eq:iso-path-star-map}
\end{equation}
We call this $T_1^*$-rule. If we apply the $T_1^*$-rule to $G$ then $T_1^*(G)$
will have no more $T_1$-subgraphs.
\label{def:iso-path-star-map}
\end{definition}

\begin{definition}($F$-rules). Let $G(V,E,\mathcal{F})$ be a labeled cuboid graph with iso-level
$c\in (0,1)$. Let $G$ be an $F_1$- or an $F_2$-graph. Let $V_h=\{P_1,P_2,P_3,P_4\}$
$\subset V$ be such that $\mathcal{F}(P)=c$ for all $P\in V_h$. Then we call the mapping $Q_1$,
defined by setting $W:=V_h$ in \eqref{eq:iso-path-4}, an $F_1$-rule or an $F_2$-rule if $G$ is an
$F_1$- or $F_2$-graph, respectively. The $F_1$- and $F_2$-rules are denoted by $F_1$ and $F_2$,
respectively.
\label{def:iso-path-8}
\end{definition}
\noindent We subsume the $F_1$- and $F_2$-rules as $F$-rules.

The $T$- and $F$-rules are illustrated in Figure~\ref{image_13_14} and~\ref{image_15_16}, respectively.
The sequence $1$ and $2$ in Figure~\ref{image_13_14} represents $T_1$- and $T_2$-rules, respectively.
The sequence $1$ and $2$ in Figure~\ref{image_15_16} represents $F_1$- and $F_2$-rules, respectively.
\begin{figure}[!ht]
$1$.
\begin{tabular}[c]{l}
\includegraphics[width=0.4\linewidth]{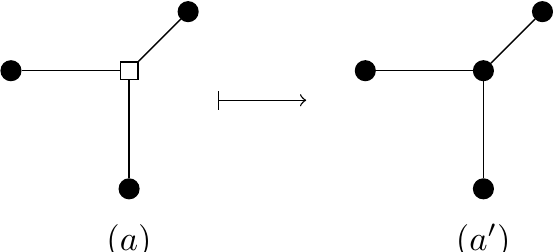}
\end{tabular}\\ \\

$2$.
\begin{tabular}[c]{l}
\includegraphics[width=0.4\linewidth]{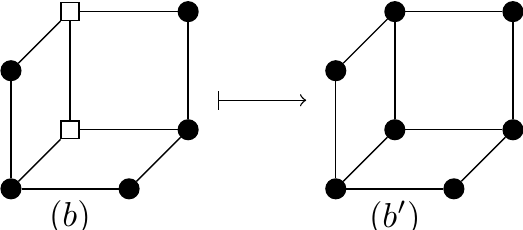}
\end{tabular}
\caption{The sequence $1$ and $2$ illustrate the $T_1$- and $T_2$-rules, respectively.}
\label{image_13_14}
\end{figure}
\begin{figure}[!ht]
$1$.
\begin{tabular}[c]{l}
\includegraphics[width=0.4\linewidth]{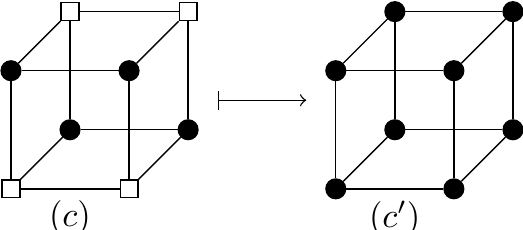}
\end{tabular}\\ \\

$2$.
\begin{tabular}[c]{l}
\includegraphics[width=0.6\linewidth]{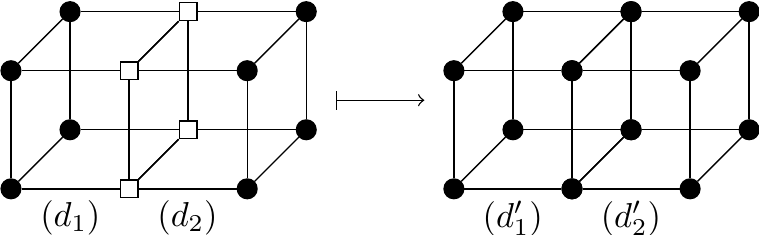}
\end{tabular}
\caption{The sequence $1$ and $2$ illustrate the $F_1$- and $F_2$-rules, respectively.}
\label{image_15_16}
\end{figure}
\FloatBarrier

The following result about singular faces will be proved using $T$-rules.
\begin{proposition}(Singular faces). Let $G(V,E,\mathcal{F})$ be a labeled cuboid graph with iso-level
$c\in (0,1)$. Let $G$ be regular. Suppose that $G$ has a singular face. Then the maximum number of
singular faces which $G$ can have is three. Furthermore, if $G$ has $n=2$ or $n=3$ singular faces then
there exist a total of $n$ iso-points which are iso-nodes in $G$ such that these iso-nodes lie on the
singular faces. Then each pair of these iso-nodes lies on a regular face of $G$ or on a space diagonal
of the cuboid of $G$. Moreover, an iso-node can never be on two distinct singular faces.
\label{prop:iso-path-1}
\end{proposition}
\begin{proof}
We give the proof of Proposition~\ref{prop:iso-path-1} by using Figure~\ref{image_17}.
If $G\in [(a1)]_{\square}^{\star}$ then $G$ has at least one singular face. The other possibilities
for $G$ to have two singular faces occur only if $G\in [(a2)]_{\square}$ or $G\in [(a3)]_{\circ}$.
The only possibility for $G$ to have three singular faces is $G\in [(a4)]_{\square}$.
The singular faces are marked by bold lines as displayed in Figure~\ref{image_17}. There is no possibility
to get more than three singular faces of $G$. The pair of iso-nodes corresponding to singular faces
of $G$ in case $(a2)$ and $(a4)$ lies on regular face diagonals of $G$. But in case $(a3)$ the pair of
iso-nodes corresponding to singular faces of $G$ lies on a diagonal of the cuboid of $G$. Furthermore,
the rule $T_1$ forbids the possibility of an iso-node being on two distinct singular faces, which
proves the last claim.
\begin{figure}[!ht]
\includegraphics[width=0.8\linewidth]{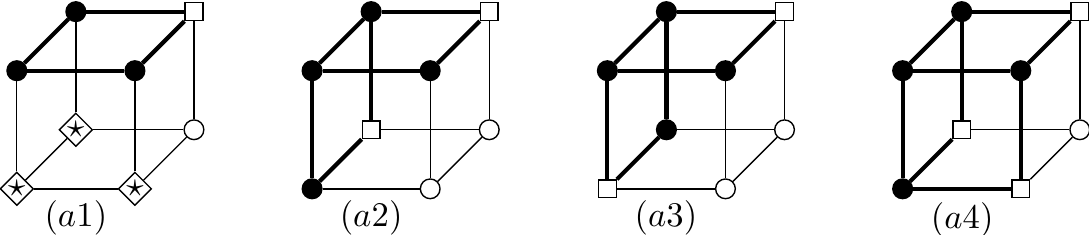}
\caption{Sketches $(a1)$, $(a2)$, $(a3)$ and $(a4)$ illustrate the possibilities of a labeled
cuboid graph to have singular faces.}
\label{image_17}
\end{figure}
\end{proof}
\FloatBarrier

\subsection{Iso-path Computation Rules}
In this section we give rules for iso-path computation for a labeled cuboid graph $G(V,E,\mathcal{F})$ with
iso-level $c\in (0,1)$. We consider two types of rules which are called $S$-rules and $C$-rules.
$S$-rules compute iso-paths in $G$, whereas $C$-rules compute iso-lines in $G$. By combining $C$- and
$S$-rules we get the iso-paths in $G$. Furthermore, consecutive application of $C$-rules to a labeled
cuboid graph $G$, on which no $S$-rules apply gives an additional iso-path in $G$.

In this section we consider not only labeled cuboid graphs but as well iso-points, iso-lines and iso-paths.
We also give graphical sketches to illustrate for a given labeled cuboid graph the corresponding iso-points,
iso-lines and iso-paths. Figure~\ref{image_18} illustrates that for $G(V,E,\mathcal{F})\in [(a1)]_{\square}$
the iso-points on the edges are marked by the symbol $\square$ as shown in sketch $(a2)$, and the iso-lines
that connect two iso-points on a face are marked by $\small{\boxempty}$\!\textminus\!\textminus\!$\small{\boxempty}$
as shown in sketch $(a3)$. The simple closed path shown in sketch $(a3)$ is an iso-path of $G$.
\begin{figure}[!ht]
\includegraphics[width=0.6\linewidth]{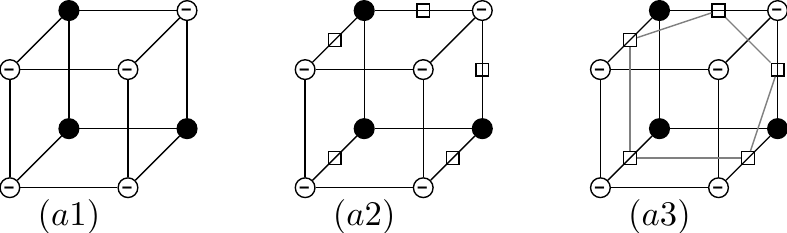}
\caption{Sketches $(a1)$, $(a2)$ and $(a3)$ illustrate the steps of iso-path computation.}
\label{image_18}
\end{figure}
\FloatBarrier
\subsubsection{$S$-subgraphs, $S$-cuboid graphs and $S$-rules}
Let $G(V,E,\mathcal{F})$ be a labeled cuboid graph with iso-level $c\in (0,1)$. The $S$-rules are denoted
by $S_1,S_2,S_3$ and will be applied on $G(V,E,\mathcal{F})$ for computing and deleting iso-paths of
$G$. They are graph operations which will be applied on so-called $S$-subgraphs of $G$ as described below.

\begin{definition}($S_1$-subgraph). Let $G(V,E,\mathcal{F})$ be a labeled cuboid graph
with iso-level $c\in (0,1)$. Let $H(V_h,E_h,\mathcal{F}_h)$ be a subgraph of $G$ such that the following
conditions hold:
\begin{enumerate}
\item the number of points in $V_h$ is four ($|V_h|=4$),
\item there exist a point $P'\in V_h$ such that $P'$ is incident to all points $P\in V_h\setminus\{P'\}$ in $H$ and
\[
\mathcal{F}_h(P')>c\qquad\mbox{and } \qquad \mathcal{F}_h(P)\leq c \quad\;\forall P\in V_h\setminus\{P'\}.
\]
\end{enumerate}
Then $H$ is called an $S_1$-subgraph.
\label{def:iso-path-10}
\end{definition}

\begin{definition}($S_2$-subgraph). Let $G(V,E,\mathcal{F})$ be a labeled cuboid graph
with iso-level $c\in (0,1)$. Let $H(V_h,E_h,\mathcal{F}_h)$ be a subgraph of $G$ such that the following
conditions hold:
\begin{enumerate}
\item the number of points in $V_h$ is six ($|V_h|=6$),
\item there exist two points $P_1,P_2\in V_h$ such that $P_1$ is incident to $P_2$ in $H$,
each $P_1$ and $P_2$ are incident in $H$ to three points in $V_h$ and
\[
\mathcal{F}_h(P_1)>c, \quad\mathcal{F}_h(P_2)>c\qquad\mbox{and }\qquad
\mathcal{F}_h(P)\leq c \quad\;\forall P\in V_h\setminus\{P_1,P_2\}.
\]
\end{enumerate}
Then $H$ is called an $S_2$-subgraph.
\label{def:iso-path-11}
\end{definition}

\begin{definition}($S_3$-subgraph). Let $G(V,E,\mathcal{F})$ be a labeled cuboid graph
with iso-level $c\in (0,1)$. Let $H(V_h,E_h,\mathcal{F}_h)$ be a subgraph of $G$ such that the following
conditions hold:
\begin{enumerate}
\item the number of points in $V_h$ is four ($|V_h|=4$),
\item there exist a point $P'\in V_h$ such that $P'$ is incident to all points $P\in V_h\setminus\{P'\}$ in $H$ and
\[
\mathcal{F}_h(P')<c\qquad\mbox{and } \qquad \mathcal{F}_h(P)> c \quad\;\forall P\in V_h\setminus\{P'\}.
\]
\end{enumerate}
Then $H$ is called an $S_3$-subgraph.
\label{def:iso-path-12}
\end{definition}
\noindent $S_1$-, $S_2$- and $S_3$-subgraphs are subsumed as $S$-subgraphs.

In Figure~\ref{image_19_hier}, sketches $(a)$, $(b)$ and $(c)$ represent $[(a)]_{\circ}$,
$[(b)]_{\circ}$ and $[(c)]_{\square}$, respectively. The $S_1$-, $S_2$- and $S_3$-subgraphs
are elements of $[(a)]_{\circ}$, $[(b)]_{\circ}$ and $[(c)]_{\square}$, respectively.
\begin{figure}[!ht]
\includegraphics[width=0.6\linewidth]{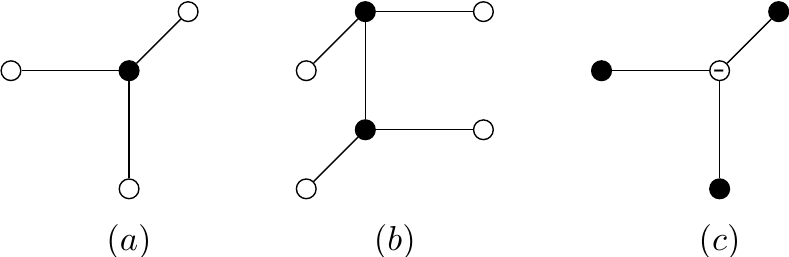}
\caption{The sketches $(a)$, $(b)$ and $(c)$ illustrate $S_1$-, $S_2$- and  $S_3$-subgraphs,
respectively.}
\label{image_19_hier}
\end{figure}
\FloatBarrier
\begin{definition}($S$-cuboid graphs). Let $G(V,E,\mathcal{F})$ be a labeled cuboid graph
with iso-level $c\in (0,1)$. Let $H(V_h,E_h,\mathcal{F}_h)$ be an $S_1$-subgraph of $G$. Then by
using Definition~\ref{def:labeled-graph-3-1} we get a labeled cuboid graph $G'|_{H'\rightarrow H}$,
where $H'$ and $G'$ are defined as given by the Definition~\ref{def:labeled-graph-3-1}. Then
we call $G'|_{H'\rightarrow H}$ the $S_1$-cuboid graph corresponding to $H$. If $H$ is an
$S_2$-subgraph of $G$ then we call $G'|_{H'\rightarrow H}$ the $S_2$-cuboid graph corresponding to $H$.
Analogously, in case $H$ is an $S_3$-subgraph of $G$ using Definition~\ref{def:labeled-graph-3-2}
we get a labeled cuboid graph $G'|_{H'\rightarrow H}$, where $H'$ and $G'$ are defined as given
by the Definition~\ref{def:labeled-graph-3-2}. Then we call $G'|_{H'\rightarrow H}$ the $S_3$-cuboid graph
corresponding to $H$.
\label{def:iso-path-10-11-12}
\end{definition}
\noindent $S_1$-, $S_2$- and $S_3$-cuboid graphs are subsumed as $S$-cuboid graphs.

In Figure~\ref{image_20}, sketches $(a)$, $(b)$ and $(c)$ represent $[(a)]_{\circ}$,
$[(b)]_{\circ}$ and $[(c)]_{\square}$, respectively. The $S_1$-, $S_2$- and $S_3$-cuboid graphs
are elements of $[(a)]_{\circ}$, $[(b)]_{\circ}$ and $[(c)]_{\square}$, respectively.
\begin{figure}[!ht]
\includegraphics[width=0.6\linewidth]{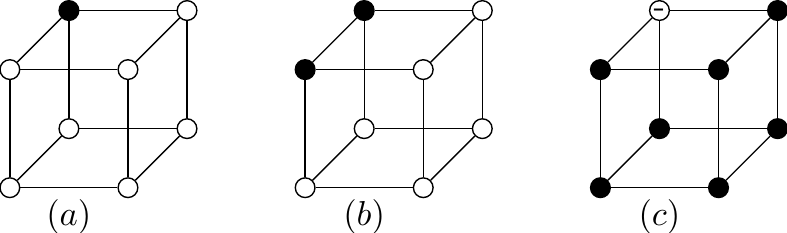}
\caption{The sketches $(a)$, $(b)$ and $(c)$ illustrate $S_1$-, $S_2$- and  $S_3$-cuboid graphs,
respectively.}
\label{image_20}
\end{figure}
\FloatBarrier
\begin{definition}(Subgraph with iso-path).
Let $H(V_h,E_h,\mathcal{F}_h)$ be an $S_1$- or $S_3$-subgraph and let $c\in (0,1)$ be an iso-level.
Let $V_{iso}=\{P\in e\,:\,e\in E_h,\; P\mbox{ an iso-point}\}$ be the set of iso-points
corresponding to $H$. Let $E_{iso}=\{e=\overline{P_iP_j}\,:\,P_i,P_j\in V_{iso}\mbox{ and } P_i\neq P_j\}$.
Then we define a labeled graph $\tilde{H}(\tilde{V}_h,\tilde{E}_h,\tilde{\mathcal{F}}_h)$, where
$\tilde{V}_h=V_h\cup V_{iso}$, $\tilde{E}_h=E_h\cup E_{iso}$ and
\begin{equation}
 \tilde{\mathcal{F}}_h=\left\{
   \begin{array}{ll}
     \mathcal{F}_h & \mbox{on }\;\; V_h, \\
     c & \mbox{on }\;\; V_{iso}.
   \end{array}\right.
\label{eq:iso-path-5}
\end{equation}
We call $\tilde{H}$ an $S_1$- or $S_3$-subgraph with iso-path if $H$ is an $S_1$- or $S_3$-subgraph,
respectively.
\label{def:iso-path-13}
\end{definition}

\begin{definition}(Subgraph with iso-path).
Let $H(V_h,E_h,\mathcal{F}_h)$ be an $S_2$-subgraph and let $c\in (0,1)$ be an iso-level.
Let $V_{iso}=\{P\in e\,:\,e\in E_h\; P\mbox{ an iso-point}\}$ be the set of iso-points
corresponding to $H$. Let $E_{iso}=\{e=\overline{P_iP_j}\,:\,P_i,P_j\in V_{iso}$ with $P_i\neq P_j$ and
$P_1,\ldots,P_4$ are cyclically ordered\}. Then we define a labeled graph
$\tilde{H}(\tilde{V}_h,\tilde{E}_h,\tilde{\mathcal{F}}_h)$, where $\tilde{V}_h=V_h\cup
V_{iso}$, $\tilde{E}_h=E_h\cup E_{iso}$ and
\begin{equation}
 \tilde{\mathcal{F}}_h=\left\{
   \begin{array}{ll}
     \mathcal{F}_h & \mbox{on }\;\; V_h, \\
     c & \mbox{on }\;\; V_{iso}.
   \end{array}\right.
\label{eq:iso-path-6}
\end{equation}
We call $\tilde{H}$ an $S_2$-subgraph with iso-path.
\label{def:iso-path-14}
\end{definition}
\noindent $S_1$-, $S_2$- and $S_3$-subgraphs with iso-path are subsumed as $S$-subgraphs with iso-path.
In Figure~\ref{image_21}, sketches $(a')$, $(b')$ and $(c')$ represent $S_1$-, $S_2$- and
$S_3$-subgraphs with iso-paths, respectively.
\begin{figure}[!ht]
\includegraphics[width=0.6\linewidth]{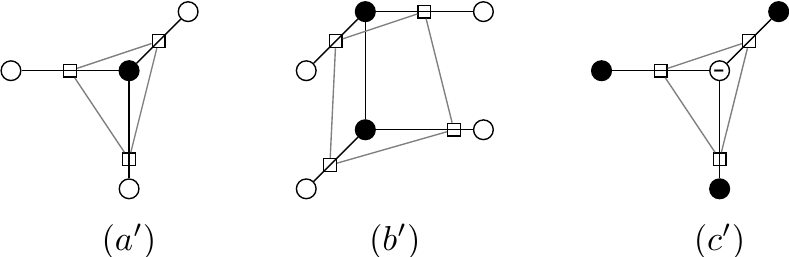}
\caption{Sketches $(a')$, $(b')$ and $(c')$ illustrate $S_1$-, $S_2$- and $S_3$-subgraphs with
iso-paths, respectively.}
\label{image_21}
\end{figure}
\FloatBarrier
In the following three definitions we use the functions $q_0$ and $q_1$ defined by \eqref{eq:iso-path-1} and
\eqref{eq:iso-path-2}.
\begin{definition}(Iso-path free subgraph).
Let $H(V_h,E_h,\mathcal{F}_h)$ be an $S_i$-subgraph $(i=1,2,3)$ and let $c\in (0,1)$ be an iso-level.
Then we call the labeled graph $H(V_h,E_h,\mathcal{R}_h)$ an iso-path free $S_i$-subgraph, where
$\mathcal{R}_h=q_0\circ\mathcal{F}_h$ in case $i=1,2$ and $\mathcal{R}_h=q_1\circ\mathcal{F}_h$ if
$i=3$. Iso-path free $S_1$-, $S_2$- and $S_3$-subgraphs are subsumed as iso-path free $S$-subgraphs.
\label{def:iso-path-16-17-18}
\end{definition}

In Figure~\ref{image_22}, sketches $(a'')$, $(b'')$ and $(c'')$ illustrate the iso-path
free \mbox{$S_1$-,} $S_2$- and $S_3$-subgraphs which are contained in $[(a'')]_{\circ}$, $[(b'')]_{\circ}$
and $[(c'')]_{\circ}$, respectively.
\begin{figure}[!ht]
\includegraphics[width=0.6\linewidth]{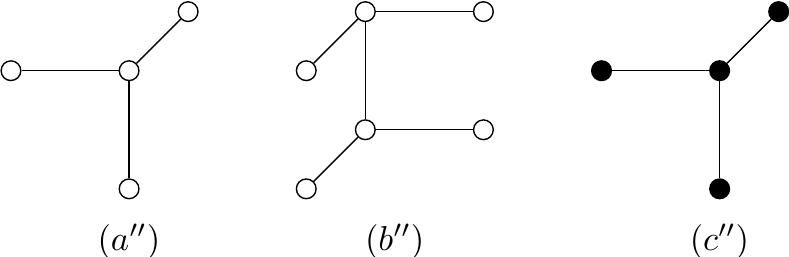}
\caption{Sketches $(a'')$, $(b'')$ and $(c'')$ illustrate iso-path free $S_1$-, $S_2$-
and $S_3$-subgraphs, respectively.}
\label{image_22}
\end{figure}
\FloatBarrier
Let $H(V_h,E_h,\mathcal{F}_h)$ be an $S$-subgraph of $G(V,E,\mathcal{F})$ and let $c\in (0,1)$ be
an iso-level. Then there exists an iso-path of $G$, since we have a corresponding $S$-subgraph with
iso-path. This iso-path is as well an iso-path of $G$. In the overall iso-surface construction,
the iso-paths that we get from $S$-graphs will be recorded in a list. Thereafter, the $S$-subgraph
is substituted in $G$ with an iso-path free $S$-subgraph. Hence we get from $G$ a new labeled cuboid
graph $G'(E,V,\mathcal{F}')$ with a new labeling and a reduced number of iso-paths. The complete
procedure of computing an iso-path of $G$ corresponding to an $S$-subgraph, recording the corresponding
iso-path and then substituting the $S$-subgraph in $G$ with an iso-path free subgraph is called
an $S$-rule. These $S$-rules are also called $S_1$-, $S_2$- and $S_3$-rules if they correspond to
the $S_1$-, $S_2$- and $S_3$-subgraphs, respectively.

The symbols used in the following two definitions have been introduced in
Definition~\ref{def:labeled-graph-4}.
\begin{definition}($S$-rule on a labeled cuboid graph). Let $G(V,E,\mathcal{F})$ be a labeled cuboid
graph with iso-level $c\in (0,1)$ and let $H(V_h,E_h,\mathcal{F}_h)$ be an $S_1$-subgraph of $G$.
Let $\tilde{H}(\tilde{V}_h,\tilde{E}_h,\tilde{\mathcal{F}}_h)$ be the $S_1$-subgraph with iso-path
corresponding to $H$. In addition, let  $\hat{H}(V_h,E_h,q_0\circ\mathcal{F}_h)$
be the iso-path free $S_1$-subgraph corresponding to $H$. Then we call the following sequence an
$S_1$-rule of $G$ corresponding to $H$:
\begin{equation}
G\longrightarrow G|_{H\rightarrow \tilde{H}}\longrightarrow G|_{H\rightarrow \hat{H}}.
\label{eq:iso-path-8}
\end{equation}
Likewise, we define for  $S_2$- and $S_3$-subgraphs the corresponding $S_2$- and $S_3$-rules on $G$,
using the labelings $q_0\circ\mathcal{F}_h$ and $q_1\circ\mathcal{F}_h$, respectively.
We subsume the $S_1$-, $S_2$- and $S_3$-rules on $G$ as $S$-rules.
\label{def:iso-path-20}
\end{definition}

\begin{definition}($S^n$-rules on a labeled cuboid graph). Let $G(V,E,\mathcal{F})$ be a labeled
cuboid graph with iso-level $c\in (0,1)$ and let $H_1(V_{h_1},E_{h_1},\mathcal{F}_{h_1})$,
$H_2(V_{h_2},E_{h_2},\mathcal{F}_{h_2})$ be $S_1$-subgraphs of $G$.
Let $\tilde{H}_1(\tilde{V}_{h_1},\tilde{E}_{h_1},\tilde{\mathcal{F}}_{h_1})$ and
$\tilde{H}_2(\tilde{V}_{h_2},\tilde{E}_{h_2},\tilde{\mathcal{F}}_{h_2})$ be $S_1$-subgraphs
with iso-paths corresponding to $H_1$ and $H_2$, respectively. In addition, let
$\hat{H}_1(V_{h_1},E_{h_1},q_0\circ\mathcal{F}_{h_1})$ and
$\hat{H}_2(V_{h_2},E_{h_2},q_0\circ\mathcal{F}_{h_2})$ be the iso-path free
$S_1$-subgraphs corresponding to $H_1$ and $H_2$, respectively. Then we call the following sequence an
$S_1^2$-rule of $G$, corresponding to $H_1$ and $H_2$:
\begin{eqnarray*}
&&G\longrightarrow G|_{H_1\rightarrow \tilde{H}_1}\longrightarrow G|_{H_1\rightarrow
\hat{H}_1}=:G'(V,E,\mathcal{F}')\\
&&G'\longrightarrow G'|_{H_2\rightarrow \tilde{H}_2}\longrightarrow G'|_{H_2\rightarrow \hat{H}_2}.
\end{eqnarray*}
Likewise, we define for $S_2$- and $S_3$-subgraphs the $S_2^2$- and $S_3^2$-rules of $G$. In analogy, we define
$S_1^3$-rules. All these $S$-rules will be subsumed as $S^n$-rules, where $n\in\{1,2,3\}$ if the type of $S$-rule
is $S_1$ or $n\in\{1,2\}$ if the type of $S$-rule is $S_2$ or $S_3$.
\label{def:iso-path-21}
\end{definition}
The $S$-rules will be written in a simplified {\it graph-theoretical rules} as displayed in
Figure~\ref{image_23_24_25}.
\begin{figure}[!ht]
$1$.
\begin{tabular}[c]{l}
\includegraphics[width=0.7\linewidth]{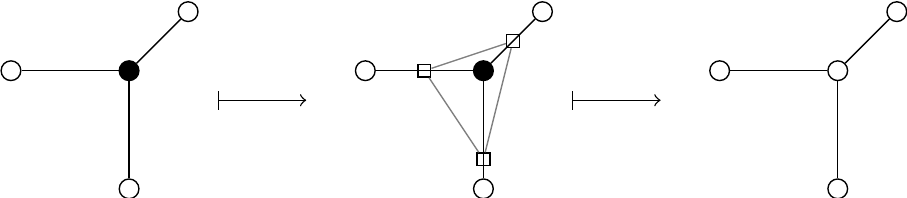}
\end{tabular}\\ \\

$2$.
\begin{tabular}[c]{l}
\includegraphics[width=0.7\linewidth]{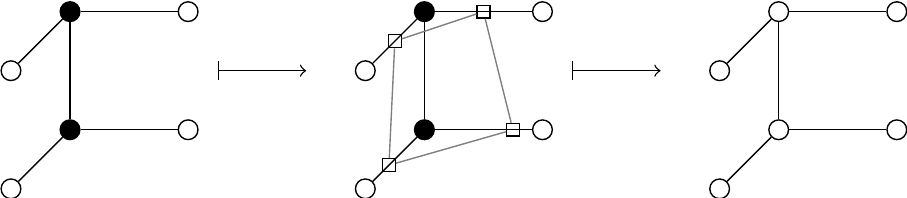}
\end{tabular}\\ \\

$3$.
\begin{tabular}[c]{l}
\includegraphics[width=0.7\linewidth]{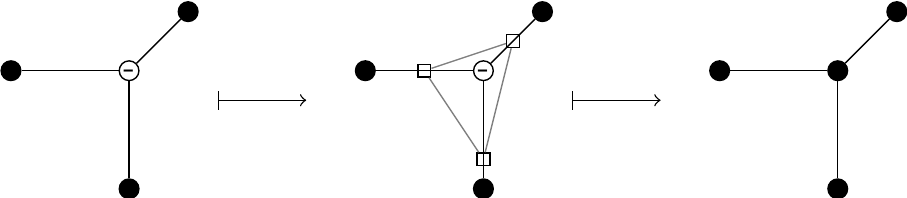}
\end{tabular}
\caption{The sequence $1$, $2$ and $3$ represent the $S_1$-, $S_2$- and $S_3$-rules, respectively.}
\label{image_23_24_25}
\end{figure}
\FloatBarrier

\subsubsection{$C$-subgraphs and $C$-rules}
In this section we describe the so-called $C$-rules. $C$-rules are graph operations which operate
on the regular faces of a labeled cuboid graph $G(V,E,\mathcal{F})$ with iso-level $c\in (0,1)$.
Applying a $C$-rule to a regular face of $G$ gives an iso-line that lies on this face.

\begin{definition}($C$-subgraphs). Let $G(V,E,\mathcal{F})$ be a labeled cuboid graph with
iso-level $c\in (0,1)$. Let $H(V_h,E_h,\mathcal{F}_h)$ be a regular face subgraph of $G$ such that
one of the following properties is satisfied:
\begin{enumerate}
\item $H$ contains one disperse node,
\item $H$ contains two disperse nodes,
\item $H$ contains three disperse nodes.
\end{enumerate}
In case $H$ satisfies 1 we call $H$ a $C_1$-subgraph. Likewise, we call $H$ a $C_2$- or a $C_3$-subgraph
if $H$ satisfies 2 or 3, respectively. We subsume the $C_1$-, $C_2$- and $C_3$-subgraphs as $C$-subgraphs.
$C$-subgraphs are regular faces and vice versa.
\label{def:iso-path-22}
\end{definition}

\begin{definition}(C-subgraphs with iso-lines and $C$-rules).
Let $H(V_h,E_h,\mathcal{F}_h)$ be a $C_1$-subgraph and let $c\in (0,1)$ be an iso-level.
Let $V_{iso}=\{P\in e\,:\,e\in E_h\mbox{ and } P\mbox{ is iso-point}\}$ be the set of iso-points
corresponding to $H$. Let $E_{iso}=\{e=\overline{P_iP_j}\,:\,P_i,P_j\in V_{iso}\mbox{ and } P_i\neq P_j\}$.
Then we define a labeled graph $\tilde{H}(\tilde{V}_h,\tilde{E}_h,\tilde{\mathcal{F}}_h)$, where
$\tilde{V}_h=V_h\cup V_{iso}$, $\tilde{E}_h=E_h\cup E_{iso}$ and
\begin{equation}
 \tilde{\mathcal{F}}_h=\left\{
   \begin{array}{ll}
     \mathcal{F}_h & \mbox{on }\;\; V_h, \\
     c & \mbox{on }\;\; V_{iso}.
   \end{array}\right.
\label{eq:iso-path-10}
\end{equation}
We call $\tilde{H}$ a $C_1$-subgraph with iso-line. Analogously, we define $C_2$- and
$C_3$-subgraphs with iso-lines in case $H$ is a $C_2$-subgraph or a $C_3$-subgraph, respectively.
We subsume these subgraphs as C-subgraphs with iso-lines. In addition, we call the transformation
process which takes $H$ to a $C_k$-subgraph with iso-line a $C_k$-rule $(k=1,2,3)$. The $C_1$-, $C_2$-
and $C_3$-rules are subsumed as $C$-rules.
\label{def:iso-path-23}
\end{definition}
\noindent The $C$-rules will be written in a simplified {\it graph-theoretical rules} as shown in
Figure~\ref{image_26_27_28}.
\begin{figure}[!ht]
$1$.
\begin{tabular}[c]{l}
\includegraphics[width=0.7\linewidth]{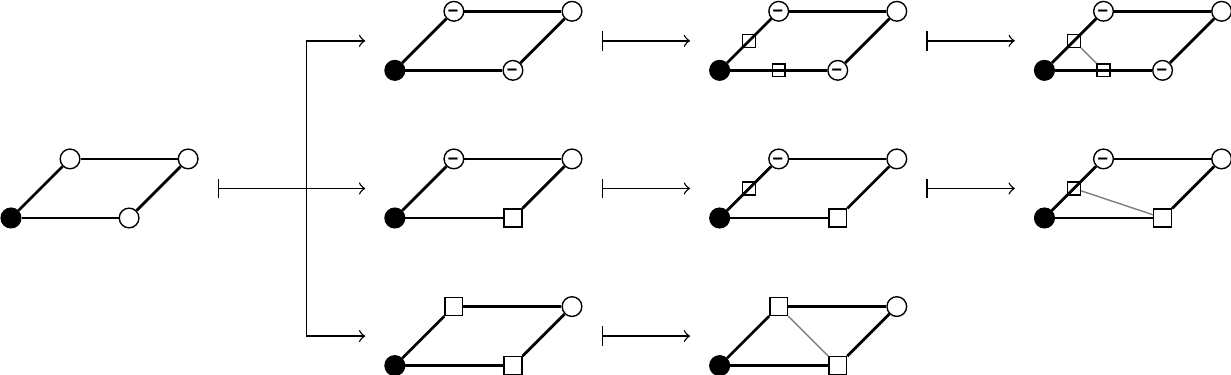}
\end{tabular}\\ \\ \\

$2$.
\begin{tabular}[c]{l}
\includegraphics[width=0.7\linewidth]{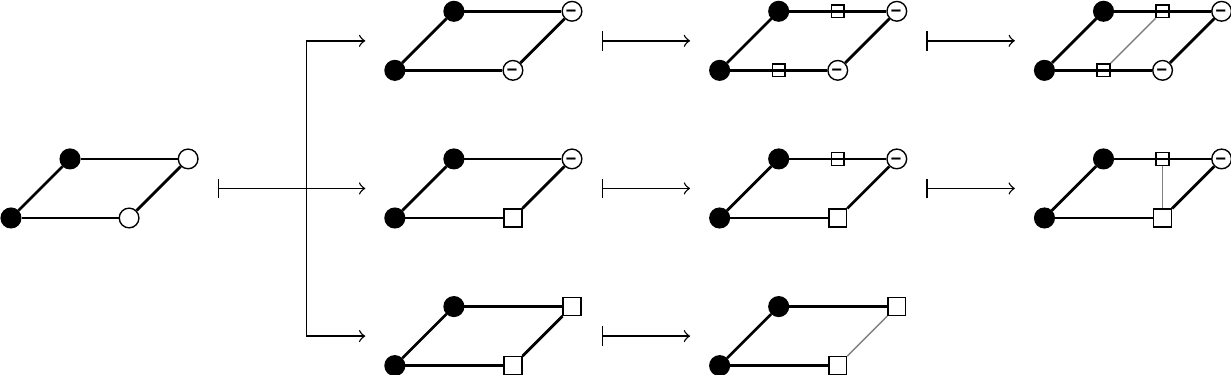}
\end{tabular}\\ \\ \\

$3$.
\begin{tabular}[c]{l}
\includegraphics[width=0.55\linewidth]{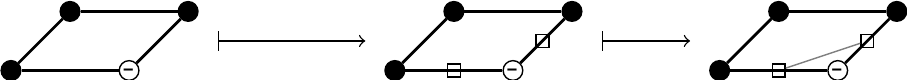}
\end{tabular}
\caption{The sequence $1$, $2$ and $3$ illustrate the $C_1$-, $C_2$- and $C_3$-rules, respectively.}
\label{image_26_27_28}
\end{figure}
\FloatBarrier
\noindent{\bf Note: }A singular face contains only one iso-point and hence it is not possible
to compute an iso-line on it. Therefore, $C$-rules will not be applied to singular faces.
\begin{proposition}
The application of any of the  $C$-rules as well as any of the $S$-rules to a labeled cuboid
graph $G(V,E,\mathcal{F})$ with iso-level $c\in (0,1)$ gives the same iso-lines on the
regular faces of $G$.
\label{prop:iso-path-2}
\end{proposition}
\begin{proof}
Both $C$- and $S$-rules compute iso-lines by connecting iso-points on the same
face of $G(V,E,\mathcal{F})$. Since a regular face of $G$ has only two iso-points
we only get one iso-line on the face. Therefore, the $C$- and $S$-rules give the same
iso-line on a regular face $G$.
\end{proof}

\begin{notation}(Corresponding node or nodes of iso-line and corresponding iso-path). Let $G(V,E,\mathcal{F})$
be a labeled cuboid graph with iso-level $c\in (0,1)$ and let $G$ be regular. Let $H(V_h,E_h,\mathcal{F}_h)$
be a regular face of $G$, having one or two or three disperse nodes (cf. sketches $(a1)$, $(a2)$ and $(a3)$ of
Figure~\ref{image_29}). Note that in case $H$ has three disperse nodes then the fourth (continuous) node
of $H$ is not allowed to be an iso-node, since otherwise $H$ is a singular face.
Let $l$ be the iso-line of $H$ which is computed by applying either the $C_1$-, or $C_2$- or the
$C_3$-rule to $H$. The types of the $C$-rules which will be applied to $H$ are chosen according to
the graph-theoretical rules as given in Figure~\ref{image_26_27_28}. Then we say that the disperse
node or nodes of $H$ correspond to the iso-line $l$.

Now, let $H'(V'_h,E'_h,\mathcal{F}_h')$ be a non-trivial L-face of $G$ (cf. sketches $(a4)$ and $(a5)$ of
Figure~\ref{image_29}). Let $P$ be one of the continuous nodes of $H'$ and let $P$ be not an iso-node.
Note that at least one continuous node of a non-trivial L-face is not an iso-node, since otherwise $H$ is a
trivial L-face. By joining the iso-points in $H'$ that are incident to $P$ in $H'$, we get an iso-line $l$ and we
say that the continuous node $P$ corresponds to the iso-line $l$. Moreover, if there exists an inner
iso-path $\omega$ of $G$ that passes through $l$ we say that $\omega$ corresponds to $P$. Additionally, let $Q$
be one of the disperse nodes of $H'$. By joining the iso-points in $H'$ that are incident to $Q$ in $H'$, we get
an iso-line $l$ and we say that the disperse node $Q$ corresponds to the iso-line $l$. Furthermore,
if there exists an inner iso-path $\omega$ of $G$ that passes through $l$ we say that $\omega$ corresponds to $Q$.
\label{note:corresponding-node}
\end{notation}
The sketches $(a1)$, $(a2)$ and $(a3)$ in Figure~\ref{image_29} show cases with an iso-line on a regular face. In
case $(a1)$ the iso-line on the face corresponds to the disperse node of the face and in the other two cases
each of the iso-lines on the faces corresponds to the disperse nodes of the face. The Sketches $(a4)$ and
$(a5)$ in Figure~\ref{image_29} show non-trivial L-faces with possible iso-lines. Iso-lines corresponding to
the continuous node (denoted by $\circleddash$) are drawn bold for $(a4)$ and $(a5)$, while iso-lines
corresponding to disperse nodes are drawn light.
\begin{figure}[!ht]
\includegraphics[width=0.99\linewidth]{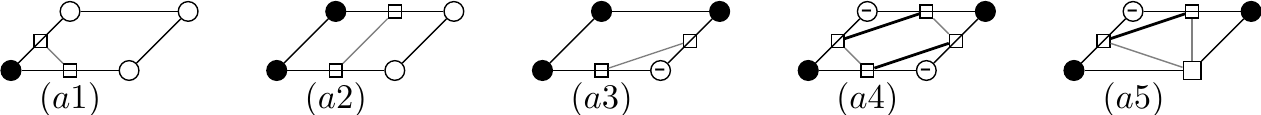}
\caption{Sketches $(a1)$, $(a2)$ and $(a3)$ show cases with an iso-line on a regular face.
Sketches $(a4)$ and $(a5)$ show cases with iso-lines on non-trivial L-faces.}
\label{image_29}
\end{figure}
\FloatBarrier

\section{Classification of Labeled Cuboid Graphs and\\ Computation of Iso-paths}
The primary objective of this section is to find a correspondence between iso-paths and subgraphs
of a labeled cuboid graph $G(V,E,\mathcal{F})$ with iso-level $c\in (0,1)$. The secondary objective
of this section is to find the set of subgraphs of $G$ that correspond to the complete set of iso-paths
of $G$, without knowing the iso-paths of $G$. Finally, we give an algorithm to find the set of subgraphs
of $G$ that correspond to the complete set of {\it inner iso-paths} of $G$, where an {\it inner iso-path}
is  an iso-path of $G$ which does not lie on a single non-trivial L-face of $G$.

For these purposes we define three different types of labeled subgraphs of $G$ as follows:
\begin{enumerate}
\item subgraphs of surface measure zero which do not lie on a face of the cuboid of $G$,
\item subgraphs of positive surface measure which do not lie on a face of the cuboid of $G$,
\item L-face subgraphs, lying on a face of the cuboid of $G$.
\end{enumerate}

Subgraphs of $G$ of surface measure zero which do not lie on a face of the cuboid of $G$ correspond to an
iso-element without surface area. The possible surface measure zero subgraphs of $G$
are illustrated by the sketches $(a)$ and $(b)$ in Figure~\ref{image_30}. We denote by $\hat{g}_1$ and
$\hat{g}_2$ arbitrary subgraphs contained in the $\square$-equivalence classes $[(a)]_{\square}$ and
$[(b)]_{\square}$, respectively. The subgraphs $\hat{g}_1$ and $\hat{g}_2$ are denoted {\it basic
zero subgraphs} of a labeled cuboid graph, since any surface measure zero subgraph of $G$ contains $\hat{g}_1$
or $\hat{g}_2$ as a subgraph.
\begin{figure}[!ht]
\includegraphics[width=0.37\linewidth]{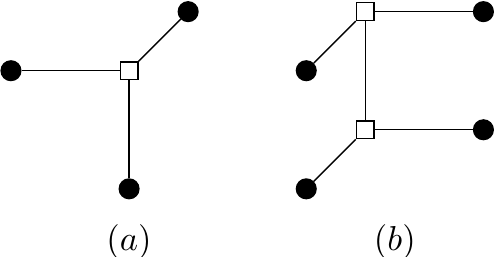}
\caption{Sketches $(a)$ and $(b)$ illustrate basic zero subgraphs of a labeled cuboid graph. The
labeled subgraphs corresponding to $(a)$ and $(b)$ are denoted by $\hat{g}_1$ and $\hat{g}_2$,
respectively.}
\label{image_30}
\end{figure}
\FloatBarrier
Subgraphs of $G$ of positive surface measure which do not lie on a face of the cuboid of $G$ correspond
to iso-elements with positive surface area. If we change the disperse nodes of $\hat{g}_1$ and
$\hat{g}_2$ to continuous nodes and the iso-node of $\hat{g}_1$ and the iso-nodes of $\hat{g}_2$ to disperse
nodes then we get labeled graphs denoted by $g_1$ and $g_2$, respectively. The labeled graphs $g_1$ and $g_2$
are contained in the $\circ$-equivalence classes $[(a)]_{\circ}$ and $[(b)]_{\circ}$, respectively,
where sketches $(a)$ and $(b)$ are shown in Figure~\ref{image_19}. If we change the iso-node of graph
$\hat{g}_1$ to a label less than the iso-value $c$ then we get a labeled graph denoted by $g_3$ which
is contained in the $\square$-equivalence class $[(c)]_{\square}$, where sketch $(c)$ is given in
Figure~\ref{image_19}.
\begin{figure}[!ht]
\includegraphics[width=0.6\linewidth]{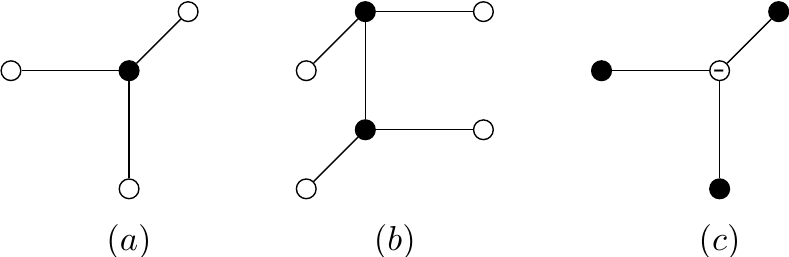}
\caption{The sketches $(a)$, $(b)$ and  $(c)$ illustrate basic positive subgraphs of a labeled cuboid graph.
The labeled subgraphs corresponding to $(a)$, $(b)$ and  $(c)$ are denoted by $g_1$, $g_2$ and $g_3$,
respectively.}
\label{image_19}
\end{figure}
\FloatBarrier
The labeled graphs $g_1$, $g_2$ and $g_3$ have positive surface measure and are called {\it basic positive subgraphs}
of a labeled cuboid graph. Here "{\it basic positive subgraphs}" means that the surface measure that we get
from the labeled graphs $g_1$, $g_2$ and $g_3$ is positive and they have the smallest number of edges compared to
other positive surface measure subgraphs which do not lie on a face of a labeled cuboid graph.

Now we give a definition which characterizes the positive surface measure subgraphs of a labeled cuboid graph
which contain a single iso-path of the graph which does not lie on a single non-trivial L-face of the graph.
\begin{definition}(Reduced positive surface measure subgraph). Let $G(V,E,\mathcal{F})$ be a regular
labeled cuboid graph with iso-level $c\in (0,1)$. Then we call the subgraph $H(V_h,E_h,\mathcal{F}_h)$
of $G$ a reduced positive surface measure subgraph if either $H$ is in $[g_1]_{\circ}$, or if $H$ has at
least two disperse nodes and satisfies the following conditions:
\begin{enumerate}
\item $H$ contains all incidence relations present in $G$ between the disperse nodes of
$G$,
\item for any two different disperse nodes of $H$ there exists a path which connects both disperse nodes such that
the path passes only through disperse edges of $H$,
\item all continuous nodes of $G$ which are incident to the disperse nodes of $H$ are in $H$ and these
are the only continuous nodes in $H$,
\item any edge $e\in E_h$ has at least one disperse node as an end point,
\item $H$ contains no L-faces,
\item $H$ does not contain two different subgraphs of $G$ which are in $[g_3]_{\square}$.
\end{enumerate}
\label{def:positive-surface-measure-graph}
\end{definition}
The labeled graphs contained in the $\circ$-equivalence classes $[(a1)]_{\circ},\ldots,[(a4)]_{\circ}$,
where sketches $(a1),\ldots,(a4)$ are given in Figure~\ref{image_31}, illustrate examples of reduced
positive surface measure subgraphs which are not basic positive subgraphs.
\begin{figure}[!ht]
\includegraphics[width=0.8\linewidth]{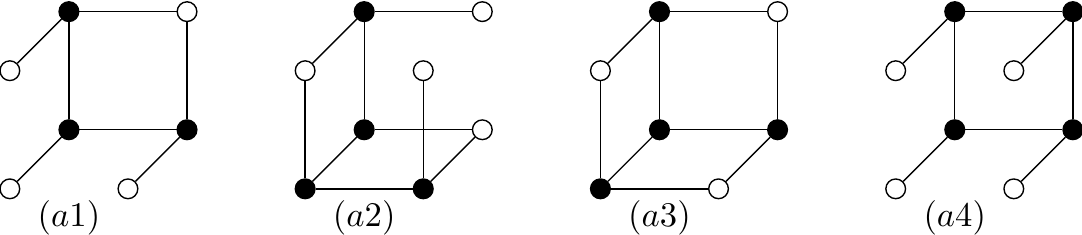}
\caption{Sketches $(a1)$ to $(a4)$ illustrate reduced positive surface measure subgraphs of a labeled
cuboid graph which are not basic and do not lie on a face.}
\label{image_31}
\end{figure}
\FloatBarrier
The existence of basic positive measure subgraphs in a labeled cuboid graph $G$ with iso-level $c\in (0,1)$
induces a classification of $G$ as reducible or irreducible as will be defined in the next subsection.

\subsection{Classification of a Labeled Cuboid Graph}
Let $G(V,E,\mathcal{F})$ be a regular labeled cuboid graph with iso-level $c\in (0,1)$. Then $G$ can
be reducible or irreducible. This classification of regular graphs is important for the computation of
iso-paths. Reducible labeled cuboid graphs will be transformed stepwise to irreducible labeled cuboid
graphs, using the $S$-rules. Each step of transformation from reducibility to irreducibility of a
labeled cuboid graph $G$ gives an iso-element of $G$ and, furthermore, an irreducible labeled cuboid
graph has a single iso-element.

Any regular labeled cuboid graph $G(V,E,\mathcal{F})$ with iso-level $c\in (0,1)$ has precisely
one of the following properties
\begin{enumerate}
\item[(a)] $G$ is L-face free,
\item[(b)] $G$ has only non-trivial L-faces,
\item[(c)] $G$ has one trivial L-face.
\end{enumerate}
If $G$ is regular and L-face free then it has one of the following forms:
\begin{enumerate}
\item there exist two disperse nodes such that each of them is incident only to continuous nodes,
\item there exist two continuous nodes such that each of them is incident only to disperse nodes,
\item all nodes satisfy none of the conditions given by 1 and 2 from above.
\end{enumerate}

\noindent Recall that $L(G)$ is the set of all L-faces of a labeled cuboid graph $G$ with iso-level
$c\in (0,1)$ and $D(G)$ denotes the number of disperse nodes of $G$.
\begin{definition}(Reducible/irreducible labeled cuboid graph). Let $G(V,E,\mathcal{F})$ be a labeled
cuboid graph with iso-level $c\in (0,1)$ and let $G$ be regular. Then we call $G$ reducible if one of
the following conditions holds:
\begin{enumerate}
\item $L(G)\neq\emptyset$,
\item $D(G)=2$ and there is $H(V_h,E_h,\mathcal{F}_h)\in [g_1]_{\circ}$ such that $H\subset G$,
\item $D(G)=6$ and there is $H(V_h,E_h,\mathcal{F}_h)\in [g_3]_{\square}$ such that $H\subset G$.
\end{enumerate}
We call $G$ irreducible if $G$ is not reducible.
\label{def:class-1}
\end{definition}
\noindent{\bf Note: }Reducible and irreducible labeled cuboid graphs are regular. A reducible
labeled cuboid graph contains at least two inner iso-paths, while an irreducible labeled cuboid graph
has one iso-path.\\

Reducible graphs will be decomposed with respect to inner iso-paths using the basic positive subgraphs.
Let $G(V,E,\mathcal{F})$ be a labeled cuboid graph with iso-level $c\in (0,1)$. Assume $G$ is reducible.
Then there are  $n\in\{1,2,3\}$ and $i\in\{1,2,3\}$ such that $S_i^n$-rule is applicable on $G$ and
\begin{equation}
\mbox{the set of inner iso-paths of }G \mbox{ is } L_1\cup L_2,
\label{eq:class-2}
\end{equation}
where $L_1$ is the set of all inner iso-paths that we get by applying the $S_i^n$-rule to $G$ and $L_2$
is the iso-path of a {\it reduced positive surface measure subgraph} $H'(V_h',E_h',\mathcal{F}_h')$
of a labeled cuboid graph $G'(V',E',\mathcal{F}')$ (with the same iso-level $c$), where $G'$ and $G$
have the same set of nodes and $H'$ contains all disperse nodes of $G'$. We call $G'$ a {\it rest graph}
of $G$. This means there exists a decomposition $\chi$ of $G$ with respect to inner iso-paths of $G$ as
\begin{equation}
\chi(G)=\left(S_{i,1},\ldots,S_{i,n},R\right),
\label{eq:class-3}
\end{equation}
where $S_{i,1},\ldots,S_{i,n}$ denote the $n$ distinct $S_i$-subgraphs of $G$ and $R$
is the rest graph of $G$ which we get after we apply the $S_i^n$-rule to $G$. The rest graph $R$ of $G$
is irreducible. In addition, we define a decomposition $\eta$ of $G$ into labeled cuboid graphs
(with the same iso-level $c$) by
\begin{equation}
\eta(G):=\left(G_1,\ldots,G_n, R\right),
\label{eq:class-3-1}
\end{equation}
where $G_l$ is the $l$-th $S_i$-cuboid graph of $G$ for $l=1,\ldots,n$ and $R$ is the rest graph of $G$.
Note that the $G_l$ are irreducible labeled cuboid graphs. By the definition of the $l$-th
$S_i$-cuboid graph of $G$, the following holds:
\begin{itemize}
\item for all $l=1,\ldots,n$, the iso-path of $G_l$ is the same as the iso-path that we get by
applying the $S_i^l$-rule to $G$.
\end{itemize}
Theoretical investigations and algorithmical computation of iso-surfaces and surface normals corresponding
inner iso-paths of $G$ are easier if we use the labeled cuboid graphs $G_l$ for $l=1,\ldots,n$ and the rest
graph $R$ of $G$ as given by the decomposition~\eqref{eq:class-3-1} instead of $G$.

\subsection{Inner Iso-paths of Labeled Cuboid Graphs}
In this section we compute inner iso-paths of a given labeled cuboid graph $G(V,E,\mathcal{F})$ with iso-level
$c\in (0,1)$. For the computation of the inner iso-paths we repeatedly refer to the labeled graphs $g_1$,
$g_2$ and $g_3$ as explained above in this section.

\begin{theorem}\label{thm:class-1}
Let $G(V,E,\mathcal{F})$ be a labeled cuboid graph with iso-level $c\in (0,1)$. Assume $G$ is regular and
has at least one L-face. Then $G$ contains a subgraph $H(V_h,E_h,\mathcal{F}_h)$ which is an element of
$[g_1]_{\circ}$ or $[g_2]_{\circ}$ or $[g_3]_{\square}$.
\end{theorem}
\begin{proof}
We consider each case from $D(G)=2$ to $D(G)=6$, separately. In the following, $G$ is an element of a
$\circ$- or $\square$-equivalence class represented by the sketch of a labeled cuboid graph as shown on
the left side of the figures below. We use as well special cases of $\square$-equivalence for $G$ as given
in Notation~\ref{not:special}.
\begin{enumerate}
\item $D(G)=2$: see Figure~\ref{image_32}. Evidently, there is $H\in [g_1]_{\circ}$.
\begin{figure}[!ht]
\includegraphics[width=0.3\linewidth]{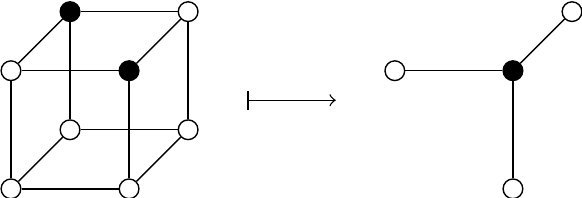}
\caption{The left part represents $[G]_{\circ}$ such that $D(G)=2$ and $L(G)\neq\emptyset$.}
\label{image_32}
\end{figure}
\FloatBarrier
\item $D(G)=3$: see Figure~\ref{image_33_34}. In case of sketch $(a)$ we have $H\in [g_1]_{\circ}$ and in case
of sketch $(b)$ we have $H\in [g_1]_{\circ}$ and $H\in [g_2]_{\circ}$.
\begin{figure}[!ht]
$(a)$
\begin{tabular}[c]{l}
\includegraphics[width=0.3\linewidth]{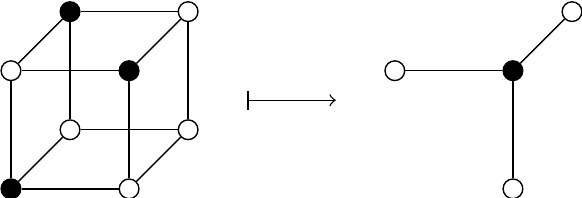}
\end{tabular}
\hspace{1cm}$(b)$
\begin{tabular}[c]{l}
\includegraphics[width=0.35\linewidth]{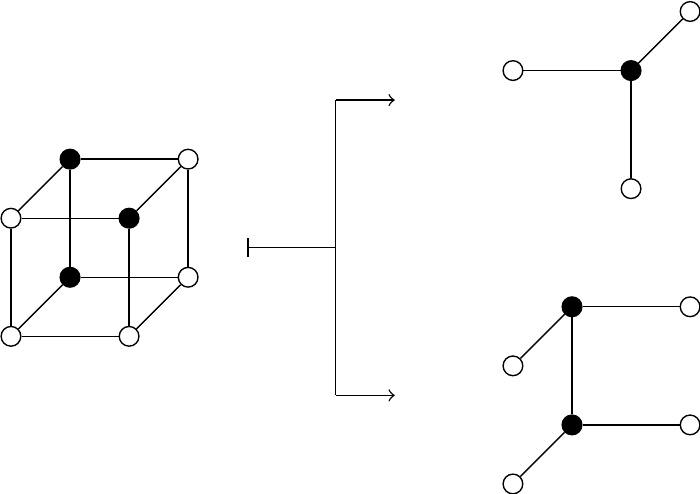}
\end{tabular}
\caption{The left side of the sketches $(a)$ and $(b)$ represent the two possibilities of $G$ such that $D(G)=3$ and
$L(G)\neq\emptyset$.}
\label{image_33_34}
\end{figure}
\FloatBarrier
\item $D(G)=4$: see Figure~\ref{image_35_36}. In case of sketch $(a)$ we have $H\in [g_1]_{\circ}$ and in case
of sketch $(b)$ we have $H\in [g_2]_{\circ}$.
\begin{figure}[!ht]
$(a)$
\begin{tabular}[c]{l}
\includegraphics[width=0.37\linewidth]{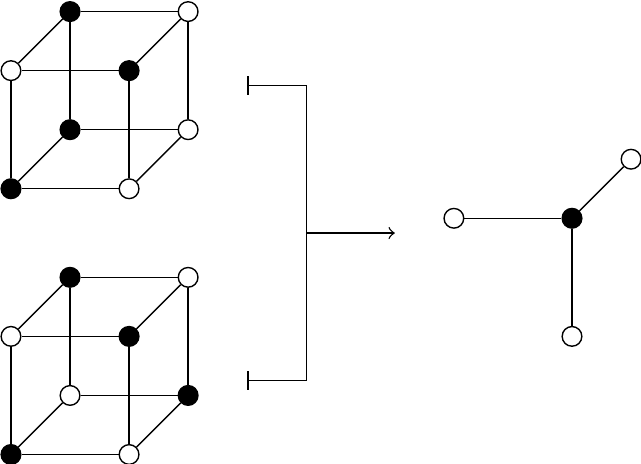}
\end{tabular}
\hspace{0.8cm}$(b)$
\begin{tabular}[c]{l}
\includegraphics[width=0.35\linewidth]{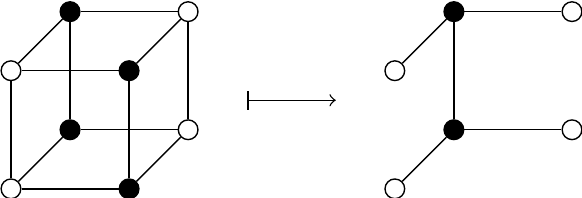}
\end{tabular}
\caption{The left side of the sketches $(a)$ and $(b)$ represent the three possibilities of $G$ such that $D(G)=4$
and $L(G)\neq\emptyset$.}
\label{image_35_36}
\end{figure}
\FloatBarrier
\item $D(G)=5$: see Figure~\ref{image_37_38}. In case of sketch $(a)$ we have $H\in [g_1]_{\circ}$ and
$H\in [g_3]_{\square}$ and in case of sketch $(b)$ we have $H\in [g_3]_{\square}$.
\begin{figure}[!ht]
$(a)$
\begin{tabular}[c]{l}
\includegraphics[width=0.35\linewidth]{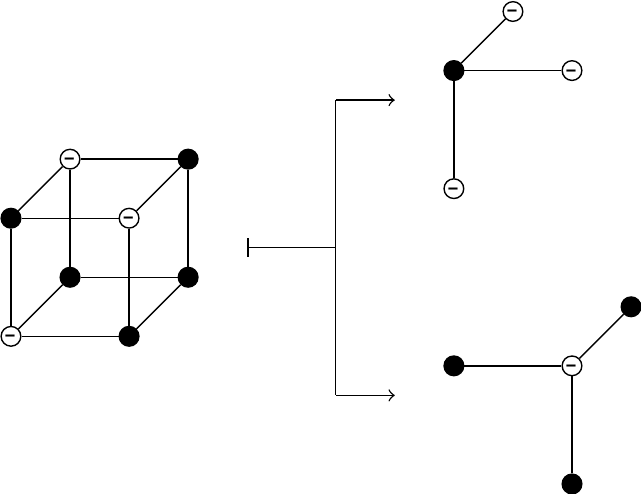}
\end{tabular}
\hspace{0.8cm}$(b)$
\begin{tabular}[c]{l}
\includegraphics[width=0.33\linewidth]{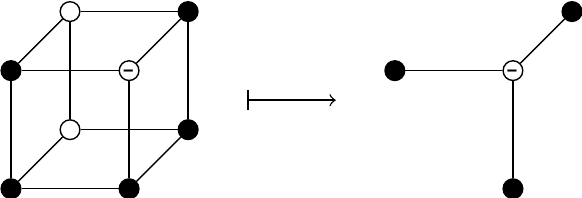}
\end{tabular}
\caption{The left side of the sketches $(a)$ and $(b)$ represent the two possibilities of $G$ such that $D(G)=5$
and $L(G)\neq\emptyset$.}
\label{image_37_38}
\end{figure}
\FloatBarrier
\item $D(G)=6$: see Figure~\ref{image_39}. It holds that $H\in [g_3]_{\square}$.
\begin{figure}[!ht]
\includegraphics[width=0.3\linewidth]{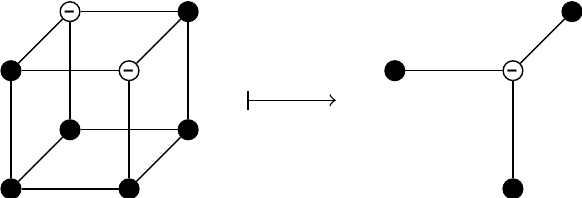}
\caption{The sketch on the left represents $[G]_{\square}$ such that $D(G)=6$ and $L(G)\neq\emptyset$.}
\label{image_39}
\end{figure}
\end{enumerate}
\vspace{-0.9cm}
\end{proof}
\FloatBarrier
\noindent {\bf Note: }A labeled cuboid graph $G(V,E,\mathcal{F})$ with iso-level $c\in (0,1)$ and $D(G)=1$
or $D(G)=7$ has no L-faces, since on any L-face there exists two disperse and two continuous nodes.\\

Propositions \ref{prop:class-1} and \ref{prop:class-2} draw consequences of Theorem~\ref{thm:class-1}.
\begin{proposition}Let $G(V,E,\mathcal{F})$ be a labeled cuboid graph with iso-level $c\in (0,1)$
and let $G$ be regular. Then, if $G$ has L-faces at all, the numbers of possible L-faces in dependence
of $D(G)$ is given in Table~\ref{table_1} .
\label{prop:class-1}
\end{proposition}

\begin{proposition}(Removing L-faces). Let $G(V,E,\mathcal{F})$ be a labeled cuboid graph with
iso-level $c\in (0,1)$. Suppose that $G$ is regular and has at least one L-face. If an $S^n$-rule is
chosen according to row three of Table~\ref{table_1} and if we apply this $S^n$-rule to $G$, we get an
L-face free labeled cuboid graph $G(V,E,\mathcal{F}')$. Here, $n\in\{1,2,3\}$ has to be chosen according
to row three of Table~\ref{table_1}. The computation of $\mathcal{F}'$ is described in
Definition~\ref{def:iso-path-21}. Here we denote by $|L(G)|$ the total number of L-faces in $G$
and, as before, $D(G)$ denotes the number of disperse nodes of $G$.
\begin{table}[h]
\begin{center}
\begin{tabular}{|l||c|c|c|p{3.5cm}|c|c|c|c|}
\hline
$D(G)$&{\bf 2}&\multicolumn{2}{|c|}{\bf 3}&\multicolumn{2}{|c|}{\bf 4}
&\multicolumn{2}{|c|}{\bf 5}&{\bf 6}\\ \hline\hline
$|L(G)|$&1&1&3&\hspace{1.6cm}2&6&1&3&1\\ \hline
$S^n$-rule&$S_1$&$S_1$&$S_1^2$&\begin{tabular}{c|c}L-faces $\nparallel$
& L-faces $\parallel$ \\
$S_1$& $S_2$\end{tabular}&$S_1^3$&$S_3$&$S_1$&$S_3$\\ \hline
$S$-rule&$S_1$&\multicolumn{2}{|c|}{$S_1$}&\begin{tabular}{c|c}$\hspace{0.6cm}S_1
\hspace{0.5cm}$&$\hspace{0.5cm}S_2$
\end{tabular}&$S_1$&$S_3$&$S_1$&$S_3$\\ \hline
\end{tabular}
\end{center}
\caption{Rules for removing L-faces. The signs $\parallel$ and $\nparallel$ correspond to
parallel and non-parallelity of the L-faces in case of two L-faces.}
\label{table_1}
\end{table}
\label{prop:class-2}
\end{proposition}
\FloatBarrier
We get the results of Proposition~\ref{prop:class-1} and~\ref{prop:class-2} given in Table~\ref{table_1}
by the same arguments as used to prove Theorem~\ref{thm:class-1} for the cases $D(G)=2$ to $D(G)=6$.
Therefore, a detailed proof is omitted.\\


\noindent {\bf Note: }From here on, when we say that an $S$-rule applies on a labeled cuboid graph
$G(V,E,\mathcal{F})$ with iso-level $c\in (0,1)$, where $G$ has at least one L-face, it is understood
that the corresponding $S$-rule is chosen according to Table~\ref{table_1}. Moreover, if we say that
we apply the $S_i$-rule $(i\in\{1,2,3\})$ to $G$, it is understood that the $S_i$-rule is chosen
according to Table~\ref{table_1}.

\begin{proposition}
Let $G(V,E,\mathcal{F})$ be a labeled cuboid graph with iso-level $c\in (0,1)$. Assume $G$ is regular with
at least one L-face and $2\leq D(G)\leq 6$. Suppose the subgraph $H(V_h,E_h,\mathcal{F}_h)$ of $G$
is an L-face. Let the $S$-rule that will be applied on $G$ be chosen according to Table~\ref{table_1}. If
the $S_1$- or $S_2$-rule is applied on $G$ then one of the two disperse nodes of $H$ is the disperse node
of an $S_1$- or $S_2$-subgraph of $G$, respectively. If the $S_3$-rule is applied on $G$ then one of the
two continuous nodes of $H$ which is not an iso-node is the continuous node of an $S_3$-subgraph of $G$.
\label{prop:class-3}
\end{proposition}
\begin{proof}
Choosing the $S$-rule for $G$ according to Table~\ref{table_1} and using the arguments used to prove
Theorem~\ref{thm:class-1} proofs the claim.
\end{proof}

\begin{theorem}\label{thm:class-2}
Let $G(V,E,\mathcal{F})$ be a labeled cuboid graph with iso-level $c\in (0,1)$. Assume $G$ is regular
and $D(G)\in \{2,6\}$. Suppose that $G$ has no L-face and let one of the following hold:
\begin{enumerate}
\item[$(a)$] for $D(G)=2$, the two disperse nodes are on the diagonal of the cuboid of $G$,
\item[$(b)$] for $D(G)=6$, the two continuous nodes are on the diagonal of the cuboid of $G$.
\end{enumerate}
Then $G$ contains two subgraphs $H(V_h,E_h,\mathcal{F}_h)$ and $H'(V_h',E_h',\mathcal{F}_h')$
which are in $[g_1]_{\circ}$ for the case $(a)$ and in $[g_3]_{\square}$ for the case $(b)$, respectively.
\end{theorem}
\begin{proof}In the following, $G$ is an element of the $\circ$-equivalence class represented by the
sketch of a labeled cuboid graph as shown on the left side of the figures below.
\begin{enumerate}
\item Case $D(G)=2$: see Figure~\ref{image_40}. Evidently, there is $H,H'\in [g_1]_{\circ}$.
\FloatBarrier
\begin{figure}[!ht]
\includegraphics[width=0.35\linewidth]{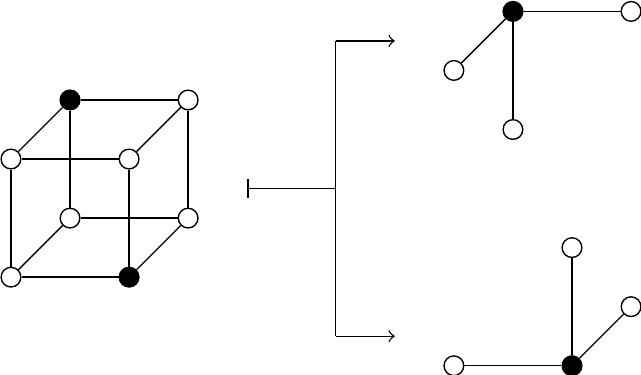}
\caption{The sketch on the left represents $[G]_{\circ}$ such that $D(G)=2$ and the two disperse nodes
are on the space diagonal of the cuboid of $G$.}
\label{image_40}
\end{figure}
\FloatBarrier
\item Case $D(G)=6$: see Figure~\ref{image_41}. Evidently, there is $H,H'\in [g_3]_{\square}$.
\begin{figure}[!ht]
\includegraphics[width=0.35\linewidth]{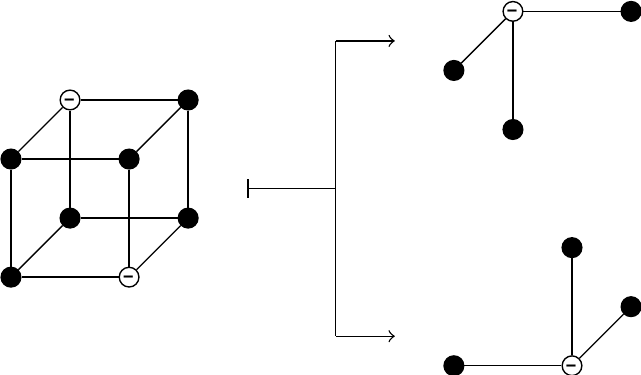}
\caption{The sketch on the left represents $[G]_{\square}$ such that $D(G)=6$ and the two
continuous nodes are on the space diagonal of the cuboid of $G$.}
\label{image_41}
\end{figure}
\FloatBarrier
\end{enumerate}
\vspace{-0.9cm}
\end{proof}
\FloatBarrier

\begin{proposition}
Let $G(V,E,\mathcal{F})$ be a labeled cuboid graph with iso-level $c\in (0,1)$ and let $G$ be irreducible.
Then the following holds:
\begin{enumerate}
\item $1\leq D(G)\leq 7$,
\item $G$ has no L-face,
\item if $D(G)=2$ or $D(G)=6$ then $G$ satisfies neither the assumption $(a)$ nor $(b)$ of Theorem~\ref{thm:class-2}.
\end{enumerate}
\label{prop:class-4}
\end{proposition}
\begin{proof}
The results follow using Table~\ref{table_1} and Theorem~\ref{thm:class-2} and the fact that $G$ is as well
irreducible in case $D(G)=1$ or $D(G)=7$.
\end{proof}

\begin{theorem}\label{thm:class-3}
Let $G(V,E,\mathcal{F})$ be a labeled cuboid graph with iso-level $c\in (0,1)$ and let $G$ contain
only one reduced positive surface measure subgraph $H(V_h,E_h,\mathcal{F}_h)$ and let $H$ contain
all disperse nodes of $G$. Then $G$ is irreducible and hence contains only one iso-path. Vice versa,
if $G$ is irreducible then there exists a subgraph of $G$ which satisfies the above stated properties
of $H$.
\end{theorem}

\begin{proof} We give the prove by computing the iso-paths of $G$ by joining the iso-points with iso-lines.
The steps of the iso-path computation are given in each case by figures with a sequence of sketches from
left to right. The symbols on the rightmost side with only disperse or only continuous nodes are used
to characterize the resulting type of inner iso-path. The sketch on the left shows the respective
labeled cuboid graph, in the second sketch the iso-points are marked and in the third sketch the iso-lines
are inserted by applying the $C$-rules.
\begin{enumerate}
\item $D(G)\neq 2$ and $D(G)\neq 6$.
\begin{itemize}
\item[$(a)$] $D(G)=1$: see Figure~\ref{image_42}.
\begin{figure}[!ht]
\includegraphics[width=0.9\linewidth]{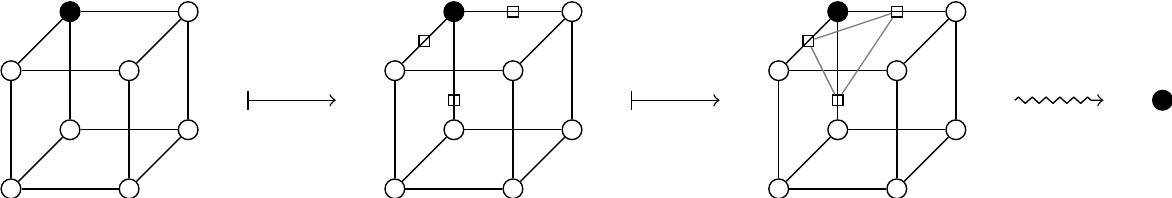}
\caption{Computation of the iso-path if $D(G)=1$.}
\label{image_42}
\end{figure}
\FloatBarrier
\item[$(b)$] $D(G)=3$: see Figure~\ref{image_43}.
\begin{figure}[!ht]
\includegraphics[width=0.9\linewidth]{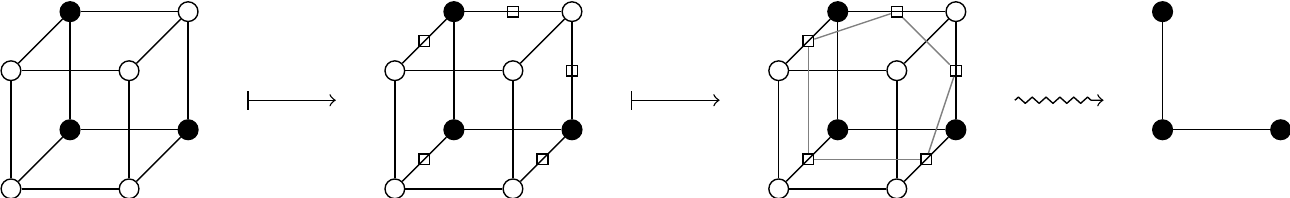}
\caption{Computation of the iso-path if $D(G)=3$ and $L(G)=\emptyset$.}
\label{image_43}
\end{figure}
\FloatBarrier
\item[$(c)$] $D(G)=4$: see Figure~\ref{image_44_45_46}.
\begin{figure}[!ht]
$(c1)$
\begin{tabular}[c]{l}
\includegraphics[width=0.8\linewidth]{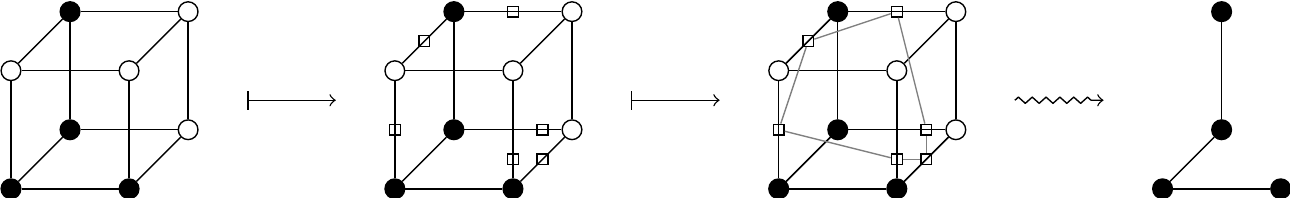}
\end{tabular}

\vspace{0.2cm}
$(c2)$
\begin{tabular}[c]{l}
\includegraphics[width=0.8\linewidth]{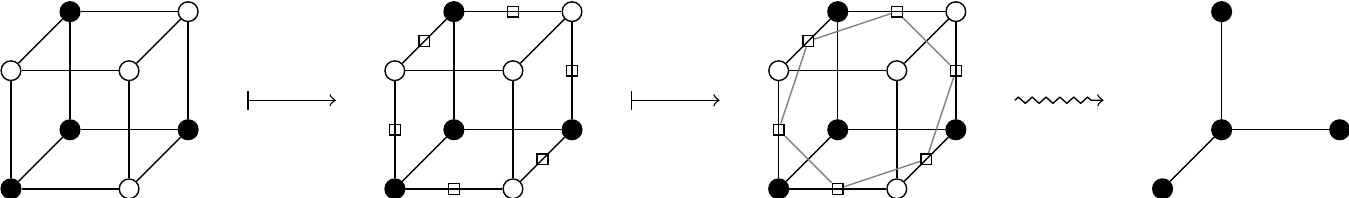}
\end{tabular}

\vspace{0.2cm}
$(c3)$
\begin{tabular}[c]{l}
\includegraphics[width=0.8\linewidth]{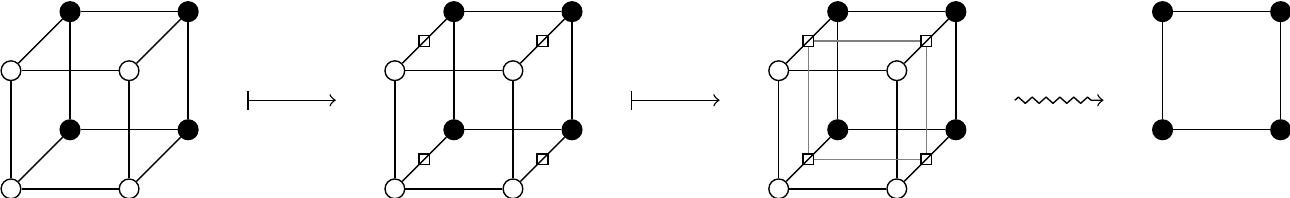}
\end{tabular}
\caption{Computation of the iso-path if $D(G)=4$ and $L(G)=\emptyset$ in three different cases as shown by
the sketches (c1), (c2) and (c3).}
\label{image_44_45_46}
\end{figure}
\FloatBarrier
\item[$(d)$] $D(G)=5$: see Figure~\ref{image_47}.
\begin{figure}[!ht]
\includegraphics[width=0.8\linewidth]{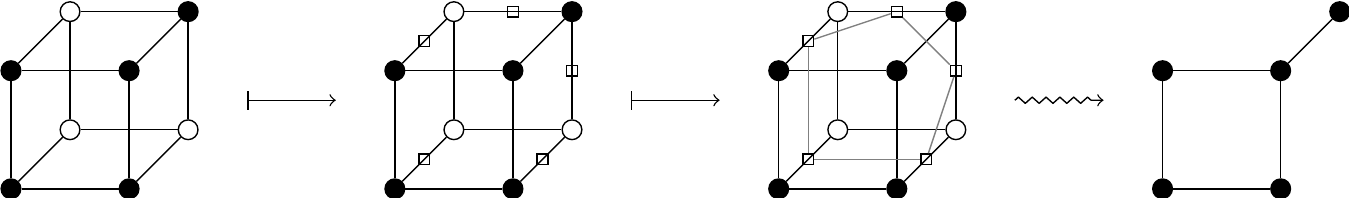}
\caption{Computation of the iso-path if $D(G)=5$
and $L(G)=\emptyset$.}
\label{image_47}
\end{figure}
\FloatBarrier
\item[$(e)$] $D(G)=7$: see Figure~\ref{image_48}.
\begin{figure}[!ht]
\includegraphics[width=0.8\linewidth]{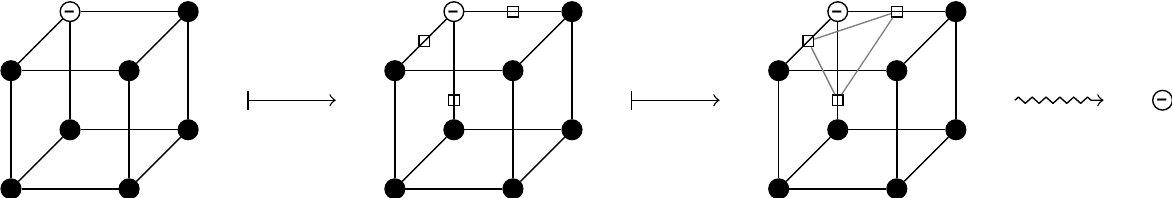}
\caption{Computation of the iso-path if $D(G)=7$.}
\label{image_48}
\end{figure}
\FloatBarrier
\end{itemize}
\item $D(G)=2$ or $D(G)=6$.
\begin{itemize}
\item[$(a)$] $D(G)=2$: see Figure~\ref{image_49}.
\begin{figure}[!ht]
\includegraphics[width=0.8\linewidth]{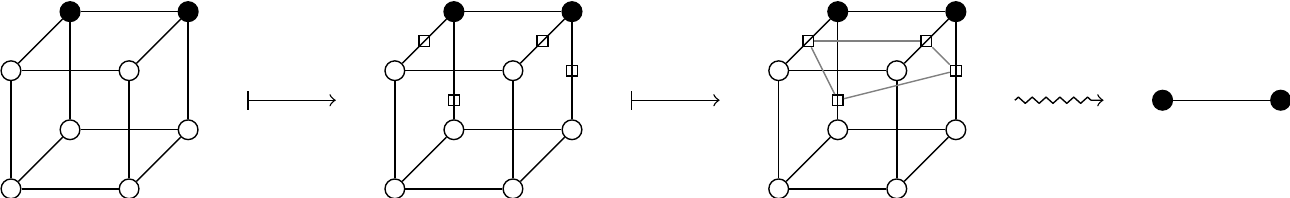}
\caption{Computation of the iso-path if $D(G)=2$ and both disperse nodes lie on the same edge of the
cuboid of $G$.}
\label{image_49}
\end{figure}
\FloatBarrier
\item[$(b)$] $D(G)=6$: see Figure~\ref{image_50}.
\begin{figure}[!ht]
\includegraphics[width=0.8\linewidth]{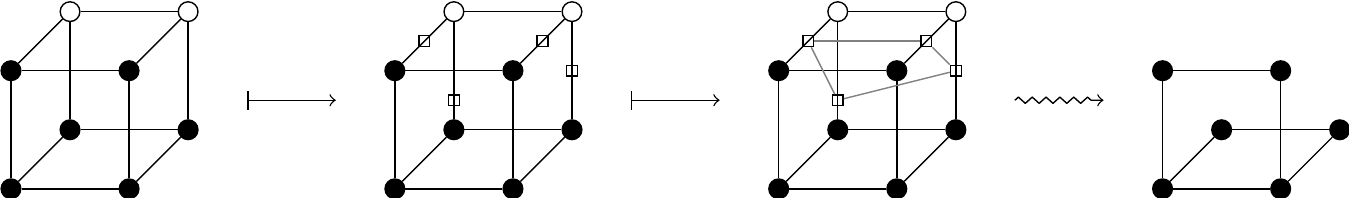}
\caption{Computation of the iso-path if $D(G)=6$ and the two continuous nodes lie on the same edge of
the cuboid of $G$.}
\label{image_50}
\end{figure}
\FloatBarrier
\end{itemize}
\end{enumerate}
The simple paths constructed above are iso-paths, since they satisfy the criteria for an
inner iso-path given in Definition~\ref{def:labeled-graph-6}. Furthermore, in all cases,
$G$ contains only one reduced positive surface measure subgraph containing all disperse nodes of
$G$. Therefore, $G$ is irreducible.

To show the last statement, observe that the graphs given in all cases above are the only irreducible
graphs of $G$ which proves the last claim.
\end{proof}

Proposition \ref{prop:class-5} draws a consequence of Theorem~\ref{thm:class-1} and~\ref{thm:class-3}.
\begin{proposition}Let $G(V,E,\mathcal{F})$ be a labeled cuboid graph with iso-level $c\in (0,1)$.
Assume $G$ is regular and has at least one L-face. Let us apply the corresponding $S^n$-rules to $G$ given
in Table~\ref{table_1} to obtain a labeled cuboid graph $G'(V,E,\mathcal{F}')$. Then $G'$ contains only one
iso-path.
\label{prop:class-5}
\end{proposition}
\begin{proof}
The application of $S^n$-rule to $G$, where the $S$-rule is chosen according to Table~\ref{table_1},
gives an irreducible graph $G'$. Hence, the irreducible graph $G'$ has a single iso-path as proven in
Theorem~\ref{thm:class-3}. The iso-path of $G'$ is as well one of the inner iso-paths of $G$.
\end{proof}

Proposition \ref{prop:class-6} draws a consequence of Theorem~\ref{thm:class-1}, \ref{thm:class-2}
and~\ref{thm:class-3}.
\begin{proposition}Let $G(V,E,\mathcal{F})$ be a labeled cuboid graph with iso-level $c\in (0,1)$ and let
$G$ be reducible. Then the successive application of one of the $S$-rules to $G$ transforms $G$
to an irreducible graph. Here the $S$-rule is chosen according to Table~\ref{table_1} if $G$ has at
least one L-face and, otherwise, the $S$-rule is chosen using the results given in Theorem~\ref{thm:class-2}.
\label{prop:class-6}
\end{proposition}
\begin{proof} First, if $G$ has at least one L-face we apply
Proposition~\ref{prop:class-5}. Second, if $G$ is L-face free, then if $D(G)=2$ then apply once the
$S_1$-rule to $G$ to get an irreducible graph $G'$ with only one disperse node. In case $D(G)=6$
apply once the $S_3$-rule to $G$ to get an irreducible graph $G'$ with seven disperse nodes.

The second result follows by the arguments used to prove Theorem~\ref{thm:class-2}.
\end{proof}

\begin{proposition}
Let $G(V,E,\mathcal{F})$ be a labeled cuboid graph with iso-level $c\in (0,1)$ and let $G$ be regular.
Suppose that the subgraph $H(V_h,E_h,\mathcal{F}_h)$ of $G$ is a regular face. Then there exists only one
iso-path of $G$ that passes through the iso-line on the regular face $H$ if one of the following
conditions is satisfied:
\begin{itemize}
\item[(a)] $G$ is irreducible,
\item[(b)] $G$ has no L-face,
\item[(c)] $G$ has one or more L-faces, but there is no subgraph $H(V_h,E_h,\mathcal{F}_h)$ of
$G$ such that $H\in [\hat{g}_2]_{\circ}$.
\end{itemize}
\label{prop:class-7}
\end{proposition}
\begin{proof} We consider two cases.\\

\noindent{\bf Case 1: }If $G$ is an irreducible graph then there exists exactly one iso-path in $G$ and the
iso-path passes through all iso-lines of regular faces of $G$ as proven by Theorem~\ref{thm:class-3}. \\

\noindent{\bf Case 2: }If $G$ is reducible then we consider two subcases.

First, if $G$ has no L-face then application of the $S_1$-rule to $G$ in case $D(G)=2$ changes the disperse
node on the face to a continuous node such that the resulting face is a continuous face. In case $D(G)=6$
application of the $S_3$-rule to $G$ changes the continuous node on the face to a disperse node such that
the resulting face is a disperse face. These results follow from the arguments used to prove
Theorem~\ref{thm:class-2} for the cases $D(G)=2$ and $D(G)=6$. Now, note that no iso-lines of $G$ lie on
disperse or continuous faces of $G$. Hence in both cases only one iso-path of $G$ passes through the iso-line
on the regular face $H$.

Second, if $G$ has at least one L-face then application of one of the $S$-rules chosen according to
Table~\ref{table_1} will decrease the number of L-faces of $G$ and never create new L-faces. Hence, there
exists only one iso-path of $G$ that passes through the iso-line on the regular face $H$. This holds,
because after computing an iso-path of $G$ that passes through an iso-line on a regular face, then the
number of disperse nodes of $G$ that lie on the regular face increase or decrease. In case the number of
disperse nodes on the regular face increases, then the resulting face is a disperse face or a singular face.
But in case the number of disperse nodes of the regular face decreases, then the resulting face is a
continuous face. Note that no iso-lines of $G$ lie on singular faces of $G$. The arguments given in the
second subcase follow from the arguments used to prove Theorem~\ref{thm:class-1} for the cases $D(G=2$
to $D(G)=6$. Hence, through the iso-line of a regular face $H$ passes only one iso-path of $G$.
\end{proof}

\vspace{0.2cm}

\begin{proposition}
Let $G(V,E,\mathcal{F})$ be a labeled cuboid graph with iso-level $c\in (0,1)$. Assume $G$ is regular and
has at least one non-trivial L-face $H(V_h,E_h,\mathcal{F}_h)$. Let the $S$-rule that will be applied on $G$
be chosen according to Table~\ref{table_1}. If the $S$-rule that will be applied on $G$ is either the
$S_1$- or $S_2$-rule then there exist two iso-lines on $H$ and to each iso-line there exists an inner
iso-path of $G$ that passes through it. If the $S_3$-rule will be applied on $G$ then one of the following
holds:
\begin{itemize}
\item[(a)] in case $H$ has no iso-nodes then there exist two iso-lines on $H$ and to each iso-line there
exists an iso-path of $G$ that passes through it such that none of the iso-paths lies on the L-face.
\item[(b)] in case one continuous node of $H$ is an iso-node then there exists one iso-line on $H$ and to the
iso-line there exists an iso-path of $G$ that passes through it such that the iso-path does not lie on the
L-face.
\end{itemize}
\label{prop:class-8}
\end{proposition}
\begin{proof} We consider two cases. \\

\noindent{\bf Case 1: }Either the $S_1$- or $S_2$-rule will be applied on $G$. Then, applying
Proposition~\ref{prop:class-3} we find that through one of the iso-lines passes an iso-path of $G$.
The $S_1$- or $S_2$-rule then changes the disperse node corresponding the iso-line to a continuous node.
Then there remain only one disperse node on the face. Hence the L-face will be transformed to a regular face
with only one disperse node. On a regular face lies only a single iso-line and through this iso-line
passes only one iso-path of $G$ according to Proposition~\ref{prop:class-7}. But $H$ is a non-trivial L-face
and, hence, both iso-paths passes through two different iso-lines on $H$ and both iso-paths are different.\\

\noindent{\bf Case 2: }The $S_3$-rule will be applied on $G$. Then, applying Proposition~\ref{prop:class-3}
we find that through one of the iso-lines passes an iso-path of $G$. The $S_3$-rule then changes one of the
continuous nodes of $H$ which is not an iso-node of $H$ to a disperse node. Then there remains only one
continuous node on the face. Hence, the L-face will be transformed to a regular face if $H$ has no
iso-node. On a regular face lies only a single iso-line and through this iso-line passes only one iso-path
of $G$ according to Proposition~\ref{prop:class-7}. In this case, $H$ is a non-trivial L-face
and, hence, both iso-paths passes through two different iso-lines on $H$ and both iso-paths are different.
But in case, $H$ has iso-node then the L-face will be transformed to a singular face. But on a singular
face lies no iso-line. Therefore, in this case, only one iso-path passes through an iso-line of the L-face $H$.
\end{proof}

\subsection{Outer Iso-path of a labeled cuboid graph}
Until now we have not considered iso-paths on an L-face of a labeled cuboid graph $G(V,E,\mathcal{F})$ with
iso-level $c\in (0,1)$ which is regular and has at least one L-face. This section is devoted to compute
outer iso-paths of $G$. As we know on each regular and on each trivial L-face of $G$ lies only one iso-line
and on a singular face of $G$ lies no iso-line. Consequently, there is no iso-path on each of these faces
of $G$. But on a non-trivial L-face of $G$ we can have an iso-path if there exists a face-neighbored
labeled cuboid graph $G'(V',E',\mathcal{F}')$ with iso-level $c$ which is regular.

Let $G_1(V_1,E_1,\mathcal{F}_1)$ and $G_2(V_2,E_2,\mathcal{F}_2)$ be labeled cuboid graphs with iso-level
$c\in (0,1)$. Suppose that $G_1$ and $G_2$ are regular and face neighbors. Let $H_1(V_{h_1},E_{h_1},
\mathcal{F}_{h_1})\subset G_1$ be a face of $G_1$ and $H_2(V_{h_2},E_{h_2},\mathcal{F}_{h_2})\subset G_2$ be
a face of $G_2$ such that $V_{h_1}=V_{h_2}$. Note that, if $H_1$ is an L-face of $G_1$ this does not mean
that the face $H_2$ is an L-face of $G_2$. Hence, if we have an iso-path in $H_1$, this does not mean that
we have an iso-path in $H_2$. Hence, the node values of the common face of $G_1$ and $G_2$ can be different
in case both $H_1$ and $H_2$ are not regular faces. All these claims will be proved in this section.

If $H_1$ is an L-face then we give a {\it graph-theoretical rule} to indicate if an iso-path on it exists.
This graph-theoretical rule is used for the computation of iso-paths on L-faces in Section 5.4. The
graph-theoretical rules depend on the type of the L-face of $G_1$, on the type of the $S$-rule
which will be applied on $G_1$ and is chosen according to Table~\ref{table_1}, and on the type of the face
$H_2$ of $G_2$, and in case $G_2$ has at least one L-face then, in addition, from the type of the $S$-rule
which will be applied on $G_2$ and is chosen according to Table~\ref{table_1}.\\

\noindent{\bf Note: }When needed we use in special cases $\square$-equivalence for a labeled cuboid graph
with iso-level $c\in (0,1)$ as given in Notation~\ref{not:special}.

\begin{theorem}\label{thm:class-4}
Let $G_1(V_1,E_1,\mathcal{F}_1)$, $G_2(V_2,E_2,\mathcal{F}_2)$ be labeled cuboid graphs with iso-level
$c\in (0,1)$. We assume that both $G_1$ and $G_2$ are face neighbors such that the common face
$H(V_h,E_h,\mathcal{F}_h)$ is a non-trivial L-face. Furthermore, suppose that both $G_1$ and $G_2$ are
regular. Let the $S$-rules that will be applied on $G_1$ and on $G_2$ be chosen according to
Table~\ref{table_1}. Assume that on $G_1$ the rules $S_1$ or $S_2$ apply and on $G_2$ the rule $S_3$
applies. Then there exists an iso-path in $H$. Furthermore, to each iso-line on $H$ there exists only one
inner iso-path that passes through the iso-line.
\end{theorem}
\begin{proof} Let $G_1\in [(1)]_{\square}^{\star}$ and $G_2\in [(2)]_{\square}^{\star}$ with
the sketches $(1)$ and $(2)$ as shown in Figure~\ref{image_51}. Then the L-face $H$ is
in $[(a)]_{\square}$ or in $[(b)]_{\square}$ as shown in Figure~\ref{image_51}. In the case
$H\in [(a)]_{\square}$, application of the $S_1$- or $S_2$-rule to $G_1$ gives a graph represented
by the sketch $(a1)$ of Figure~\ref{image_52_53}. For the same case, application of the $S_3$-rule
to $G_2$ gives a graph represented by sketch $(a2)$ of Figure~\ref{image_52_53}. These graph
operations transform $H$ to a graph represented by the sketch $(a')$ of Figure~\ref{image_52_53}.
This is shown by the graph-theoretical rule $(A1)$ and by the sequence $(A2)$ in Figure~\ref{image_52_53}.
Hence we get an iso-path on $H$. For the case $H\in [(b)]_{\square}$ we apply the same procedure to
transform $H$ to a graph represented by the sketch $(b')$ of Figure~\ref{image_54_55}. These procedures
are shown by the graph-theoretical rule $(B1)$ and by the sequence $(B2)$ in Figure~\ref{image_54_55}.
Hence we get an iso-path on $H$.

The last claim holds true because to each disperse node in $H$ there is a {\it corresponding} iso-path
in $G_1$ (we get this using the $S_1$- or $S_2$-rule on $G_1$), and to the continuous node of $H$ which
is not an iso-node, there exists a corresponding iso-path in $G_2$ (we get this using the $S_3$-rule
on $G_2$). Here, the word "{\it corresponding}" is to be understood in the sense of
Notation~\ref{note:corresponding-node}.
\end{proof}
\begin{figure}[!ht]
\includegraphics[width=0.7\linewidth]{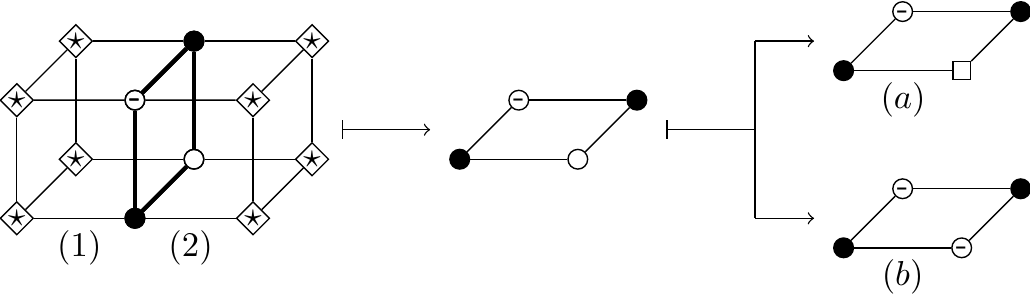}
\caption{Sketches $(1)$ and $(2)$ represent $[(1)]_{\square}^{\star}$ and $[(2)]_{\square}^{\star}$.
The common face of $G_1\in [(1)]_{\square}^{\star}$ and $G_2\in [(2)]_{\square}^{\star}$ is a non-trivial
L-face. The sketches $(a)$ and $(b)$ represent $[(a)]_{\square}$ and $[(b)]_{\square}$ in which the two
different types of non-trivial L-faces are obtained.}
\label{image_51}
\end{figure}
\FloatBarrier
\begin{figure}[!ht]
$A1.$
\begin{tabular}[c]{l}
\includegraphics[width=0.6\linewidth]{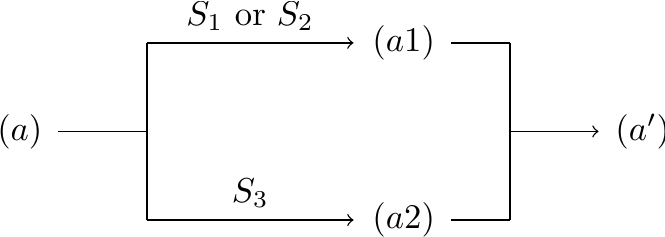}
\end{tabular}\\ \\
$A2.$
\begin{tabular}[c]{l}
\includegraphics[width=0.89\linewidth]{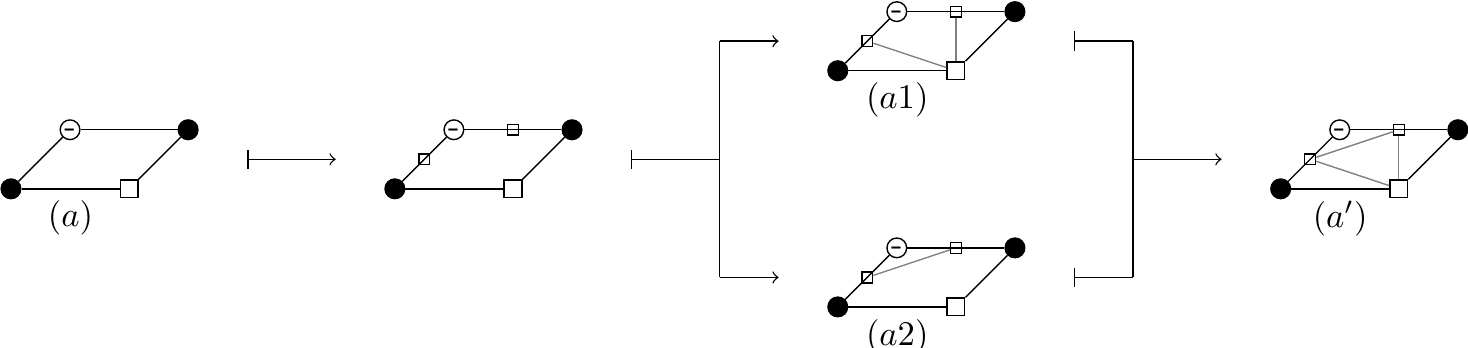}
\end{tabular}
\caption{The graph-theoretical rule given by $A1$ and the sequence $A2$ illustrate the procedure
how to compute an iso-path that lies on a non-trivial L-face which is in $[(a)]_{\square}$.}
\label{image_52_53}
\end{figure}
\FloatBarrier

\begin{figure}[!ht]
$B1.$
\begin{tabular}[c]{l}
\includegraphics[width=0.6\linewidth]{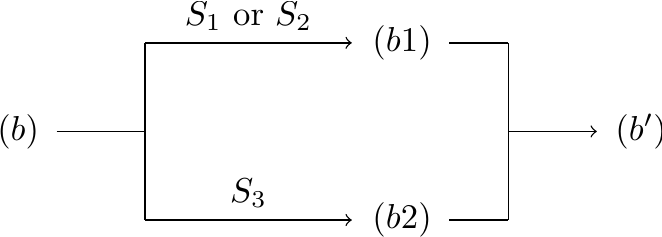}
\end{tabular}
\\ \\
$B2.$
\begin{tabular}[c]{l}
\includegraphics[width=0.89\linewidth]{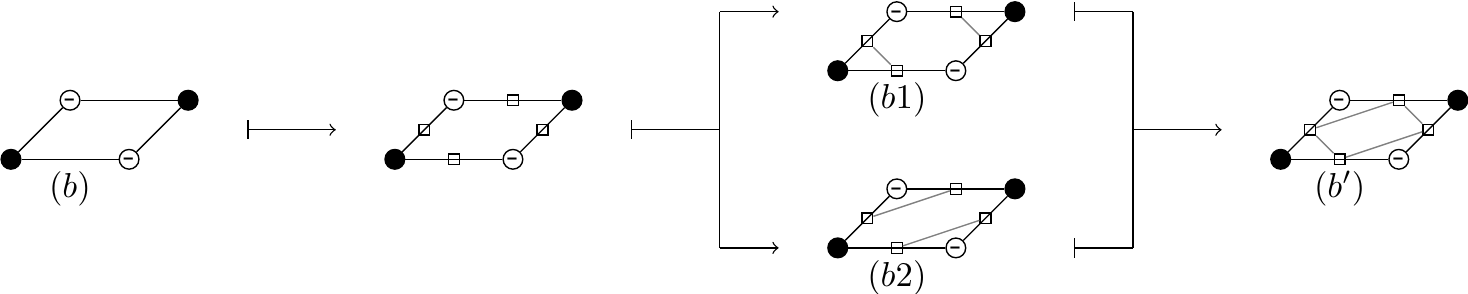}
\end{tabular}
\caption{The graph-theoretical rule given by $B1$ and the sequence $B2$ illustrate the procedure
to compute an iso-path that lies on a non-trivial L-face which is in $[(b)]_{\square}$.}
\label{image_54_55}
\end{figure}
\FloatBarrier
\begin{theorem}\label{thm:class-5}
Let $G_1(V_1,E_1,\mathcal{F}_1)$ and $G_2(V_2,E_2,\mathcal{F}_2)$ be labeled cuboid
graphs with iso-level $c\in (0,1)$. Assume that both $G_1$ and $G_2$ are face neighbors
with a common L-face $H(V_h,E_h,\mathcal{F}_h)$. Furthermore, suppose that $G_1$ and $G_2$ are
isolated iso-path free. Let $G_1\in [(1)]_{\square}^{\star}$ and $G_2\in [(2)]_{\square}^{\star}$,
where the sketches $(1)$ and $(2)$ are as shown in Figure \ref{image_56}. Let the $S$-rules that will
be applied on $G_2$ be chosen according to Table~\ref{table_1}. Let one of the rules $S_1$ or $S_2$ be
applicable on $G_2$. Then there exists an iso-path in $H$. But if rule $S_3$ is applicable on $G_2$ then
there exists no iso-path in $H$. Furthermore, to each iso-line on $H$ there exists only one inner iso-path
that passes through the iso-line.
\end{theorem}
\begin{proof} First, transform $G_1$ and $G_2$ to $G'_1(V_1,E_1,\mathcal{F}_1')$ and
$G'_2(V_2,E_2,\mathcal{F}_2')$, respectively, by applying the $T_1$-rule to each of
them. Then we have $G_1'\in [(1')]_{\square}^{\star}$ and $G_2'=G_2$, where sketch $(1')$ is as
shown in Figure~\ref{image_56}. Then $H\subset G_1$ will be transformed to a graph in
$[(a)]_{\square}$ and $H\subset G_2$ will be transformed to a graph in $[(b)]_{\square}$,
with sketches $(a)$ and $(b)$ from Figure~\ref{image_56}. We then apply the $C_3$-rule
to $G_1'$ and one of the rules $S_1$ or $S_2$ to $G_2$ in order to compute the iso-path on $H$
as shown in sketch $(c)$ of the sequence $A2$ of Figure~\ref{image_57_58}. This procedure is
illustrated by the graph-theoretical rule given by $A1$ in Figure~\ref{image_57_58}.

The last claim holds true because to each disperse node in $H$ there is a {\it corresponding}
iso-path in $G_2$ (we get this using the $S_1$- or $S_2$-rule on $G_2$), and to the continuous
node of $H$ which is not an iso-node, there exists a corresponding iso-path in $G_1$ only if the
$S_3$-rule is applied on $G_1$. Here, the word "{\it corresponding}" is to be understood in the
sense of Notation~\ref{note:corresponding-node}.
\end{proof}
\begin{figure}[!ht]
\includegraphics[width=0.7\linewidth]{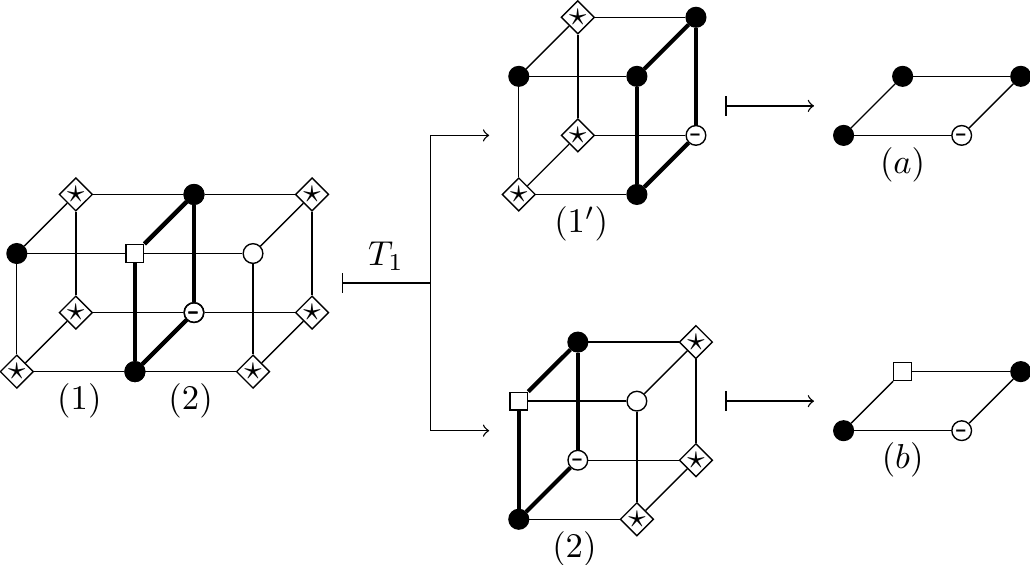}
\caption{The sketches $(1)$ and $(2)$ represent $[(1)]_{\square}^{\star}$ and $[(2)]_{\square}^{\star}$,
respectively. Sketches $(a)$ and $(b)$ represent faces of $G_1'\in [(1')]_{\square}^{\star}$ and
$G_2\in [(2)]_{\square}^{\star}$, respectively. Both faces $(a)$ and $(b)$ have the same nodes, but
different node weights.}
\label{image_56}
\end{figure}
\FloatBarrier
\begin{figure}[!ht]
$A1.$
\begin{tabular}[c]{l}
\includegraphics[width=0.89\linewidth]{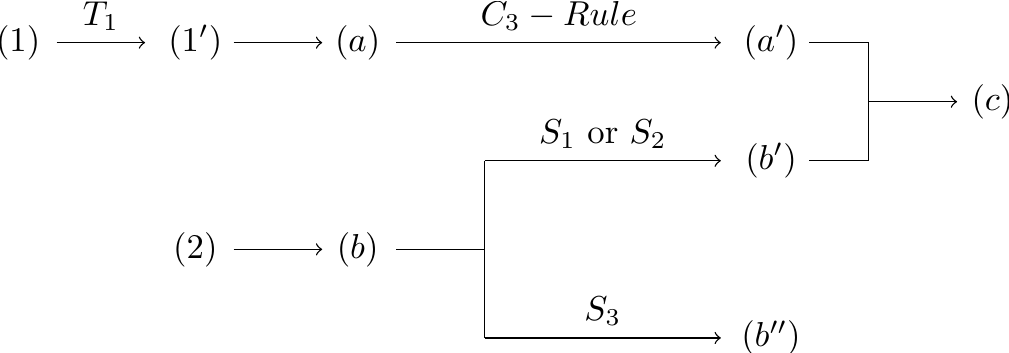}
\end{tabular}\\ \\
\mbox{}\\ \\
$A2.$
\begin{tabular}[c]{l}
\includegraphics[width=0.89\linewidth]{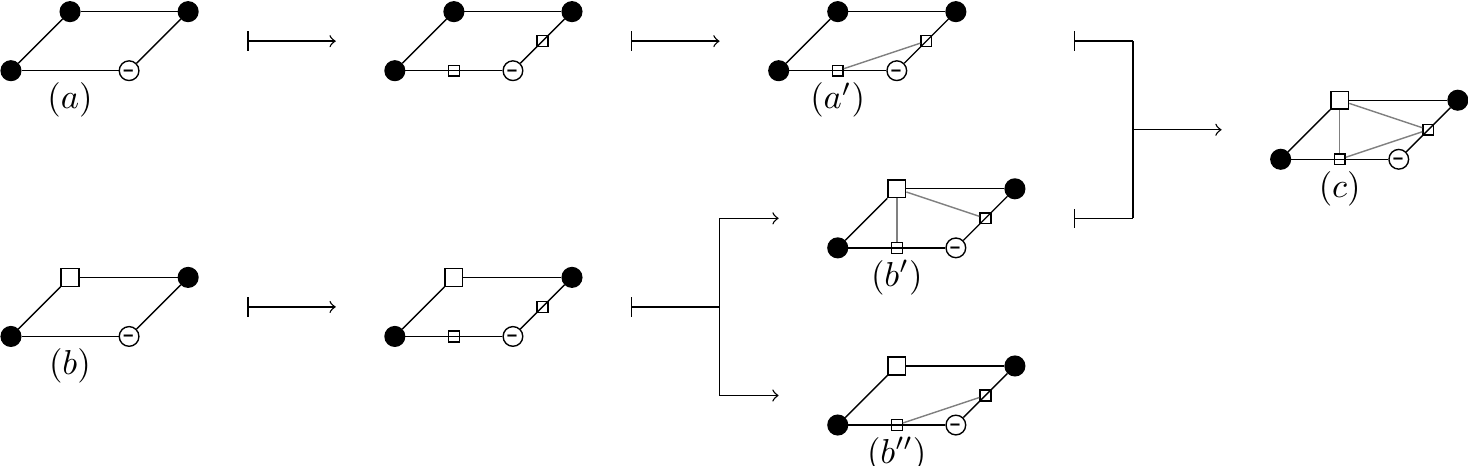}
\end{tabular}
\caption{The graph-theoretical rule given by $A1$ and the sequence $A2$ illustrate the procedure
to compute an iso-path that lies on a non-trivial L-face which is in $[(b)]_{\square}$.}
\label{image_57_58}
\end{figure}
\FloatBarrier
\noindent{\bf Explanation of Figure~\ref{image_57_58} : }The graph-theoretical rule $A1$ together with
the sequence of the sketches $A2$ means that we first apply the $T_1$-rule to $G_1$, obtaining $G_1'$.
Then we apply to the common face of $G_1'$ and $G_2$ the $C_3$-rule to get $(a')$. By applying the $S_1$- or
$S_2$-rule to $G_2$, we get $(b)\longrightarrow (b')$. But if we apply the $S_3$-rule to $G_2$ we get
$(b)\longrightarrow (b'')$. From $(b')$ we get an iso-path on the common face of $G_1'$ and $G_2$. The
sequence $A2$ in Figure~\ref{image_57_58} shows that the regular face obtained in $[(a)]_{\square}$ and the
L-face obtained in $[(b)]_{\square}$ have iso-lines as shown by $(a')$ and $(b')$ or $(b'')$, respectively.
In case $(b')$, the common face of $G_1'$ and $G_2$ has an iso-path as shown by $(c)$. In case $(b'')$, there
is no iso-path on the common face of $G_1'$ and $G_2$.
\FloatBarrier

\subsection{Algorithm for Complete Iso-path Computation}
In this section we will give an algorithm for the complete iso-paths computation of a labeled cuboid graph
$G(V,E,\mathcal{F})$ with iso-level $c\in (0,1)$ which is neither disperse nor continuous.

The complete iso-paths of $G$ will be computed in three steps. In the first step we delete all singular
iso-paths or an isolated iso-path of $G$. These deletions will be carried out by transforming the graph
$G$ to a graph $G'$. If $G'$ is regular then the second and the third step will be
to compute the iso-paths of $G'$ which are as well iso-paths of $G$. The only difference between the
iso-paths of $G$ and $G'$ is that $G'$ contains no singular iso-paths and no isolated iso-path. Here,
if $D(G')\leq 7$ then $G'$ is regular but if $D(G')=8$ then $G$ contains only singular or isolated iso-paths.
In this case $G'$ is a disperse labeled cuboid graph and has no iso-path and hence $G$ has as well no iso-path.
This means, in case $G'$ is a disperse graph we consider as well $G$ as a disperse graph.\\

Let $\Omega\subset\mathbb{R}^3$ be a polygonal domain allowing a partition into cuboids,
i.e. $\overline{\Omega}=\cup_{i=1}^NC_i$ for the partition $\mathcal{T}=\{C_i\}_{i=1}^N$.  Let
$\mathcal{G}=\{G_1,\ldots,G_N\}$ be a set of labeled cuboid graphs
with common iso-level $c\in (0,1)$, and let $C_i$ be the cuboid of $G_i$ for $i=1,\ldots,N$. The following
algorithm computes the complete iso-paths of $G_i$ for $i=1,\ldots,N$.

\subsection*{Algorithm for Iso-paths Computation}
\noindent{\bf Notations}: Let $G(V,E,\mathcal{F})$ be a labeled cuboid graph with iso-level $c\in (0,1)$.
Denote by $D(G)$ the total number of disperse nodes of $G$ and by $L(G)$ and $L_{NT}(G)$ the set
of L-faces and non-trivial L-faces of $G$, respectively. Recall the $T_1^*$-mapping as defined in
Definition~\ref{def:iso-path-star-map}.\\

\noindent {\bf For } $i=1,\ldots,N$ {\bf do:}\\
\mbox{}\hspace{0.3cm}    {\bf Step 1}: {\bf removing singular iso-paths or isolated iso-path}\\
\mbox{}\hspace{0.6cm}      i) removing singular iso-paths if $G_i=F_2(F_1(G_i))$\,:\\
\mbox{}\hspace{1.1cm}         a) $G_i':=T_2(T_1^*(G_i))$\\
\mbox{}\hspace{1.1cm}         b) if $(D(G_i')=8)$ then no iso-path, $i:=i+1$, {\bf go to} Step 1\\
\mbox{}\hspace{0.6cm}      ii) removing isolated iso-path if $G_i=T_2(T_1(G_i))$\,:\\
\mbox{}\hspace{1.1cm}         a) $G_i':=F_2(F_1(G_i))$\\
\mbox{}\hspace{1.1cm}         b) if $(D(G_i')=8)$ then no iso-path, $i:=i+1$, {\bf go to} Step 1\\

\noindent\mbox{}\hspace{0.3cm}    {\bf Step 2}: {\bf iso-path computation of} $G_i'$
{\bf for the case} $L(G_i')\neq\emptyset$\\
\mbox{}\hspace{0.6cm}      (i) register all L-faces: \\
\mbox{}\hspace{1.2cm}          let $L(G_i')=\{\mathcal{L}_1,\ldots,\mathcal{L}_l\}$ and $l=|L(G_i')|$\\
\mbox{}\hspace{0.6cm}      (ii) register all non-trivial L-faces of $G_i'$ if $L_{NT}(G_i')\neq\emptyset$\,:\\
\mbox{}\hspace{1.2cm}           let $L_{NT}(G_i')=\{\mathcal{M}_1,\ldots,\mathcal{M}_m\}$ and $m=|L_{NT}(G_i')|$\\
\mbox{}\hspace{0.6cm}      (iii) let $\mathcal{N}_j\subset\mathcal{G}$ for $j=1,\ldots,m$ be all face neighbors of $G_i'$ such that \\
\mbox{}\hspace{1.4cm}            the common nodes of $G_i'$ and $\mathcal{N}_j$ are the nodes of $\mathcal{M}_j$\\
\mbox{}\hspace{0.6cm}      (iv) use Table~\ref{table_1} to get the type of the $S^n$-rule corresponding to $G_i'$\\
\mbox{}\hspace{1.3cm}            a) use $S^n$-rule to compute iso-paths of $G_i'$ \\
\mbox{} \hspace{1.7cm}              $G_i'':=S^n(G_i')$\\
\mbox{} \hspace{1.7cm}              {\bf Note: }$G_i''$ is the rest graph of $G_i'$ and hence irreducible\\
\mbox{}\hspace{1.3cm}            b) use $C$-rules to compute the iso-path of $G_i''$ (see Theorem~\ref{thm:class-1})\\
\mbox{}\hspace{0.6cm}     (v) {\bf if} $m>0$ {\bf then} compute iso-path on non-trivial L-faces of $G_i'$\,:\\
\mbox{}\hspace{1.2cm}         {\bf for} $j=1,\ldots,m$\\
\mbox{}\hspace{1.7cm}             a) use Table~\ref{table_1} to get the type of the $S$-rule corresponding to $\mathcal{N}_j$\\
\mbox{}\hspace{1.7cm}             b) use Theorems~\ref{thm:class-4} and~\ref{thm:class-5} to compute a possible iso-path on\\
\mbox{}\hspace{2.27cm}                $\mathcal{M}_j$\\

\noindent\mbox{}\hspace{0.3cm}   {\bf Step 3}: {\bf iso-path computation of} $G_i'$ {\bf for the case} $L(G_i')=\emptyset$\\
\mbox{}\hspace{0.6cm}        {\bf if} $D(G_i')\in\{2,6\}$ and {\bf if} there exists $H'$ in $[g_1]_{\circ}$ or in
$[g_3]_{\square}$, {\bf where}\\
\mbox{}\hspace{1.1cm}              $H'\subset G_i'$, use Theorem~\ref{thm:class-2} to compute iso-path\\
\mbox{}\hspace{0.6cm}        {\bf else}\\
\mbox{}\hspace{1.1cm}           apply $C$-rules to compute iso-path \\
\mbox{}\hspace{0.6cm}        {\bf endif}\\

Note that the complexity of this algorithm is $O(N)$.

\subsection{Application}
Tracking or capturing interfaces of two-phase systems is an important issue for instance in computational
fluid dynamics simulations~\cite{opac-b1133222}. The interfaces are used not only to track the phases but
there can be adsorbed quantities on them like surfactants as well, which affect the hydrodynamics of the
system~\cite{Alke_and_Bothe}. Two well-known volume tracking methods are the Volume of Fluid (VOF)
method~\cite{Nich81} and the Level Set method~\cite{Stan02}. The Level Set method uses a signed distance
function which implicitly determines the interface as the zero level set. While level set methods are
advantageous concerning discrete mean curvature computation, they suffer from volume loss of the disperse phase.
The latter requires reinitialization of the level set function which introduces non-physical changes.
The VOF-method conserves the phase volumes, but the standard interface reconstruction using a piecewise
planar approximation (PLIC,~\cite{Rider97reconstructingvolume}) leads to disconnected interface representations.
The present iso-surface algorithm can be employed to obtain a connected interface approximation instead.

Consider a polygonal domain $\Omega\subset\mathbb{R}^3$ and a domain partition
$\mathcal{T}=\{C_i\}_{i=1}^N$ of $\overline{\Omega}$ into cuboids. Given a function $v\,:\,\mathcal{T}
\longrightarrow [0,1]$ which gives the volume fractions of one of the phases (say the disperse phase),
we define a labeling function $\mathcal{F}\,:\mathcal{P}\longrightarrow [0,1]$, where $\mathcal{P}$
is the set of vertices of all cuboids in $\mathcal{T}$, by
\begin{equation}
\mathcal{F}(P)=\frac{1}{|I_P|}\sum_{i\in I_p}^nv(C_i),
\label{eq:application-1}
\end{equation}
where $i\in I_P$ if and only if $P\in C_i$. Using the labeling function $\mathcal{F}$, we get from
$\mathcal{T}=\{C_i\}_{i=1}^N$ labeled cuboid graphs $G_1(V_1,E_1,\mathcal{F}_1),\ldots,G_N(V_N,E_N,\mathcal{F}_N)$.
Next, we interpolate the node values onto the edges according to \mbox{Definition}~\ref{def:labeled-graph-5}. Then, for
a given $\epsilon>0$, we determine for all graphs $G_1,\ldots,G_N$ a common iso-level $c\in (0,1)$ by solving the
inequality
\begin{equation}
\left|1-\frac{\Omega_{D}(c)}{\sum_{i=1}^Nv(C_i)|C_i|}\right|<\epsilon,
\label{eq:application-2}
\end{equation}
where the domain $\Omega_{D}(c)\subset\mathbb{R}^3$ is the union of all bounded volumes enclosed by
the iso-surface for the given iso-level $c$. Note that $\Omega_{D}(c)$ contains all disperse nodes
of $G_1,\ldots,G_N$.

The function $\gamma\,:\,[0,1]\longrightarrow (-\infty,\infty)$ given by
\begin{equation}
\gamma(c)=1-\frac{\Omega_{D}(c)}{\sum_{i=1}^Nv(C_i)|C_i|}
\label{eq:application-3}
\end{equation}
measures the deviation between the total volume of the given disperse phase and that of the computed
enclosed volume. The function $\gamma$ is decreasing, but not necessarily strictly decreasing.
Furthermore, $\gamma$ can have (small) jumps. The latter can appear at an iso-level $c$ if
$\mathcal{F}$ attains the value $c$ on several, complete edges. In such cases, the error bound $\epsilon$
in~\eqref{eq:application-2} cannot be chosen arbitrarily small.

Using the iso-surface algorithm from Section 5.4 we first computed iso-surfaces for snap shots of the
simulated collision of two liquid droplets. The snap shots are taken at different time steps and the
resulting iso-surfaces are shown in Figure~\ref{image_59_60_61_62}. We have also applied the
iso-surface algorithm to the outcome of a crown splash, where one typical snap shot is shown in
Figure~\ref{image_63}. For the iso-surface computation of the binary droplet collision and
the splash we computed the iso-level $c$ using $\epsilon=10^{-9}$ in~\eqref{eq:application-2}.
\begin{figure}[!ht]
\includegraphics[width=0.3\linewidth]{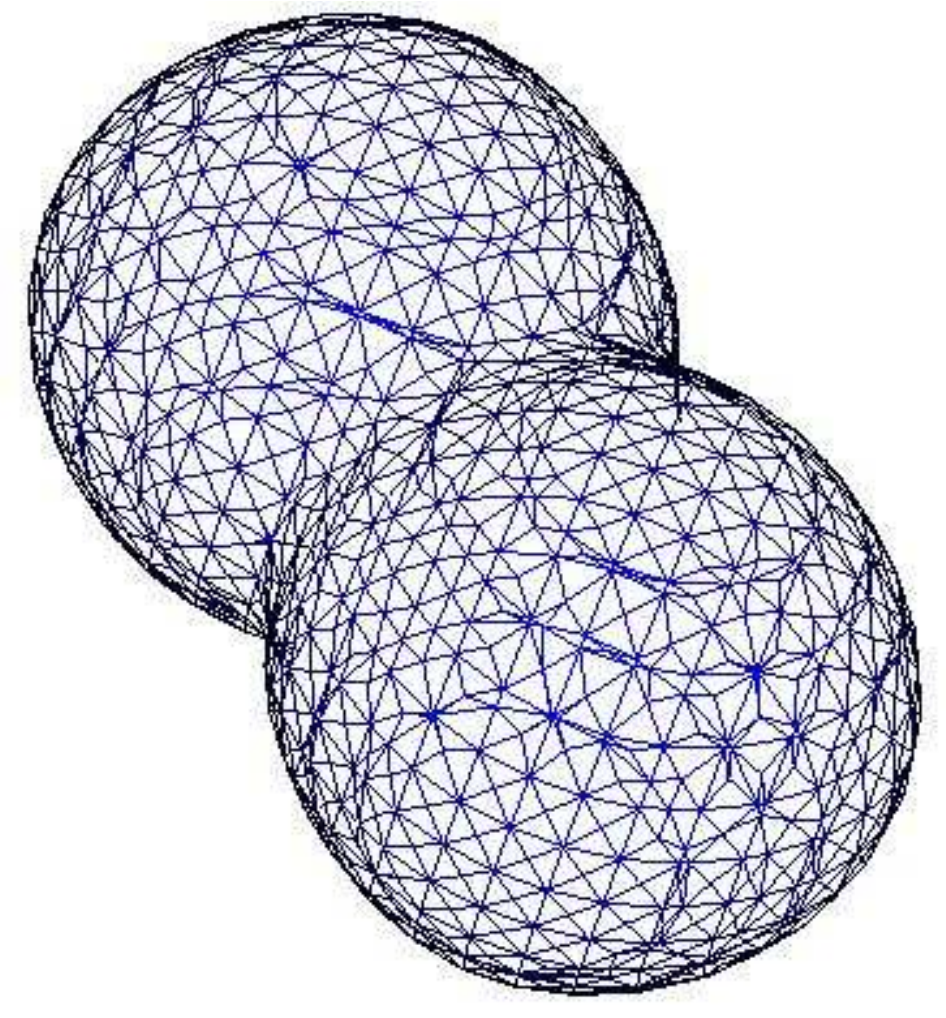}\hspace{1cm}
\includegraphics[width=0.4\linewidth]{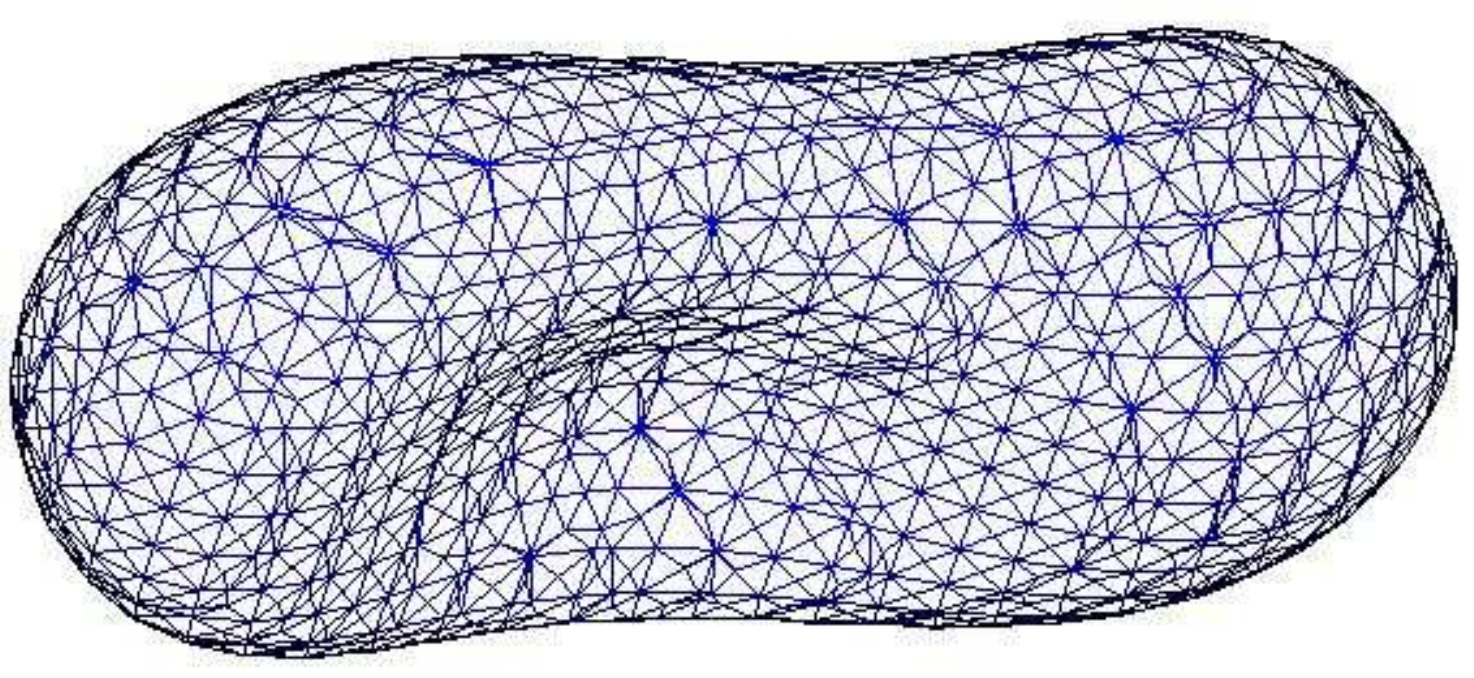}\\ \\
\includegraphics[width=0.4\linewidth]{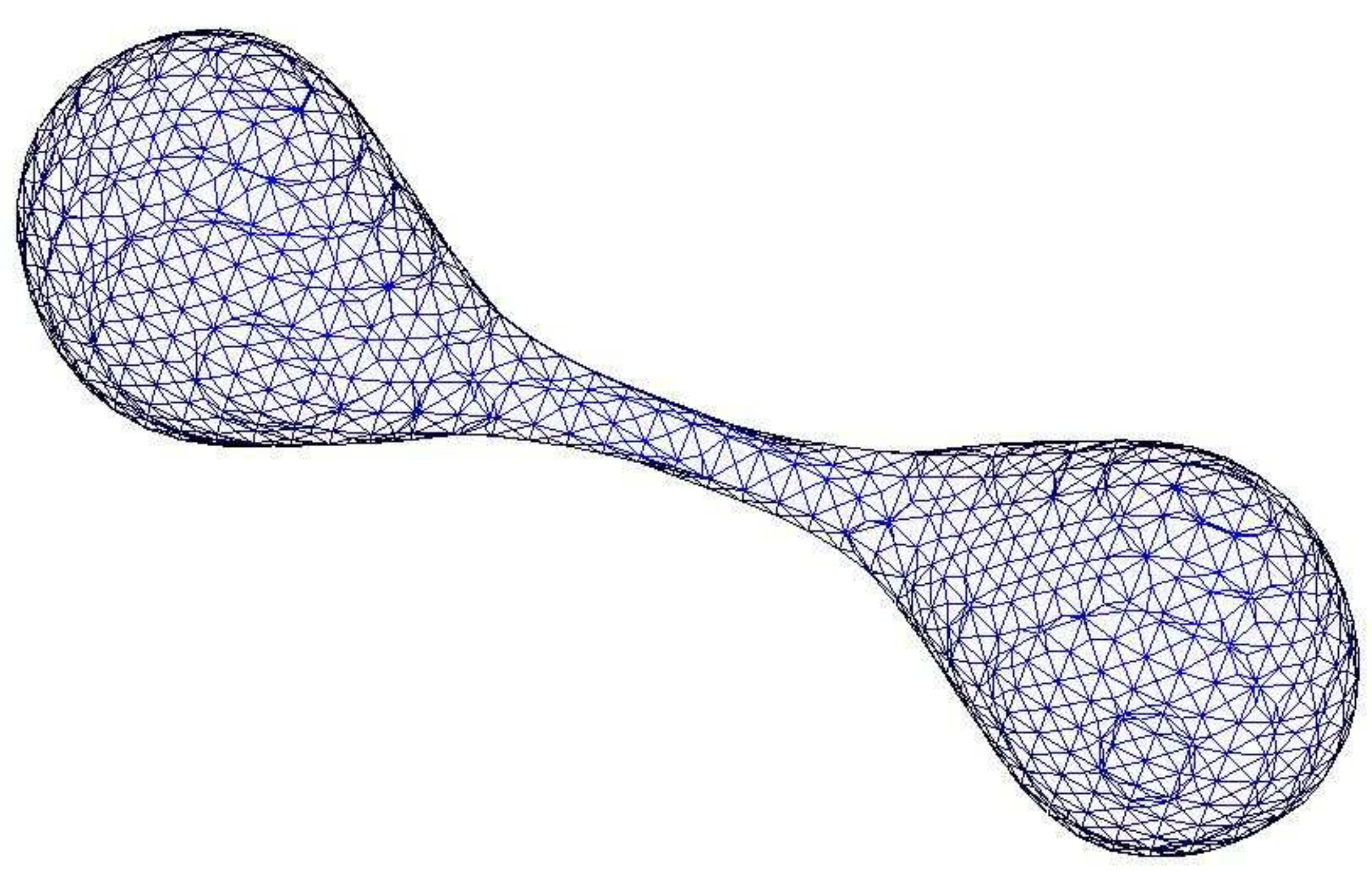}\hspace{1cm}
\includegraphics[width=0.4\linewidth]{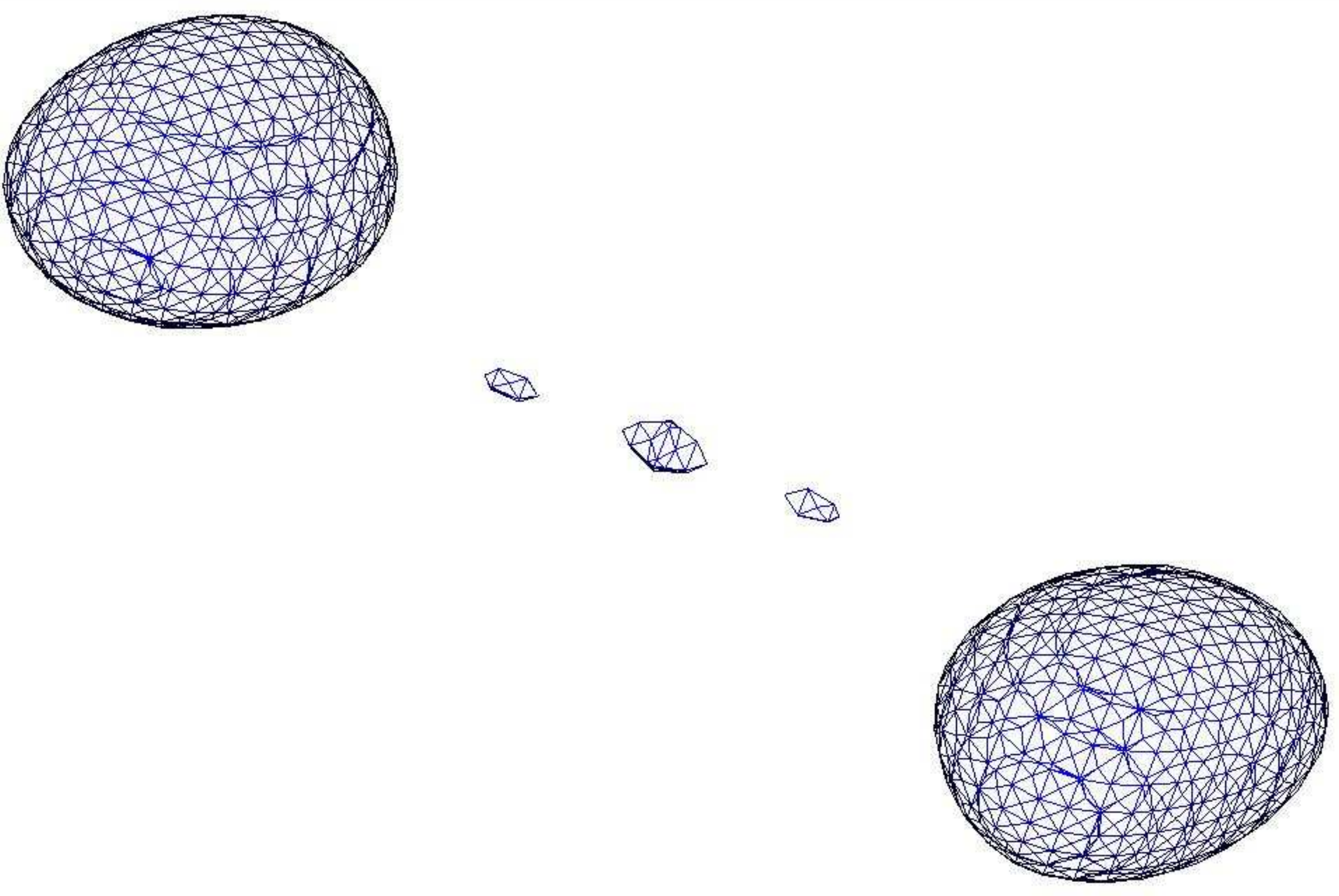}
\caption{The sequence of sketches show iso-surface meshes for snap shots of a boundary droplet collisions taken at different
time steps.}
\label{image_59_60_61_62}
\end{figure}
\begin{figure}[!ht]
\includegraphics[width=0.7\linewidth]{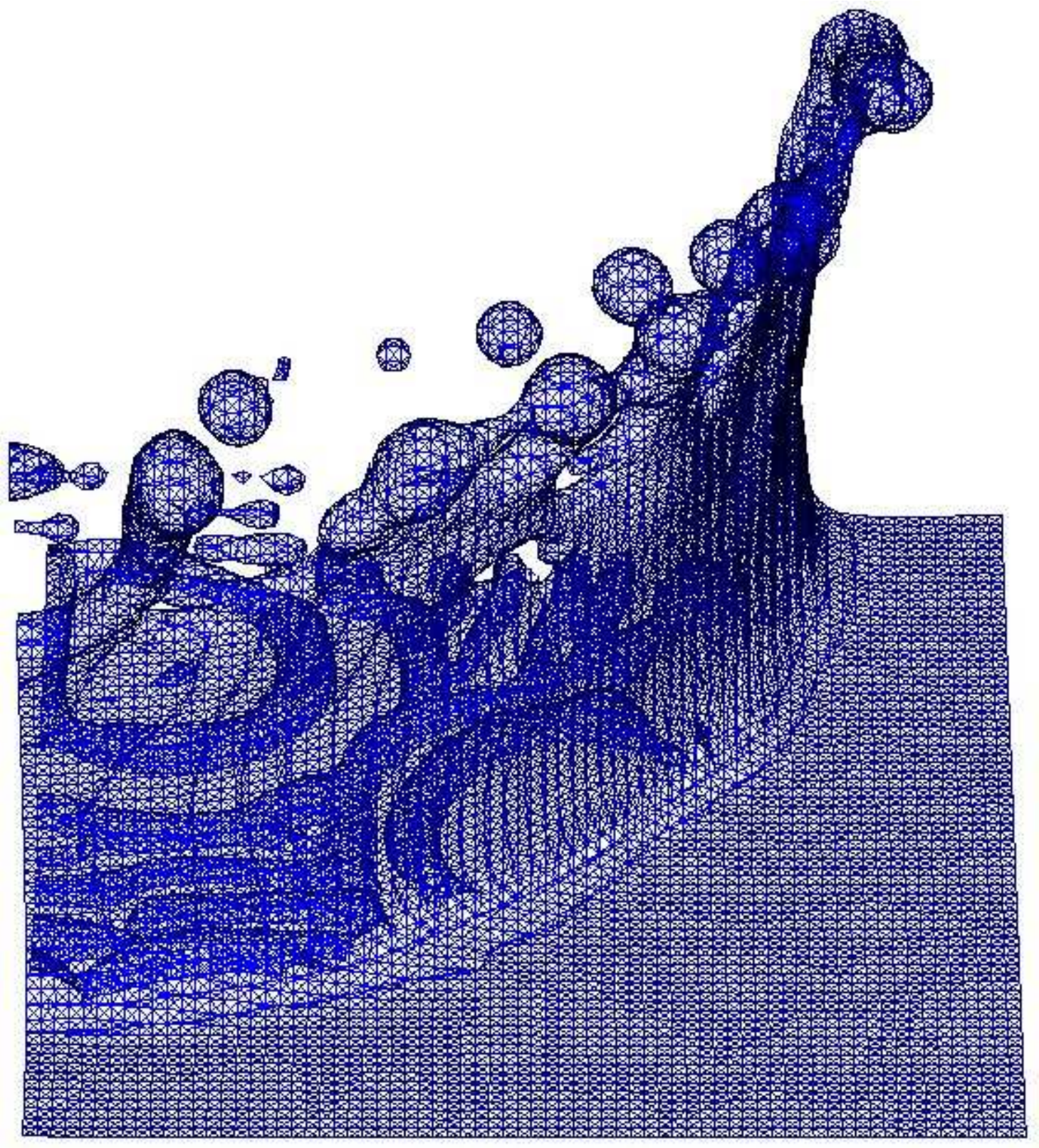}
\caption{The figure shows the iso-surface mesh for a snap shot of a crown splash.}
\label{image_63}
\end{figure}
\FloatBarrier

The highly dynamic impact of a droplet into a liquid layer which produces the crown-splash is a good
validation example, since during the splash several special cases occur in the iso-surface computation.
Indeed, there appear labeled cuboid graphs $G(V,E,\mathcal{F})$ containing
\begin{enumerate}
\item singular faces,
\item a trivial L-face,
\item non-trivial L-faces,
\item edges with iso-node end points.
\end{enumerate}

Sketches $(a1)$, $(a2)$ and $(a3)$ of Figure~\ref{image_64_65_66} illustrate some of the special
cases which appear. Sketch $(a1)$ illustrates a labeled cuboid graph and its iso-path, where the graph
contains a singular face. Sketch $(a2)$ illustrates a labeled cuboid graph and its iso-path, where the
graph contains a trivial L-face. Sketch $(a3)$ illustrates a labeled cuboid graph and its iso-path, where
the graph contains two non-trivial L-faces and an edge with iso-nodes as end points of it.
\begin{figure}[!ht]
\includegraphics[width=0.16\linewidth]{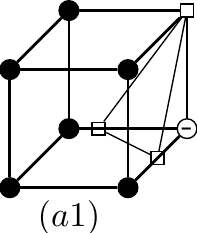}\hspace{2cm}
\includegraphics[width=0.16\linewidth]{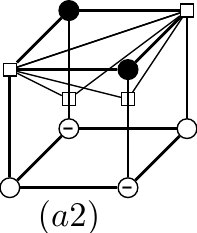}\hspace{2cm}
\includegraphics[width=0.16\linewidth]{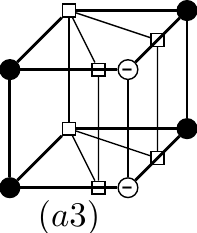}
\caption{Sketches $(a1)$, $(a2)$ and $(a3)$ illustrate labeled cuboid graphs and their iso-paths, where each of the
cuboids has at least one iso-node.}
\label{image_64_65_66}
\end{figure}
\FloatBarrier

\section{Connectedness of Iso-paths}
Let $G(V,E,\mathcal{F})$ be a labeled cuboid graph with iso-level $c\in (0,1)$. Then an iso-line in $G$
can be common for two distinct iso-paths that lie in $G$. In all other cases an iso-line in $G$
can be common for at least two iso-paths, one corresponding to $G$ and one or more corresponding to another
labeled cuboid graph $G'(V',E',\mathcal{F}')$, where $G'$ is a face or edge neighbor of $G$. Iso-lines which
are common for at least two iso-paths are as well common for the corresponding iso-elements. Therefore,
a common iso-line of at least two iso-paths is common for at least two distinct iso-elements. This means
that iso-elements are connected via such iso-lines. We also say that the iso-paths are connected at
the common iso-line. If each edge of an iso-path is common for at least two distinct iso-paths then
we say that the iso-path is connected. Connected iso-paths give rise to connected iso-surfaces. Therefore,
connectedness of iso-paths is a very important property which will be investigated in this section.

Consider a polygonal domain $\Omega\subset\mathbb{R}^3$ having a cuboid partition $\mathcal{T}=\{C_i\}_{i=1}^N$
such that $\overline{\Omega}=\cup_{i=1}^NC_i$. Let $\mathcal{G}=\{G_1,\ldots,G_N\}$ be a set of labeled cuboid
graphs with common iso-level $c\in (0,1)$ and $C_i$ be the cuboid of $G_i$ for $i=1,\ldots,N$. We compute the
iso-paths in each $G_i$ by applying the algorithm given in Section 5.4. Then the following questions arise:
\begin{enumerate}
\item is it possible to show iso-surface connectivity?
\item is it possible to decompose iso-surfaces into components such that each edge of an iso-element
in a component is a common to only two distinct iso-elements in the same component?
\item how to compute discrete mean curvature at iso-points?
\item is it possible to get all local topological information of an iso-surface in a simple way?
\end{enumerate}
The first problem will be answered affirmative in this section and the remaining questions will
receive a positive answer in Sections 7 and 8.\\

\noindent{\bf Note: }When needed we use in special cases $\square$-equivalence for a labeled cuboid graph
with iso-level $c\in (0,1)$ as given in Notation~\ref{not:special}.

\begin{theorem}\label{thm:connect-1}
Let $G(V,E,\mathcal{F})$ be a labeled cuboid graph with iso-level $c\in (0,1)$. Suppose $G$ is regular
and has a trivial L-face. Then $G$ contains exactly two distinct inner iso-paths. Furthermore, $G$ has
only one trivial L-face.
\end{theorem}
\begin{proof} It holds that $G\in [(a)]_{\square}$ or $G\in [(b)]_{\square}$ or $G\in [(c)]_{\circ}$,
where the sketches $(a),(b),(c)$ are as given in Figure~\ref{image_67}, since precisely in these cases
$G$ is isolated iso-path free. We then apply the $S_1^2$-rule to $G$ if $G\in [(a)]_{\square}$,
or apply $S_1$- and $S_2$-rules to $G$ if $G\in [(b)]_{\square}$, or apply the $S_2^2$-rule to $G$ if
$G\in [(c)]_{\square}$. We get in all cases two distinct inner iso-paths in $G$. Furthermore, if
$G\in [(d)]_{\square}$ then application of the $F_1$-rule to $G$ transforms $G$ to a labeled cuboid
graph $G'(V,E,\mathcal{F}')\in [(e)]_{\square}$ where the sketches $(d)$ and $(e)$ are as given in
Figure~\ref{image_67}. But $G'$ is a disperse graph and has no L-faces. Hence, $G$ has only one trivial
L-face.
\end{proof}
\begin{figure}[!ht]
\includegraphics[width=0.7\linewidth]{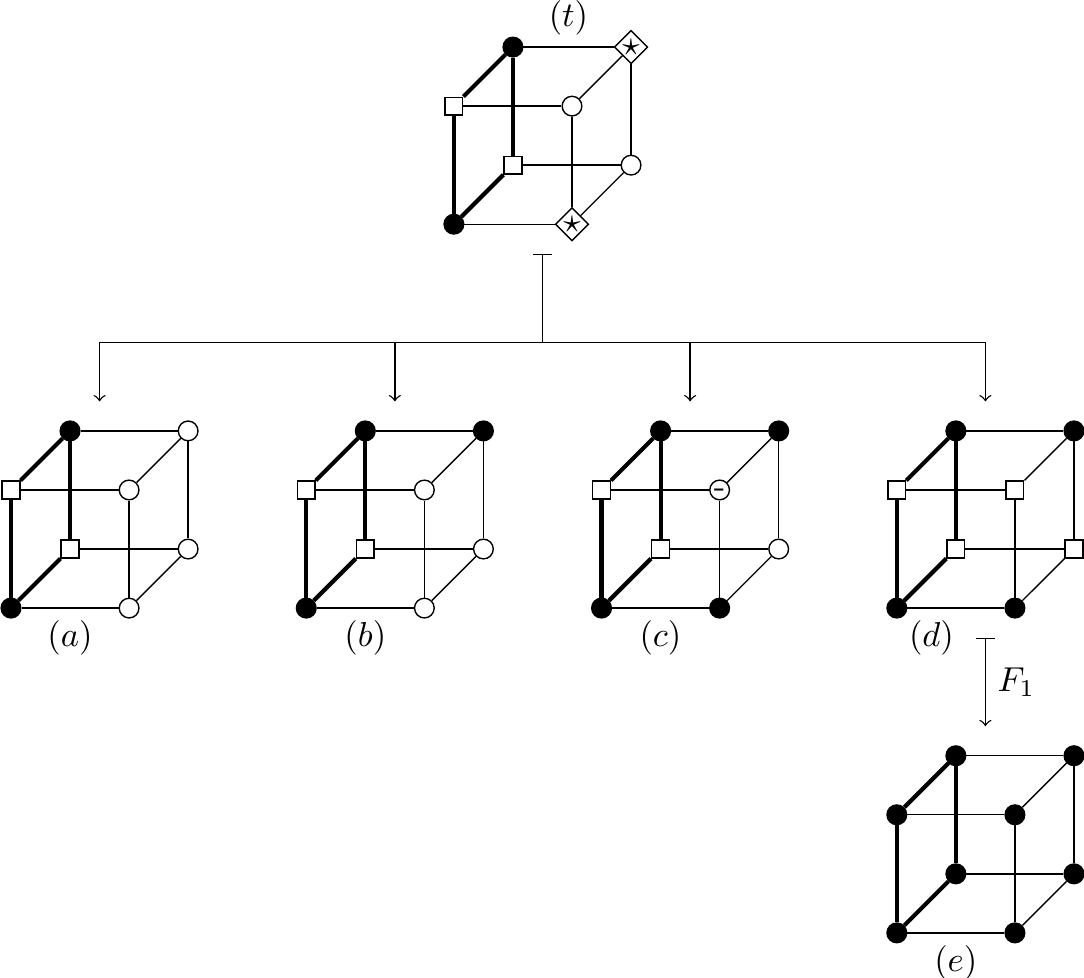}
\caption{Sketch $(t)$ represents $[(t)]_{\square}^{\star}$. Then, by substituting in $(t)$ for the
unknowns of the two nodes marked by the symbol $\diamondsuit\!\!\!\!\star$ the symbols
$\circ$ or $\bullet$, we get four different types of labeled cuboid graphs that lie in
$[(a)]_{\square}$, $[(b)]_{\square}$, $[(c)]_{\square}$ or $[(d)]_{\square}$, respectively.}
\label{image_67}
\end{figure}
\FloatBarrier
\begin{lemma}\label{lemma:connect-1}
Let $G_1(V_1,E_1,\mathcal{F}_1)$ and $G_2(V_2,E_2,\mathcal{F}_2)$ be labeled cuboid graphs with
iso-level $c\in (0,1)$. We assume that both $G_1$ and $G_2$ are face neighbors with a trivial
L-face $H(V_h,E_h,\mathcal{F}_h)$. Furthermore, suppose that $G_1$ and $G_2$ are isolated iso-path
free. Then there exist exactly two distinct inner iso-paths in $G_2$ that pass through the iso-line
in $H$.
\end{lemma}
\begin{proof} We consider two different labelings of $G_1\in [(1)]_{\square}^{\star}$ and
$G_2\in [(2)]_{\square}^{\star}$ as given by the graphical sequences in $A$ and $B$ of
Figure~\ref{image_68_69}. Transform first $G_1$ and $G_2$ to $G'_1(V_1,E_1,\mathcal{F}_1')$ and
$G'_2(V_2,E_2,\mathcal{F}_2')$, respectively, by applying the $T_1$-rule to each of
them. Then we have $G_1'\in [(1')]_{\square}^{\star}$ and $G_2'=G_2$, where sketch $(1')$ is as
shown in Figure~\ref{image_68_69}. Furthermore, the common trivial L-face of graphs $G_1$ and
$G_2$ will be transformed to $H_1(V_h,E_h,\mathcal{F}_h)\in [(a)]_{\square}$ and
to $H_2(V_h,E_h,\mathcal{F}_h)\in [(b)]_{\square}$, respectively. Sketches
$(a)$ and $(b)$ are as shown in Figure~\ref{image_68_69}. Both $G_1'$ and $G_2$ are isolated iso-path
free and, hence, there exists an iso-path in $G_2$. But since we have a trivial L-face in $G_2$,
we have at least two distinct inner iso-paths in $G_2$ that pass through the iso-line in $H$ according
to Theorem~\ref{thm:connect-1}. Application of $S_1^2$- or $S_2^2$-rule or $S_1$- and $S_2$-rules
to $G_2\in [(2)]_{\square}$, where $G_2$ is isolated iso-path free, we get exactly two distinct inner
iso-paths of $G_2$.
\end{proof}
\begin{figure}[!ht]
\noindent $A.$
\begin{tabular}[c]{l}
\includegraphics[width=0.7\linewidth]{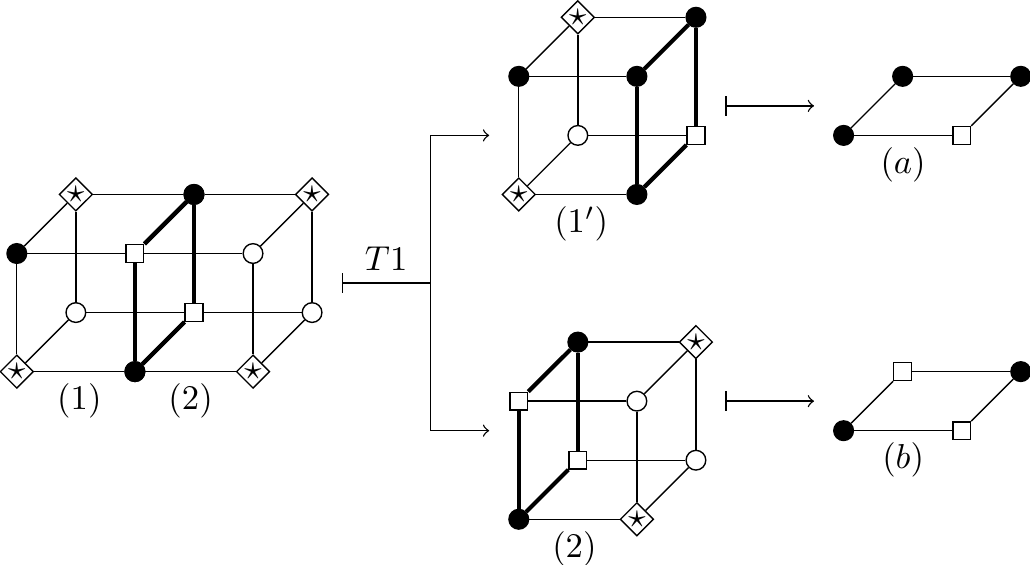}
\end{tabular}
\\ \\

\noindent $B.$
\begin{tabular}[c]{l}
\includegraphics[width=0.7\linewidth]{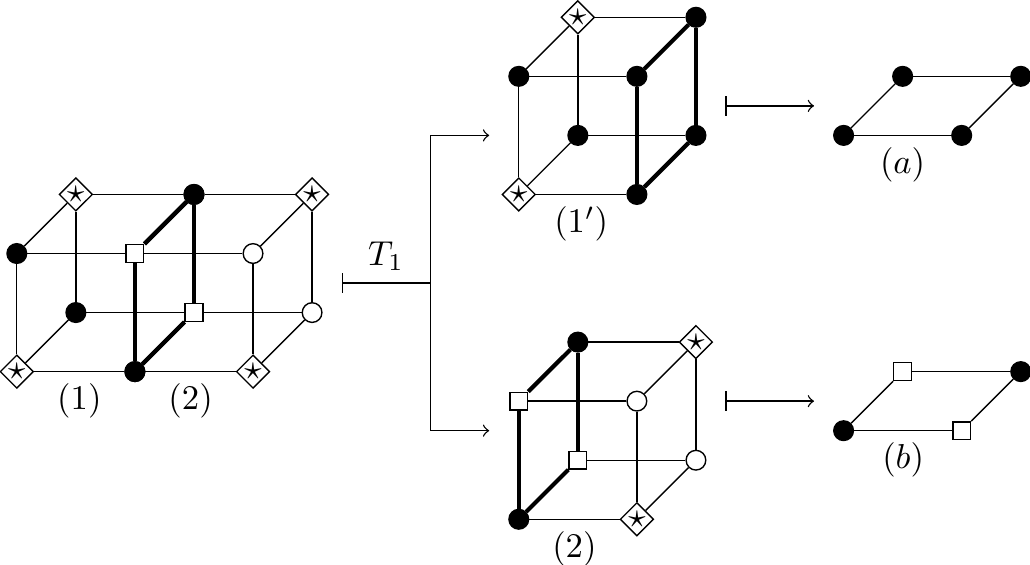}
\end{tabular}
\caption{Application of $T_1$-rule to $G_1\in [(1)]_{\square}^{\star}$ and to $G_2\in [(2)]_{\square}^{\star}$
transforms the common trivial L-face of both graphs to graphs in $[(a)]_{\square}$ and in $[(b)]_{\square}$,
respectively. The sequences of sketches $A$ and $B$ show the complete possible results.}
\label{image_68_69}
\end{figure}
\FloatBarrier
\begin{theorem}\label{thm:connect-2}
Let $G_1(V_1,E_1,\mathcal{F}_1)\in [(1)]_{\square}^{\star}$ and $G_2(V_2,E_2,\mathcal{F}_2)
\in [(2)]_{\square}^{\star}$ with iso-level $c\in (0,1)$ where sketches $(1)$ and $(2)$ are as
shown in Figure~\ref{image_70_71_72}. Both $G_1$ and $G_2$ are face neighbors with a regular
face $H(V_h,E_h,\mathcal{F}_h)$ which is in $[(a)]_{\circ}$ or $[(b)]_{\square}$ or
$[(c)]_{\square}$, where sketches $(a),(b),(c)$ are as shown in Figure~\ref{image_73}.
Then there is only one iso-line in $H$.
\end{theorem}
\begin{proof} Apply the $C$-rules to $H$. We consider in Figure \ref{image_70_71_72} different
labelings of $G_1$ and $G_2$ which give different types of regular faces as a face neighbor
of both graphs. In all these cases as shown by the graphical sequences $A$, $B$ and $C$ in
Figure~\ref{image_70_71_72} we get, by applying the $C$-rules to each of the regular faces,
a single iso-line.
\end{proof}
\begin{figure}[!ht]
\noindent $A.$
\begin{tabular}[c]{l}
\includegraphics[width=0.8\linewidth]{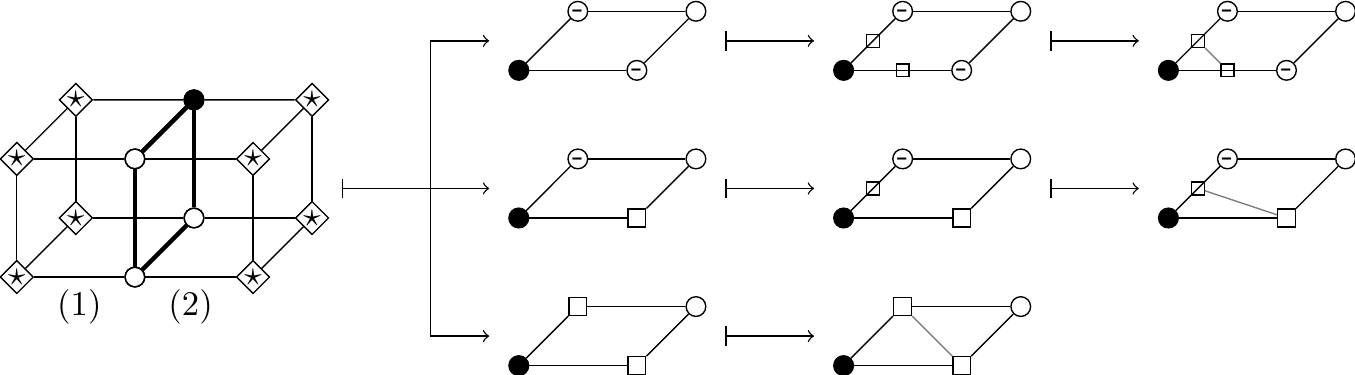}
\end{tabular}
\\ \\

\noindent $B.$
\begin{tabular}[c]{l}
\includegraphics[width=0.8\linewidth]{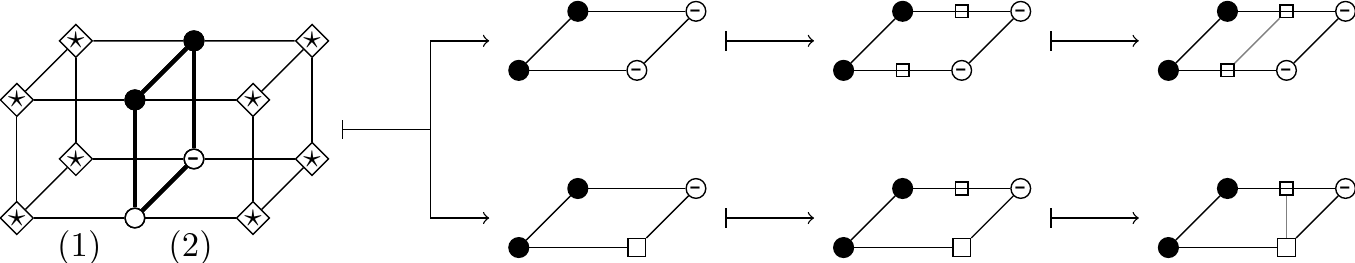}
\end{tabular}
\\ \\

\noindent $C.$
\begin{tabular}[c]{l}
\includegraphics[width=0.8\linewidth]{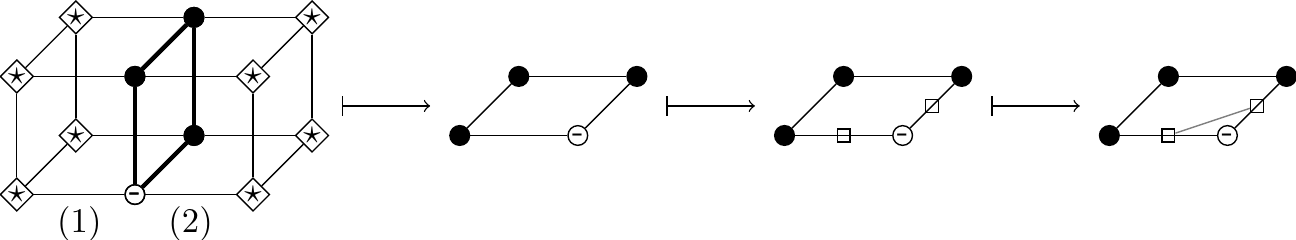}
\end{tabular}
\caption{The sequences $A$, $B$, and $C$ show that for a common regular face of
$G_1\in [(1)]_{\square}^{\star}$ and $G_2\in [(2)]_{\square}^{\star}$ we get an iso-line
on the common face on $G_1$ as well as on $G_2$ as shown in the last sketch of each sequence.}
\label{image_70_71_72}
\end{figure}
\begin{figure}[!ht]
\includegraphics[width=0.7\linewidth]{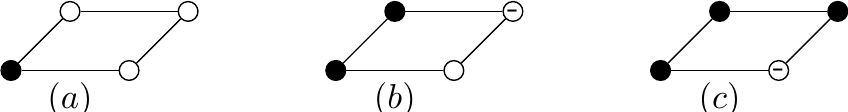}
\caption{Regular faces of one, two and three disperse nodes. Each of the regular faces of
two and three disperse nodes has at least one non-iso-node.}
\label{image_73}
\end{figure}
\FloatBarrier

\begin{lemma}\label{lemma:connect-2}
Let $G(V,E,\mathcal{F})$ be a labeled cuboid graph with iso-level $c\in (0,1)$. Let $G\in [(a)]_{\square}$,
where the sketch $(a)$ is as shown in Figure~\ref{image_74}. Then $G$ contains exactly two distinct
inner iso-paths. Both iso-paths pass through the edge $e$ of the graph $G$.
\begin{figure}[!ht]
\includegraphics[width=0.15\linewidth]{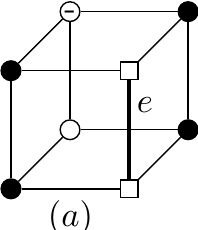}
\caption{Sketch $(a)$ illustrates $G(V,E,\mathcal{F})\in [(a)]_{\square}$. The edge $e$ of $G$ has
two iso-node end points.}
\label{image_74}
\end{figure}
\FloatBarrier
\end{lemma}
\begin{proof}
Apply the $S_2^2$-rule to $G$.
\end{proof}

\begin{proposition}
Let $G(V,E,\mathcal{F})$ be a labeled cuboid graph with iso-level $c\in (0,1)$. Suppose $G$ is regular.
Let $G$ have two iso-points which are iso-nodes that lie on the same edge of a cuboid of $G$. Then the
maximum number of distinct iso-paths that pass through both iso-nodes is two.
\label{prop:connect-1}
\end{proposition}
\begin{proof}
According to Lemma~\ref{lemma:connect-2} any graph in $[(a)]_{\square}$, where sketch $(a)$ is as shown
in Figure~\ref{image_74}, has two distinct iso-paths that pass through both iso-nodes. The maximum
number of inner iso-paths that passes through two iso-points which are iso-nodes and are end points of
an edge of $G$ is attained only if $G\in[(a)]_{\square}$.
\end{proof}
\FloatBarrier

\begin{definition}(System of cuboids and system of labeled cuboid graphs at an edge). Let
$C_1$ be a cuboid with vertices $P_1,\ldots,P_8$. Let $e=\overline{P_1P_2}$ be an edge of $C_1$.
Then let $C_2,C_3,C_4$ be distinct cuboids such that  $C_i$ for $i=1,\ldots,4$ have the edge $e$
in common and if $C_i\cap C_j\nsubseteq e$ for $i\neq j$ then $C_i$ and $C_j$ have a common face.
Then we say $C_1,\ldots,C_4$ are the system of cuboids with common edge $e$. Furthermore, let
$G_i(V_i,E_i,\mathcal{F}_i)$ for $i=1,\ldots,4$ be labeled cuboid graphs with a common iso-level
$c\in (0,1)$ and be singular iso-path free and isolated iso-path free. In addition, let $C_i$ be
the cuboid of $G_i$ for $i=1,\ldots,4$. Then we say $G_1,\ldots,G_4$ are the system of labeled
cuboid graphs with common edge $e$.
\label{def:system-cuboid-1}
\end{definition}

The following lemma is a direct geometrical consequence of Definition~\ref{def:system-cuboid-1}.
\begin{lemma}
Let $C_i$ and $G_i(V_i,E_i,\mathcal{F}_i)$ for $i=1,\ldots,4$ be a system of cuboids and a system
of labeled cuboid graphs, respectively, corresponding to the common edge $e$ such that $C_i$ is a
cuboid of $G_i$. Then to each cuboid $C\in\{C_1,\ldots,C_4\}$ there exist two distinct cuboids
$C_i$ and $C_j$ in $\{C_1,\ldots,C_4\}$ such that $C$ and $C_i$ as well as $C$ and $C_j$ have a
common face. Analogously, to each labeled cuboid graph $G\in\{G_1,\ldots,G_4\}$ there exist two
distinct graphs $G_i$ and $G_j$ in $\{G_1,\ldots,G_4\}$ such that $G$ and $G_i$ as well as $G$
and $G_j$ are face neighbored.
\label{lem:system-cuboid-1}
\end{lemma}

\begin{theorem}
Let $C_i$ and $G_i(V_i,E_i,\mathcal{F}_i)$ for $i=1,\ldots,4$ be a system of cuboids and a system of
labeled cuboid graphs, respectively, corresponding to the common edge $e$ such that $C_i$ is a
cuboid of $G_i$. Let $e=\overline{P_1P_2}$ and assume that the nodes of $G_1$ corresponding to
the end points of $e$ are iso-nodes ($\mathcal{F}_1(P_1)=\mathcal{F}_1(P_2)=c$), where $c\in (0,1)$
is the common iso-level of $G_i$ for $i=1,\ldots,4$. Let the total number of iso-paths
that pass through $e$ be denoted by $N(e)$. Then $N(e)\in\{0,2m\}$, where
$m\leq 4$. If $m \geq 1$, denote the iso-paths by $\omega_1,\ldots,\omega_{2m}$. Additionally,
denote by $G_{\omega_i}$ the labeled cuboid graph corresponding to $\omega_i$ and by $C_{\omega_i}$
the cuboid of $G_{\omega_i}$ for $i=1,\ldots,2m$. If $m \geq 1$ then to any $\omega\in\{\omega_1,
\ldots,\omega_{2m}\}$ there exists $\omega'\in\{\omega_1,\ldots,\omega_{2m}\}$,
$\omega'\neq\omega$, such that one of the following holds:
\begin{itemize}
\item[(i)] $G_{\omega}$ and $G_{\omega'}$ are face neighbored and have two common
disperse nodes,
\item[(ii)] if (i) does not hold then $m\leq 3$ and there exists a unique $l\in\{1,\ldots,4-m\}$ which
depends on $\omega$ and $\omega'$ such that $G_{\omega}$ is face neighbored to $G_{r_1}$ and
$G_{r_l}$ is face neighbored to $G_{\omega'}$, where $\{G_{r_1},G_{r_l}\}\subset\{G_1,
\ldots,G_4\}$. Furthermore, if $l>1$ then $G_{r_i}$ is face neighbored to $G_{r_{i+1}}$
for $i=1,\ldots,l-1$. In addition, each $G_{r_i}$ for $i=1,\ldots,l$ has at least seven disperse
nodes and, hence, they have no iso-path that passes through $e$. Note that for $l=1$ it holds
$G_{r_1}=G_{r_l}$.
\end{itemize}
Then we say the pair $\omega$ and $\omega'$ are disperse connected with respect to the common
iso-line $e$. Furthermore, this property is unique and, hence, there is no other iso-paths which
are disperse connected either to $\omega$ or $\omega'$ with respect to $e$.\\

\noindent Note that all graphs and cuboids stated above are in the system of labeled cuboid
graphs and in the system of cuboids corresponding to the edge $e$, respectively.
\label{thm:connect-3}
\end{theorem}
\begin{proof}
The proof will be given using the following two parts.\\

\noindent{\bf 1.} Uniqueness of disperse connectivity of $\omega$ and $\omega'$ with respect to $e$.\\

If this is not the case, then either $G_{\omega}$ or $G_{\omega'}$ has an additional disperse node
and there exists a graph $G\in\{G_1,\ldots,G_4\}$, where $G\nsubseteq\{G_{r_1},G_{r_l}\}$ such that
$G$ is either face neighbored to $G_{\omega}$ or to $G_{\omega'}$ and $G$ is disperse or $G$ contains
no iso-path that passes through $e$. If $G$ is not disperse but has no iso-path that passes through $e$
then $G$ contains seven disperse nodes (this follows from the application of $T_2$-rule to $G$). But
in this case the face neighbored graph to $G$, which is $G_{\omega}$ or $G_{\omega'}$, has
two additional disperse nodes in common with $G$. Then either $G_{\omega}$ or $G_{\omega'}$ has at
least five disperse nodes. But then application of $T_2$-rule to the graph with at least five disperse
nodes leads to a labeled cuboid graph with seven disperse nodes. Then the graph will have no iso-path that
passes through $e$. But this is a contradiction to the regularity of $G_{\omega}$ and $G_{\omega'}$.
This proves the claim for the unique disperse connectivity of $\omega$ and $\omega'$.\\

\noindent{\bf 2.} Proof validity of cases (i) and (ii).\\

Let us assume that an iso-path $\omega$ exists in $G_{\omega}$ such that the iso-path passes
through $e$. Then $G_{\omega}$ is regular and, hence, $G_{\omega}$ has at least two disperse nodes.
From Lemma~\ref{lem:system-cuboid-1} we know that it is possible to rename the graphs $G_1,\ldots,G_4$
such that $G_i$ is face neighbored to $G_{i+1}$ for $i=1,\ldots,3$ and $G_4$ is face neighbored
to $G_1$. Now the proof follows using the following two steps:\\

\noindent{\bf Step 1.} Since $G_{\omega}$ is regular there exists a face neighbor
$G\in \{G_1,\ldots,G_4\}$ to $G_{\omega}$ such that $G$ and $G_{\omega}$ have two common disperse
nodes corresponding to $e$. If $G$ is regular and has an iso-path that passes through $e$ then case (i)
is satisfied.\\

\noindent{\bf Step 2.} But if $G$ given in Step 1 has no iso-path that passes through $e$ then from the
application of $T_2$-rule to $G$ we get $\tilde{G}$, where $\tilde{G}$ has at least seven disperse nodes.
Then $\tilde{G}$ has no iso-path that passes through $e$. But since all graphs $G_1,\ldots,G_4$ are
singular iso-path free and isolated iso-path free we have $G=\tilde{G}$. Again $G$ is face neighbored
with $G'\in\{G_1,\ldots,G_4\}$ such that $G'\notin\{G,G_{\omega}\}$. Hence, $G'$ and $G$ have at least
two disperse nodes in common. If $G'$ is regular and has an iso-path that passes through $e$ then case
(ii) is satisfied by choosing $G_{r_1}=G$. But if $G'$ has no iso-path that passes through $e$ then
$G'$ has at least seven disperse nodes (again after application of $T_2$-rule to $G'$). Hence, in
case $G'$ has at least seven disperse nodes set $G_{r_2}:=G'$, and $G:=G'$ and then apply again
Step 2 until case (ii) is satisfied.

Note that there exists a
regular graph $G_{\omega'}$, disperse connected to $G_{\omega}$ with respect to edge $e$. If this is
not the case, then each face neighbored graphs $G_p$ and $G_q$ of $G_{\omega}$, where
$\{G_p,G_q\}\subset\{G_1,\ldots,G_4\}$ and $G_p\neq G_q$, have at least seven disperse nodes.
Then $G_q$ and $G_{\omega}$ have two common disperse nodes. But since $G_p$ is face neighbored
to $G_{\omega}$ and not face neighbored to $G_q$, $G_p$ and $G_{\omega}$ have another two common
disperse nodes. But then $G_{\omega}$ has at least four disperse nodes. Furthermore, there exist
four disperse nodes of $G_{\omega}$ denoted by $Q_1,Q_2,Q_3,Q_4$ such that each $P_1,P_2,Q_1,Q_2$
and $P_1,P_2,Q_3,Q_4$ are face vertices of the cuboid of $G_{\omega}$, where $P_1,P_2$ are the end points
of $e$. Furthermore, there exists a graph $G_s$ which is face neighbored with $G_q$ and $G_p$. In
addition,  $G_s$ has at least seven disperse nodes (this follows from the assumption that there exists
no $G_{\omega'}$ which is disperse connected with $G_{\omega}$). If $G_{\omega}$ is regular and has
only four disperse nodes then application of Proposition~\ref{prop:connect-1} gives that $G_{\omega}$
has two iso-paths and both iso-paths passes through $e$ and hence both iso-paths are disperse connected
with respect to $e$; therefore, case (ii) is satisfied. In this case,
we use $G_{r_1}=G_p$, $G_{r_2}=G_s$,  $G_{r_3}=G_q$ in (ii) of Theorem~\ref{thm:connect-3}. But then
in $G_{\omega}$ there exist two iso-paths $\omega$ and $\omega'$ which are disperse connected
with respect to $e$. This is a contradiction to the assumption that there exists no
$\omega'$ which is disperse connected to $\omega$ with respect to $e$. In case $G_{\omega}$ has
five disperse nodes then application of $T_2$-rule to $G_{\omega}$ gives that $G_{\omega}$ has at
least seven disperse nodes and, hence, $G_{\omega}$ has no iso-path that passes through $e$. This is
a contradiction to the assumption that $G_{\omega}$ has an iso-path that passes through
$e$, hence this case does not occur. Consequently, case (ii) holds.\\

To conclude, the set of all iso-paths that pass through $e$ can be uniquely decomposed into pairs
such that each pair is disperse connected with respect to $e$. Hence, the total number of iso-paths in
the system of graphs which pass through $e$ is an even number. The maximum number is 8 as follows
from Proposition~\ref{prop:connect-1}).
\end{proof}

\begin{theorem}\label{thm:connect-4}(Connectedness of iso-paths).
Let the polygonal domain $\Omega\subset\mathbb{R}^3$ have a domain partition $\mathcal{T}=\{C_i\}_{i=1}^N$
into cuboids. Let $\mathcal{G}=\{G_1,\ldots,G_N\}$ be a set of labeled
cuboid graphs with common iso-level $c\in (0,1)$, and $C_i$ be the cuboid of $G_i$ for $i=1,\ldots,N$.
Assume that all $G_i$, $i=1,\ldots,N$, singular iso-path free and isolated iso-path free.
Compute the complete iso-paths in each $G_i$ by applying the algorithm given in Section 5.4.
Then each iso-line of an arbitrary $G\in \mathcal{G}$ is common for at least two distinct iso-paths, where
the iso-paths can be in $G$ or in face- or in edge-neighbors of $G$.
\end{theorem}
\begin{proof} We prove the claim by distinguishing different cases:\\

\noindent {\bf Case 1:} In case an iso-line $l$ is an edge of $G\in \mathcal{G}$ then
Theorem~\ref{thm:connect-3} says that we have at least two iso-paths in $\mathcal{G}$ which have $l$
as a common iso-line.\\

\noindent {\bf Case 2:} In case an iso-line $l$ lies on a trivial L-face of $G\in \mathcal{G}$,
Theorem~\ref{thm:connect-1} says that we have two iso-paths in $G$ which have $l$ as a common iso-line.
We even get four iso-paths with $l$ as a common iso-line if the trivial L-face is common for $G$ and
$G'$, where $G'\in \mathcal{G}$ is a face neighbor of $G$.\\

\noindent {\bf Case 3:} Let $H\subset G$ be a regular face of $G$. Assume that $H$ is as well a face
of $G'\in \mathcal{G}$, where $G'$ is a face neighbor of $G$. Then Proposition~\ref{prop:class-7}
implies that each of the graphs $G$ and $G'$ contains exactly one inner iso-path running through $l$.\\

\noindent {\bf Case 4:} Let an iso-line $l$ lie on a non-trivial L-face $H(V_h,E_h,\mathcal{F}_h)$ of
$G\in \mathcal{G}$. Then we consider two subcases:
\begin{itemize}
\item[(a)] Suppose $H$ is a common non-trivial L-face of $G$ and $G'(V',E',\mathcal{F}')\in\mathcal{G}$. Then
one of the following holds:
\begin{itemize}
\item[(i)] if there exists an iso-path on each of the L-faces then each iso-line on the L-face is part of an
iso-path that does not lie on the L-face as shown in Proposition~\ref{prop:class-8}. Therefore, to each iso-line
there exist two iso-paths, where the first iso-path is an inner iso-path and the second iso-path is the
iso-path on the L-face.
\item[(ii)] if there exists no iso-path on each of the L-faces then Proposition~\ref{prop:class-8} says that
to each iso-line on the L-face there exists an iso-path in $G$ and in the face neighboring graph. Hence to
each iso-line on the L-face there exist two iso-paths.
\end{itemize}
\item[(b)] Suppose $G'(V',E',\mathcal{F}')\in \mathcal{G}$ is a face neighbor of $G$, where
$H'(V'_h,E'_h,\mathcal{F}'_h)$ is a regular face of $G'$ and $V_h=V'_h$. Then one of the following
holds:
\begin{itemize}
\item[(i)] if there exists an iso-path on the L-face then we have the same conclusion as in (i) of (a).
\item[(ii)] if there exists no iso-path on the L-face then we have the same conclusion as in (ii) of (a).
\end{itemize}
\end{itemize}
\vspace{-0.7cm}
\end{proof}

\section{Components of iso-surfaces}
In this section we show how to compute separate components of connected iso-surfaces such that
on each component normals and discrete mean curvature can be calculated. We give definitions required to define
iso-path connectivity such that an iso-surface can be decomposed into its components. These components
are orientable and connected.

We have proved in Section 6 using Theorem~\ref{thm:connect-4} that the iso-surfaces computed by
applying the algorithm given in Section 5.4 are connected. This means to each iso-line $l$ of
an iso-path there exist at least two iso-paths such that $l$ is a common edge to them.
Theorem~\ref{thm:connect-3} and Theorem~\ref{thm:connect-1} show that an iso-line $l$ can be
common for up to eight or four iso-paths, respectively. It is clear that discrete mean curvature
computation at the end points of $l$, where $l$ is common to more than two iso-paths, is not defined.
Hence, we will give a definition which allows to decompose the iso-surfaces into connected components
such that at each iso-point of the connected components discrete mean curvature computation is possible.

To define components of closed iso-surfaces we need the notion of a {\it disperse path} which is a simple
path but not a loop such that it consists solely of edges of cuboids which only have disperse nodes as
end points. Such a path runs within the {\it system of graphs}, which is defined next.

\begin{definition}(System of cuboids and system of graphs).  Let $C$ be a cuboid with vertices
$P_1,\ldots,P_8$ and edges $e_1,\ldots,e_{12}$. Then let $\mathcal{C}_1,\ldots,\mathcal{C}_{12}$
be the system of cuboids (see Definition \ref{def:system-cuboid-1}) with respect to edges
$e_1,\ldots,e_{12}$, respectively. For each $i=1,\ldots,12$ and $j=1,\ldots,4$ we denote by
$C_{ij}$ the cuboids such that
\[
\mathcal{C}_i=\{C_{ij}\,:\,j=1,\ldots,4\} \mbox{ for } i=1,\ldots,12,
\]
where $C=C_{i1}$ for all $i=1,\ldots,12$, and the following holds:
\begin{itemize}
\item for all $x\in \{P_1,\ldots,P_8\}$ there exists a neighborhood $U\subset\mathbb{R}^3$ of $x$
      such that $U\subset \cup_{i=1}^{12}\cup_{j=1}^{4}C_{ij}$ and any two distinct cuboids
      in the set $\{C_{ij}\,:\,i=1,\ldots,12\; \mbox{ and }\; j=1,\ldots,4\}$ have a common edge
      or a common face.
\end{itemize}
Then we say $C=C_{i1}$ and $C_{ij}$ for $i=1,\ldots,12$, $j=2,\ldots,4$ are the system of cuboids
of $C$. Furthermore, let $G_{i1}(V_{i1},E_{i1},\mathcal{F}_{i1})$ and $G_{ij}(V_{ij},E_{ij},
\mathcal{F}_{ij})$ for $i=1,\ldots,12$, $j=2,\ldots,4$ be labeled cuboid graphs with a common
iso-level $c\in (0,1)$ and singular iso-path free and isolated iso-path free. In addition, let
$C_{i1}$ and $C_{ij}$ be the cuboids of $G_{i1}$ and $G_{ij}$ for $i=1,\ldots,12$, $j=2,\ldots,4$,
respectively. Then we say that $G_{i1}$ and $G_{ij}$ for $i=1,\ldots,12$, $j=2,\ldots,4$ are the
system of labeled cuboid graphs of $G(V,E,\mathcal{F})$, where $G=G_{i1}$ for all $i=1,\ldots,12$.
Additionally, we call $N(G)$ the number of system of graphs of $G$ which is given by
\begin{equation}
N(G)=\sum_{i=1}^{12}n_i-\sum_{l=1}^6(|E_{F_l}|-1)-(|E|-1),
\label{eq:system-graph-1}
\end{equation}
where $n_i=|\mathcal{C}_i|=4$ for $i=1,\ldots,12$, $F=\{F_1,\ldots,F_6\}$ is the set of faces of
$C$, $E_{F_l}$ is the set of edges of face $F_l$ of $C$ for $l=1,\ldots,6$, and $|.|$ denotes
the number of elements in a set. It holds $|E_{F_l}|=4$ for $l=1,\ldots,6$ and $|E|=12$. Hence,
in the present case of cuboids, $N(G)=19$.
\label{def:curve-1}
\end{definition}
The derivation of formula~\eqref{eq:system-graph-1} is easy and is therefore left to the reader.
Equation~\eqref{eq:system-graph-1} can be even applied to arbitrary partition
of a polygonal domain which has the form as given in Definition~\ref{def:curve-1}.

\begin{definition}(Disperse path). Let $\zeta$ be a simple, but not closed path in the system of graphs
of a given labeled cuboid graph $G(V,E,\mathcal{F})$. If there is $r\geq 2$ and to each $m=1,\ldots,r$
there exists $G'_m(V'_m,E'_m,\mathcal{F}'_m)$ in the system of graphs of $G$ and a disperse edge $l_m$
of $G_m'$ such that $\zeta$ is of the form $\zeta=\cup_{m=1}^rl_m$, we call $\zeta$ a disperse path in
the system of graphs of $G$.
\label{def:disperse-path}
\end{definition}

\begin{remark}
In the following, when we say "corresponding nodes of an iso-line" or
"corresponding iso-line of nodes", the word "corresponding" is to be understood in the
sense of Notation~\ref{note:corresponding-node}. Furthermore, when we say "disperse
node or nodes corresponding to iso-line $l$ with respect to an iso-path", then it means that
$l$ is an edge of the iso-path and $l$ corresponds to the disperse node or nodes.
\label{rem:corresponding-node-line}
\end{remark}

\begin{definition}(Disperse connectedness of two iso-paths at a common edge). Let $G(V,E,\mathcal{F})$ be
a labeled cuboid graph with iso-level $c\in (0,1)$. Let $G$ be regular and $l$ be an iso-line of $G$.
Let $\omega_1$ and $\omega_2$ be two distinct iso-paths with $l=\omega_1\cap\omega_2$, where $\omega_1$
is one of the iso-paths of $G$ and $\omega_2$ is an iso-path in the system of graphs of $G$. We say that
$\omega_1$ and $\omega_2$ are disperse connected with respect to the common edge $l$ if one of the
following holds:
\begin{enumerate}
\item $\omega_1$ and $\omega_2$ are the only iso-paths in the system of graphs of $G$ such that
$l=\omega_1\cap\omega_2$,
\item $\omega_1$ is an inner iso-path of $G$, $\omega_2$ is an inner iso-path for one labeled cuboid
graph in the system of graphs of $G$, and one of the following holds:
\begin{enumerate}
\item at least one of the disperse nodes corresponding to $l$ is the same for the iso-line $l$ with
respect to $\omega_1$ and $\omega_2$,
\item there exist two distinct disperse nodes $P_1$ and $P_2$ in the system of graphs of $G$ such that
$P_1$ is a node of $G$ and $P_2$ may not a node of $G$. The disperse nodes $P_1$ and $P_2$ correspond
to $l$ with respect to $\omega_1$ and $\omega_2$, respectively, and there exists a disperse path in the
system of graphs of $G$ which connects $P_1$ and $P_2$.
\end{enumerate}
\end{enumerate}
We say that two distinct iso-elements $Z_1$ and $Z_2$ in the system of
graphs of $G$ are {\it neighbored} with respect to the iso-line $l=Z_1\cap Z_2$ of $G$ if the corresponding
iso-paths $\omega_1$ and $\omega_2$ of $Z_1$ and $Z_2$, respectively, are disperse connected with respect
to $l$.
\label{def:component-1}
\end{definition}

\noindent{\bf Convention of iso-elements disjointness: }Let $G(V,E,\mathcal{F})$ be  a labeled cuboid graph
with iso-level $c\in (0,1)$. Let $G$ be regular and $l$ be an iso-line of $G$. Let $Z_1$ and $Z_2$ be two
distinct iso-elements in the system of graphs of $G$ which have $l$ as a common iso-line ($l=Z_1\cap Z_2$),
but are {\it not} neighbored with respect to $l$. Then, concerning their connectivity, we consider both
iso-elements $Z_1$ and $Z_2$ as disjoint. This means, given arbitrary points $P_1\in Z_1$ and $P_2\in Z_2$,
there exists no path $\omega\subset Z_1\cup Z_2$ that joins $P_1$ and $P_2$. This convention helps to
decompose iso-surfaces into connected components such that on each component computation of surface PDEs
and discrete mean curvature is possible.\\

\noindent The next theorem will be used to decompose iso-surfaces into components.
\begin{theorem}Let $G(V,E,\mathcal{F})$ be a labeled cuboid graph with iso-level $c\in (0,1)$ and let
$G$ be regular. Let $\omega$ be an iso-path of $G$, given by $\omega=[Q_1,Q_2,\ldots,Q_n]$ with
$n\geq 3$. We set $l_i=\overline{Q_iQ_{i+1}}$ for $i=1,\ldots,n-1$ and $l_n=\overline{Q_nQ_1}$ which
are iso-lines of $G$ and as well edges of $\omega$. Then there exist iso-paths $\omega_1,\ldots,\omega_n$ in
the system of labeled cuboid graphs of $G$ such that $l_i=\omega\cap \omega_i \mbox{ for }i=1,\ldots,n$ and
each pair $(\omega,\omega_i)$ is disperse connected with respect to $l_i$ for $i=1,\ldots,n$.
\label{thm:iso-path-connectedness-disperse-paths}
\end{theorem}
\begin{proof}
We prove the claim by considering three cases. In all these cases, we let $G'(V',E',\mathcal{F}')$ be a
graph in the system of graphs of $G(V,E,\mathcal{F})$ such that $G$ and $G'$ are face-neighbored and
the nodes of the face $H(V_h,E_h,\mathcal{F}_h)$ of $G$ are common to both graphs $G$ and $G'$.
Suppose $H'(V_h',E_h',\mathcal{F}_h')\subset G'$ is the face of $G'$ such that $V_h'=V_h$.\\

\noindent{\bf Case 1: }An iso-line $l\in\{l_1,\ldots,l_n\}$ lies on a regular face $H(V_h,E_h,\mathcal{F}_h)$
of $G$. \\

We consider three subcases.
\begin{enumerate}
\item Assume that $H$ contains only two disperse nodes and at least one of the continuous nodes is not
an iso-node. Then the following holds:
\begin{enumerate}
\item $H=H'$.
\item An iso-line on $H$ is common only for two iso-paths. The first iso-path lies in $G$ and the second in
$G'$. This follows from Proposition~\ref{prop:class-7}. The common iso-line of both iso-paths corresponds to
the disperse nodes of $H$ (cf. Remark~\ref{rem:corresponding-node-line}).
\end{enumerate}
Hence, condition $1$ of Definition~\ref{def:component-1} is satisfied.

\item $H$ has two disperse nodes and two iso-nodes. Then the following holds:
\begin{enumerate}
\item[(a)] If $l$ is common for $G$ and $G'$ and if there exists an iso-path $\omega'$ in $G'$ that passes
through $l$ such that the corresponding disperse nodes of $\omega$ and $\omega'$ with respect to $l$ are
the same, then both iso-paths are disperse connected.

\item[(b)] The other cases are shown in Theorem~\ref{thm:connect-3}.
\end{enumerate}

Hence, condition $1$ or condition $2$ of Definition~\ref{def:component-1} are satisfied.
\item The face $H'$ of $G'$ is a non-trivial L-face. Then one of the following holds:
\begin{enumerate}
\item Let $l$ be the only iso-line on $H'$. Then the following holds:
\begin{enumerate}
\item $H'$ contains exactly one iso-node and $H$ contains three disperse nodes and one continuous node
which is not an iso-node.
\item The iso-line $l$ is common to a pair of iso-paths such that the iso-paths are disperse
connected with respect to the common iso-line.
\end{enumerate}
These results follow from Proposition~\ref{prop:class-7} and \ref{prop:class-8}. Hence, condition $1$ of
Definition~\ref{def:component-1} is satisfied.
\item Let there be an iso-path on $H'$. Then the following holds:
\begin{enumerate}
\item There exist three iso-lines on $H'$.
\item There is one iso-line on $H$.
\item To each iso-line in $H'$ there exists an iso-path that does not lie on $H'$.
\item Each iso-line on $H'$ is common to the iso-path on $H'$ and another iso-path that does not lie on
$H'$. Both iso-paths are disperse connected with respect to their common iso-line $l$.
\end{enumerate}
These results follow from Proposition~\ref{prop:class-7} and \ref{prop:class-8}. Hence, condition $1$ of
Definition~\ref{def:component-1} is satisfied.
\end{enumerate}
\end{enumerate}
\noindent{\bf Case 2: }An iso-line $l\in\{l_1,\ldots,l_n\}$ lies on a trivial L-face
$H(V_h,E_h,\mathcal{F}_h)$ of $G$. \\

We consider two subcases.
\begin{enumerate}
\item Let $H'$ be a singular or a disperse face of $G'$.

Then, according to Theorem~\ref{thm:connect-1} there are two iso-paths in $G$ having $l$ as a common
iso-line. These iso-paths are disperse connected with respect to $l$.

Hence, condition $1$ of Definition~\ref{def:component-1} is satisfied.
\item Let $H'$ be a trivial L-face of $G'$. Then the following holds:
\begin{enumerate}
\item $H=H'$.
\item According to Theorem~\ref{thm:connect-1} we get two iso-paths in $G$ and two iso-paths in $G'$ and
all four iso-paths have the common iso-line $l$. Then there exists a unique pairing of iso-paths such that
the common iso-line $l$ of each pair of iso-paths corresponds to a disperse node of $H$. Both pairs of
iso-paths are then disperse connected with respect to $l$. But it is not possible that three of them
are disperse connected. If this is the case, then there is a disperse path which connects the disperse
nodes on $H$. But for this to happen we need at least five disperse nodes, say in $G$. Then application of
the $T_1$-rule to $G$ will change one of the iso-nodes of $G$ on $H$ to a disperse node. But then
$G$ has no iso-path that passes through $l$. Therefore, there cannot exist three iso-paths that pass
through $l$ and which are disperse connected.
\end{enumerate}
Hence, condition $2$ of Definition~\ref{def:component-1} is satisfied.
\end{enumerate}
\noindent{\bf Case 3: }An iso-line $l\in\{l_1,\ldots,l_n\}$ lies on a non-trivial L-face
$H(V_h,E_h,\mathcal{F}_h)$ of $G$.
\begin{enumerate}
\item Let there be no iso-path on the L-faces $H$ and $H'$. Then the following holds:
\begin{enumerate}
\item $H'=H$.
\item There exist two iso-lines on $H$.
\item Each iso-line from (b) is common to a pair of iso-paths which are disperse connected with respect to
the common iso-line on $H$. This follows from Proposition~\ref{prop:class-8}.
\end{enumerate}
Hence, condition $1$ of Definition~\ref{def:component-1} is satisfied.
\item Let there be an iso-path on the L-faces. Then the following holds:
\begin{enumerate}
\item $H'=H$.
\item There exist three iso-lines on $H$ if $H$ has an iso-node; otherwise, there exist four iso-lines on $H$.
\item To each iso-line from (b) there exists an iso-path that does not lie on $H$.
\item Each iso-line from (b) is common to the iso-path on $H$ and another iso-path which does not lie on
$H$. Both iso-paths are disperse connected with respect to the common iso-line $l$.
\end{enumerate}
These results follow from Proposition~\ref{prop:class-7} and \ref{prop:class-8}. Hence, condition $1$ of
Definition~\ref{def:component-1} is satisfied.
\end{enumerate}
\end{proof}

\noindent Based on Theorem~\ref{thm:iso-path-connectedness-disperse-paths} we give the following definition.
\begin{definition}(Disperse connected iso-elements and iso-surface component).
Let the polygonal domain $\Omega\subset\mathbb{R}^3$ have a domain partition $\mathcal{T}=\{C_i\}_{i=1}^N$ into cuboids. Let $\mathcal{G}=\{G_1,\ldots,G_N\}$ be a set of labeled
cuboid graphs with common iso-level $c\in (0,1)$, and $C_i$ be the cuboid of $G_i$ for $i=1,\ldots,N$.
Assume that all $G_i$ are singular iso-path free and isolated iso-path free.
Compute the complete iso-paths in each $G_i$ by applying the algorithm given in Section 5.4.
Let the union of all iso-surfaces in $\mathcal{G}$ be denoted by $\Gamma$ and let
$Z$ and $\tilde{Z}$ be two iso-elements in $\Gamma$. If there is $r\geq 1$ and
if there exist iso-elements $Z_0,\ldots,Z_r$ of $\Gamma$ such that
$Z_0=Z$, $Z_r=\tilde{Z}$ and all pairs $(Z_{i-1},Z_i)$, $i=1,\ldots,r$ are neighbored with respect
to their common iso-line, then we say $Z$ and $\tilde{Z}$ are disperse connected iso-elements.

We define a component of $\Gamma$ with respect to an iso-element $Z\subset\Gamma$, denoted by $\Gamma(Z)$, by
\begin{itemize}
\item $\Gamma(Z)$ consists of $Z$ and of all iso-elements of $\Gamma$ which are disperse connected to $Z$.
\end{itemize}
\label{def:component-2}
\end{definition}

\begin{theorem}(Connectedness and uniqueness of a component).
An iso-surface component defined as in Definition~\ref{def:component-2} is connected. Furthermore,
for any two distinct iso-elements $Z_1$ and $Z_2$ of $\Gamma$ it holds that either
$\Gamma(Z_1)=\Gamma(Z_2)$ or $\Gamma(Z_1)\cap\Gamma(Z_2)=\emptyset$ in the sense of
{\it convention of iso-elements disjointness}. Consequently, the components of $\Gamma$ are uniquely
determined.
\label{thm:iso-element-connected}
\end{theorem}
\begin{proof}We give the proof in two steps. \\

\noindent{\bf 1. Uniqueness of components: }Note that the relation between iso-elements to be disperse
connected is reflexive, symmetric and transitive by its definition. Hence, the set of all iso-elements
decomposes into disjoint (in the sense of iso-elements disjointness) equivalence classes, which are
precisely the components of $\Gamma$. This proves the uniqueness.\\

\noindent{\bf 2. Connectedness of a component: }Let $\Gamma(Z)$ be a component of $\Gamma$ and let
$Z_1,Z_2\in\Gamma(Z)$ be iso-elements. Then $Z_1$ is connected to $Z$ and $Z$ is connected to $Z_2$,
hence each two iso-elements in $\Gamma(Z)$ are disperse connected.
\end{proof}

The intersection of two different components is $\emptyset$ in the sense that the intersection does not
contain an iso-element, but it may contain discrete points or line segments (cf. {\it Convention of
iso-elements disjointness}). The following remark provides additional information on the intersection
of two distinct components of $\Gamma$.

\begin{remark}(Relations of two distinct components).
We denote from here on the components of $\Gamma$ by $\Gamma_1,\ldots,\Gamma_n$ if $\Gamma$ has $n\geq 2$
components, where $\Gamma$ and $\Gamma_i$ are as in Definition~\ref{def:component-2}. Note that the only
possibility for two distinct components to have an iso-element in common is if both have a common trivial L-face
such that on the L-face there is an iso-element. But such an L-face belongs to only one component, since if there
exists an iso-path on the L-face then each iso-line on the L-face is common to only two distinct iso-paths as
shown in Proposition~\ref{prop:class-8}. The first iso-path is the iso-path on the L-face and the second iso-path
is an inner iso-path. Therefore, according to Definition~\ref{def:component-2}, the iso-path on the L-face belongs
to only one component and, hence, all other inner iso-paths which have a common iso-line with the iso-path on the
L-face belong to the same component as well. This follows from condition $1$ of Definition~\ref{def:component-1}.

If two different components have an iso-point or iso-points in common, where none of them are end points of
an iso-line for both components, then the common points are vertices of cuboids in the partition $\mathcal{T}$.
In addition, if two different components have an iso-line or iso-lines in common, then the common line segments
are edges (cf. Theorem~\ref{thm:connect-3}) or face diagonals (cf. Lemma~\ref{lemma:connect-1}) of cuboids in the
partition $\mathcal{T}$. This follows from Definition~\ref{def:component-2} and
Theorem~\ref{thm:iso-element-connected}.
\end{remark}

\section{Surface Normal Vectors and discrete Curvature}
For iso-surfaces $\Gamma$ computed according to the algorithm given in Section 5.4 there exists
a decomposition of $\Gamma$ into components according to Definition~\ref{def:component-2} and
Theorem~\ref{thm:iso-element-connected}. By construction, a component is oriented and connected.
This section is devoted to solve the following problems:
\begin{enumerate}
\item how to compute surface normal vectors of a component,
\item how many iso-points of a component are connected to an iso-point $P$ of the component,
\item constructing an appropriate region in one component around an iso-point $P$ of the component
on which discrete mean curvature computation for $P$ is possible.
\end{enumerate}
In this paper we are not giving details on how to compute discrete mean curvature at discrete points of
a component since this can be found in the literature (see e.g. \cite{Meyer02Vismath}), but we
show how to identify the required surface region inside a component.

\subsection{Surface Normal Vectors}
The computation of surface normal vectors of an iso-surface is straight forward, except for
determining the same orientation for all iso-surface normal vectors. For all triangles of the triangulated
iso-elements, let $\pmb{N}$ be an arbitrarily oriented surface normal. Then the oriented normal field
$\pmb{n}$ is obtained by choosing $+\pmb{N}$ or $-\pmb{N}$, locally. For this purpose we first
need to know a direction in which the desired surface normal shall be approximately pointing. We call
a vector which approximates the surface normal vector, while giving the same orientation (i.e. angle to
$\pmb{n}$ is below $90^{\circ}$), a {\it surface pseudo-normal}. Usually, for the computation of such a
surface pseudo-normal one needs to know the gradient of the nodal function $\mathcal{F}$ of a labeled
cuboid graph $G(V,E,\mathcal{F})$. We give here an alternative method which does not use the gradient
of the function $\mathcal{F}$ but instead uses all node positions of $G$, distinguishing between continuous
and disperse nodes.

Let $G(V,E,\mathcal{F})$ be an irreducible labeled cuboid graph with iso-level $c\in (0,1)$. Let there be
$n\geq 3$ iso-points $P_i$ of the iso-path, denoted by $\omega$, in $G$ and let $P_c\in\mathbb{R}^3$ be the
center of the iso-element $[P_1,\ldots,P_n|P_c]$ corresponding $\omega$. By definition of an iso-surface
component as given in Definition~\ref{def:component-2}, the orientation of the normal field on an
iso-element of $G$ should be pointing towards the continuous nodes of $G$. Therefore, a surface pseudo-normal
for the iso-element, denoted as $\pmb{p}$, has to be determined using the position of
disperse and continuous nodes of $G$ such that it points from the disperse to the continuous phase. Let
Let $\pmb{N}$ be a normal of the triangle $Tri(P_1,P_2,P_c)$, say. Then the surface normal $\pmb{n}$ on
$Tri(P_1,P_2,P_c)$ is
\begin{equation}
\pmb{n}=\left\{\begin{tabular}{ll}
$\pmb{N}$ &\mbox{if } $\pmb{N}\cdot \pmb{p}>0$\\
-$\pmb{N}$ &\mbox{else.}
\end{tabular}
\right.
\label{eq:curve-1}
\end{equation}

\begin{remark} For iso-surface normal computation, the consideration of irreducible labeled cuboid graphs,
is sufficient. First, computation of surface normals of iso-paths which lie on L-faces of a labeled cuboid graph
is straight forward. Second, a labeled cuboid graph $G$ can be decomposed into irreducible graphs with respect
to inner iso-paths of $G$ as given in~\eqref{eq:class-3-1}.
\label{rem:second-remark}
\end{remark}

\noindent{\bf Surface pseudo-normal: }Let $G(V,E,\mathcal{F})$ be an irreducible labeled cuboid graph
with iso-level $c\in (0,1)$. Let $V=\{P_1,\ldots,P_8\}, N=\{1,2,\ldots,8\}$ and let $N_c$ and $N_d$
correspond to the set of continuous and disperse node indices of $G$, respectively.
Then $N=N_c\cup N_d$. We compute a surface pseudo-normal for $G$ by adding all vectors
$\pmb{v}_j$ for $j = 1,\ldots,|N_d|$ with $\pmb{v}_j$ defined by
\begin{equation}
\pmb{v}_j=\sum_{i\in N_c}(C_i-D_j), \qquad  j\in N_d,
\label{eq:curve-2}
\end{equation}
where $C_i$ and $D_j$ are continuous and disperse nodes of $G$, respectively.
Then we get
\begin{equation}
\pmb{p}=\sum_{j\in N_d}\sum_{i\in N_c}(C_i-D_j)
\label{eq:curve-3}
\end{equation}
\begin{proposition}The surface pseudo-normal vector given in \eqref{eq:curve-3} satisfies
\begin{equation*}
\pmb{p}=|N_d|\sum_{i\in N}P_i-8\sum_{j\in N_d}D_j,
\end{equation*}
where $V=\{P_1,\ldots,P_8\}$.
\label{prop:curve-2}
\end{proposition}
\begin{proof}
We get the following by using $N=N_c\cup N_d$ and $|N|=8$ :
\begin{eqnarray*}
\sum_{i\in N_c}(C_i-D_j)&=&\sum_{i\in N_c}C_i-|N_c|D_j\\
&=&\sum_{i\in N_c}C_i+\sum_{k\in N_d}D_k-|N_c|D_j-\sum_{k\in N_d}D_k\\
&=&\sum_{i\in N}P_i-|N_c|D_j-\sum_{k\in N_d}D_k.
\end{eqnarray*}
By using the above relation we get
\begin{eqnarray*}
\sum_{j\in N_d}\sum_{i\in N_c}(C_i-D_j)&=&\sum_{j\in N_d}
\left[\sum_{i\in N}P_i-|N_c|D_j-\sum_{k\in N_d}D_k\right]\\
&=&|N_d|\sum_{i\in N}P_i-|N_c|\sum_{j\in N_d}D_j-|N_d|\sum_{k\in N_d}D_k\\
&=&|N_d|\sum_{i\in N}P_i-(|N_c|+|N_d|)\sum_{j\in N_d}D_j\\
&=&|N_d|\sum_{i\in N}P_i-8\sum_{j\in N_d}D_j.
\end{eqnarray*}
\end{proof}

\subsection{Discrete Curvature}
An iso-surface $\Gamma$ which is computed according to the algorithm given in Section 5.4 is polygonal
(piecewise planar). Hence, on a component $\Gamma_0$ of $\Gamma$ only discrete mean curvature computation
has a meaning. The discrete mean curvature can be computed on points in $\Gamma_0$ which are vertices of the
triangulation of $\Gamma_0$. Any iso-element of $\Gamma$ is triangulated using a center $P_c$
of the iso-element and its iso-points. Generally, the discrete mean curvature at $P_c$ is nearly zero. Hence,
discrete mean curvature computation is mainly important at the iso-points of $\Gamma_0$.

Discrete mean curvature computation methods (see e.g. \cite{Meyer02Vismath}) use for the computation of
discrete mean curvature at an iso-point $P$ of $\Gamma_0$ a piece of triangulated surface region around $P$ in
which $P$ is contained. This surface region is contained in $\Gamma_0$. Hence, the aim of this section is to
compute such a surface region.

The following theorem will give us the minimum and maximum number of neighbors of an iso-point $P$ of
$\Gamma_0$ such that the neighbors are iso-points in $\Gamma_0$ and each of them is incident to $P$.

\begin{theorem}Let $G(V,E,\mathcal{F})$ be a labeled cuboid graph with iso-level $c\in (0,1)$ and let
$G$ be regular. Let $Z$ be an iso-element of $G$ and let the iso-path $\omega$ corresponding to $Z$ be
given by $\omega=[Q_1,Q_2,\ldots,Q_m]$ with $m\geq 3$. We set $r_i=\overline{Q_iQ_{i+1}}$ for
$i=1,\ldots,m-1$ and $r_m=\overline{Q_mQ_1}$, which are iso-lines and edges of $\omega$. Let us fix
one of the $Q_i$, say $P:=Q_1$. Then there exist $4\leq n\leq 8$ iso-paths $\omega_1,\ldots,\omega_n$
in the system of labeled cuboid graphs of $G$, where $\omega_1:=\omega$ such that
\begin{enumerate}
\item $l_i=\omega_i\cap \omega_{i+1} \mbox{ for }i=1,\ldots,n-1$,
\item $l_n=\omega_n\cap \omega_1$,
\end{enumerate}
with $l_1=r_1, l_n=r_m$ and
\begin{itemize}
\item[(i)] $\omega_i$ and $\omega_{i+1}$ are disperse connected with respect to $l_i$ for $i=1,\ldots,n-1$,
\item[(ii)] $\omega_n$ and $\omega_1$ are disperse connected with respect to $l_n$.
\end{itemize}
This means, if $\,\Gamma$ is the iso-surface computed according to the algorithm given in Section 5.4 then all iso-paths
$\omega_1,\ldots,\omega_n$ belong to the same component $\Gamma(Z)$ of $\Gamma$ and all these iso-paths have in
common the iso-point $P$ and these iso-paths are the only iso-paths in $\Gamma(Z)$ with this property.
\label{thm:iso-path-region-curvature}
\end{theorem}
\begin{proof}
We prove the claim by considering four cases. \\

\noindent{\bf Case 1: }All faces in the system of graphs of $G$, where the iso-point $P$ lies, are regular.\\

\noindent In this case the point $P$ is an end point of two iso-lines of $\omega_1$. Both iso-lines
lie in $G$ but on different regular faces. Note that $P$ does not lie on an L-face and hence both
regular faces do not each have two disperse and two iso-nodes (according to Lemma~\ref{lemma:connect-2}).
Hence, on each of these regular faces lies an edge of a unique iso-path according to
Proposition~\ref{prop:class-7}. Therefore, to each neighbor of the faces there exists an iso-path
which passes through an iso-line on the face. Hence, we have two additional iso-paths in two different
labeled cuboids. From these two additional iso-paths we get two additional iso-lines with an end point
$P$. Therefore, we get a total of four distinct iso-lines with an end point $P$. But only three
iso-paths cannot give a connected iso-surface at the point $P$ and hence, there exist at least one
additional iso-path that passes through $P$. Hence, there exist $\omega_2,\ldots,\omega_4$ satisfying
the claim of the proposition which is $n=4$.\\

\noindent{\bf Case 2: }The iso-point $P$ is not an iso-node of $G$ and all faces in the system of
graphs of $G$, where the iso-point $P$ lies, are L-faces.\\

\noindent The iso-point $P$ is common to four faces in the system of graphs of $G$. If each of the
L-faces is non-trivial, then on each face we have two iso-lines with end point $P$. Then we have a
total of $n=8$  iso-lines which are neighbors of $P$. If some of the L-faces are trivial L-faces,
we get $4\leq n\leq 8$.\\

\noindent{\bf Case 3: }The iso-point $P$ is an iso-node of $G$ and all faces in the system of graphs
of $G$, where the iso-point $P$ lies, are L-faces.\\

\noindent The iso-point $P$ is common to eight faces in the system of graphs of $G$. The prove uses
the same argument as in Case 1 to get that $n$ is at least four. In case each of the L-faces is non-trivial,
we get the same result as in Case 2. In this case, the same argument as in Case 2 can be applied.\\

\noindent{\bf Case 4: }The iso-point $P$ of $G$ lies on an L-face in the system of graphs of $G$.\\

\noindent By combining the cases 1, 2 and 3 we then get $4\leq n\leq 8$.
\end{proof}

In the next definition, we give the so-called {\it surface region} $\Gamma_P$ of an iso-point $P$ within the
component $\Gamma_0$ of $\Gamma$. The surface region $\Gamma_P$ contains $P$ and all iso-points in $\Gamma_0$
which are incident to $P$. Additionally, $\Gamma_P$ contains each center of the iso-elements which have $P$ in
common. The surface region $\Gamma_P$ of $P$ can be used to compute discrete mean curvature at $P$
(see e.g. \cite{Meyer02Vismath}).

\begin{definition}(Neighboring iso-lines, iso-points, points and surface region). Let $Z$ be an
iso-element of the iso-surface $\Gamma$ computed according to the algorithm given in Section 5.4.
and let $P\in Z$ an iso-point. Then we call the iso-lines $l_1,\ldots,l_n$ which we get from
Theorem~\ref{thm:iso-path-region-curvature} using $Z$, neighboring iso-lines to $P$. Then
$P$ and $l_1,\ldots,l_n$ lie in the component $\Gamma(Z)$. Let the iso-points $P_1,\ldots,P_n$ be
the other end points of the $l_i$, i.e. $l_i=\overline{PP_i}$ for $i=1,\ldots,n$.
We call $P_1,\ldots,P_n$ neighboring iso-points to $P$. Note that to each iso-line $l_i$ there exists an
iso-path $\omega_i$ in $\Gamma(Z)$ such that $l_i\subset \omega_i$. All the iso-paths
$\omega_i$ are different according to Theorem~\ref{thm:iso-path-region-curvature}.
Let $P_{c_i}$ be the center of $\omega_i$ and let us denote by $N(\omega_i)$ the number of iso-lines
of $\omega_i$. We define
\begin{equation*}
\mathcal{E}_i=\left\{\begin{tabular}{ll}
$Tri(P,P_i,P_{i+1})$\;\mbox{ if }\;$N(\omega_i)=3$\mbox{ or }\;$\omega_i$ \mbox{ is an outer iso-path}&\mbox{}\\
$Tri(P,P_i,P_{c_i})\cup Tri(P,P_{i+1},P_{c_{i+1}})$\;\mbox{ else}&\mbox{}
\end{tabular}\right.
\end{equation*}
for $i=1,\ldots,n-1$, and
\begin{equation*}
\mathcal{E}_n=\left\{\begin{tabular}{ll}
$Tri(P,P_n,P_1)$\;\mbox{ if }\;$N(\omega_n)=3$\mbox{ or }\;$\omega_n$ \mbox{ is an outer iso-path}&\mbox{}\\
$Tri(P,P_n,P_{c_n})\cup Tri(P,P_1,P_{c_1})$\;\mbox{ else}.&\mbox{}
\end{tabular}\right.
\end{equation*}
We then call the piece of iso-surface $\Gamma_P$ defined by
\[
\Gamma_P:=\cup_{i=1}^n\mathcal{E}_i
\]
the surface region of $P$. In addition, we define a set of points
$\mathcal{P}_i$ by
\begin{equation*}
\mathcal{P}_i=\left\{\begin{tabular}{ll}
$\{P_i,P_{i+1}\}$\;\mbox{ if }\;$N(\omega_i)=3$\mbox{ or }\;$\omega_i$ \mbox{ is an outer iso-path}&\mbox{}\\
$\{P_i,P_{c_i},P_{i+1},P_{c_{i+1}}\}$\;\mbox{ else}&\mbox{}
\end{tabular}\right.
\end{equation*}
for $i=1,\ldots,n-1$, and
\begin{equation*}
\mathcal{P}_n=\left\{\begin{tabular}{ll}
$\{P_n,P_1\}$\;\mbox{ if }\;$N(\omega_n)=3$\mbox{ or }\;$\omega_n$ \mbox{ is an outer iso-path}&\mbox{}\\
$\{P_n,P_{c_n},P_1,P_{c_1}\}$\;\mbox{ else}.&\mbox{}
\end{tabular}\right.
\end{equation*}
Finally, we set $\mathcal{D}_P:=\cup_{i=1}^n\mathcal{P}_i$.
\label{def:neighboring-points-1}
\end{definition}

\noindent{\bf Computation of discrete mean curvature: }Let $\Gamma_P$ and $P$ be as defined in
Definition~\ref{def:neighboring-points-1}. The discrete mean curvature at the iso-point $P$ is computed by
integrating $\Delta_{\Gamma}\pmb{x}$ over $\Gamma_P$, where $\Delta_{\Gamma}$ denotes the
Laplace-Beltrami operator. The surface $\Gamma_P$ is piecewise linear. Therefore, we can give
nodal functions defined on $\Gamma_P$ with values $1$ on one of the points $\{P\}\cup \mathcal{D}_P$
and zero on the other points, where $\mathcal{D}_P$ is computed according to
Definition~\ref{def:neighboring-points-1}. Applying these nodal functions as a basis for $\pmb{x}$
and using Gauss's theorem for surface integration of $\Delta_{\Gamma}\pmb{x}$ over $\Gamma_P$ we get
the discrete mean curvature at the iso-point $P$. For more details on this computation of discrete mean curvature
see \cite{Meyer02Vismath}.

\section*{Acknowledgements}
The authors gratefully acknowledge financial support by the German Research Foundation within the
Priority Program "Transport Processes at Fluidic Interfaces" (SPP 1506).
\newpage
\bibliographystyle{amsplain}
\bibliography{references}
\end{document}